\documentclass[twocolumn,notitlepage,floatfix,reprint,superscriptaddress]{revtex4-2}

\usepackage[lmargin=.7in,rmargin=.7in,tmargin=.6in,bmargin=0.8in]{geometry}
\usepackage{graphicx}
\usepackage{amsmath}
\usepackage{amsthm}
\usepackage{amssymb}
\usepackage{latexsym}
\usepackage{array}
\usepackage{hyperref}
\usepackage{amsfonts}
\usepackage{dsfont}
\usepackage{mathrsfs}
\usepackage{verbatim}
\usepackage{bbold}
\usepackage[normalem]{ulem}
\usepackage{upgreek}
\usepackage{makecell}
\usepackage{adjustbox}
\usepackage{algorithm}
\usepackage{algpseudocode}
\usepackage{color}
\usepackage{bm}
\usepackage{times}
\usepackage{booktabs}
\usepackage{multirow}

\newcommand{\bra}[1]{\left<#1\right|}
\newcommand{\ket}[1]{\left|#1\right>}
\newcommand{\abs}[1]{\left|#1\right|}

\newcommand{\ketbra}[2]{\ket{#1}\!\!\bra{#2}}
\newcommand{\E}{\mathbb{E}}

\theoremstyle{plain}
\newtheorem{theorem}{Theorem}
\newtheorem{proposition}{Proposition}

\newtheorem{corollary}{Corollary}
\newtheorem{definition}{Definition}
\newtheorem{remark}{Remark}
\newtheorem{result}{Result}

\DeclareMathOperator{\tr}{Tr}
\DeclareMathOperator{\Var}{Var}

\makeatletter
\renewcommand\part[1]{%
  \clearpage
  \onecolumngrid
  \section*{#1}
  
}
\makeatother

\makeatletter
\let\origaddcontentsline\addcontentsline
\renewcommand{\addcontentsline}[3]{}
\makeatother

\begin{document}

\part{}

\title{Causal Fisher-Information Inequalities: Classical Causal Model Falsification and Metrological Advantage}

\author{Jeongho~Bang}\email{jbang@yonsei.ac.kr}
\affiliation{Institute for Convergence Research and Education in Advanced Technology, Yonsei University, Seoul 03722, Republic of Korea}
\affiliation{Department of Quantum Information, Yonsei University, Incheon 21983, Republic of Korea}

\author{Su-Yong~Lee}\email{suyong2@add.re.kr}
\affiliation{Agency for Defense Development, Daejeon 34186, Republic of Korea}
\affiliation{Weapon Systems Engineering, ADD school, University of Science and Technology, Daejeon, 34060, Republic of Korea}

\date{March 20, 2026}

\begin{abstract}
Fisher-information inequalities have recently been used as operational witnesses of nonclassical metrological behavior, but their physical meaning is often tied to a particular narrative, such as, segmented dynamics or discrete trajectories. We show that a broader interpretation is available and, in fact, more natural: once an experiment is assumed to admit a classical causal model specified by a directed acyclic graph, conditional independences, and modular parameter dependence, the corresponding Fisher informations are forced to obey causal Fisher-information inequalities (CFIIs). The backbone result is a causal-path series law: for an additive causal parameter that propagates through a classical path $A\to C\to B$, the inverse Fisher information behaves as an information resistance and must add in series. Consequently, any CFII violation is a rigorous falsification of the entire classical causal model class. We then show that the violation is automatically a metrological resource certificate, because it implies a precision that no member of the classical causal class can attain. The gain mechanism is identified as Fisher-information synergy, i.e. off-diagonal score correlations that classical modularity forbids. A single-qubit coherent-rotation example demonstrates the deterministic CFII violation, estimator-level achievability of the resulting gain, robustness against split-optimized classical benchmarks, and a chain-amplified advantage in long causal decompositions. Finally, an AI-assisted adversarial finite-data stress test shows that the witness remains certifiable under the realistic visibility loss and readout error, while optimized modular classical causal models saturate but do not cross the CFII frontier.
\end{abstract}

\maketitle


Precision metrology is governed by how rapidly outcome probabilities deform under a parameter change. In the operational language of classical Fisher information (FI), this deformation directly sets the asymptotic Cram\'er--Rao limit~\cite{Frieden2004,Cramer1946,Helstrom1976}. Quantum metrology traditionally frames advantage relative to the quantum Fisher information or to explicit classical benchmarks~\cite{BraunsteinCaves1994,Paris2009,Demkowicz2020,Liu2020}. An experimentally observed FI advantage therefore has a precise meaning only after the benchmark model has been specified: it shows that the data beat the precision allowed by that model, not that they beat every conceivable classical description. In particular, many practically relevant FI inequalities are valid only when the process is assumed to decompose into identifiable segments, contexts, or intermediate mediators with independent local parameter dependences.

The decomposition requirements just described are therefore not merely technical assumptions about an estimator or a convenient dynamical narrative. They assert, explicitly or implicitly, which variables mediate the process, which conditional independences hold among those variables, and which local mechanisms are allowed to carry dependence on the unknown parameter. We make these ``causal'' contents explicit. A classical causal explanation is specified by a Bayesian network (possibly with latent variables), whose factorization induces conditional-independence relations among the relevant variables~\cite{Pearl2009}. In a metrological problem, compatibility with such a model means more than reproducing a single marginal distribution: the whole parameterized family must arise from graph-compatible local kernels whose allowed $\theta$-dependence is modular. Consequently, the score functions, i.e. the derivatives of the log-likelihoods with respect to $\theta$, and the corresponding Fisher informations must respect the graph-induced modularity and the assigned local parameter dependence. This converts FI inequalities into causal-model criteria. The proper interpretation of a violation is sharp: the observed statistics cannot be generated by any member of the assumed classical causal model class.

This perspective yields a unified program. First, derive FI inequalities from a causal model class rather than an ad hoc dynamical narrative. Second, use violations to falsify the corresponding classical causal explanations. Third, reinterpret the same violations as resource certificates: the information that exceeds the causal-model frontier is precisely the information that can lower the estimator variance relative to that frontier.

We note that because the main text is intentionally compact, see also the supplementary information (SI) that develops the above-described program in full generality. The main text--to--SI correspondence map is provided in Methods.

\section*{RESULTS}

\subsection*{Causal Fisher-information inequalities as causal-model criteria}

We consider parameterized experiments specified by conditional distributions $p(x|s,\theta)$, where $\theta$ is an unknown parameter and $s$ labels an experimental context such as a measurement setting, intermediate intervention, control choice, or segmentation strategy. Throughout, a context is simply an experimentally specified way of producing a record $x$ from which information about $\theta$ may be extracted. For a regular discrete model, the local score and FI are
\begin{eqnarray}
\ell_s(x|\theta)=\partial_\theta\ln p(x|s,\theta),\qquad
F_s(\theta)=\E[\ell_s(X|\theta)^2],
\label{eq:scoreFI_basic}
\end{eqnarray}
with the analogous integral expression for continuous records. If data are gathered independently across contexts, the total FI is additive across contexts, so the effective information is determined by both the single-context statistics and the sample allocation. This operational layer is completely general and already accommodates quantum experiments through the Born rule $p(x | s,\theta)=\tr{\hat{M}_{x|s}\hat{\rho}(\theta)}$, where $\hat{M}_{x|s}$ is the positive operator-valued measure element associated with outcome $x$ in context $s$, and $\hat{\rho}(\theta)$ is the parameter-dependent quantum state before measurement.

A classical causal model class $\mathcal{M}$ is then introduced by a directed acyclic graph (DAG), possibly with latent variables, together with a specification of how the parameter enters the local modules. Here ``modular'' means that, after the graph is fixed, each local conditional kernel has its own allowed $\theta$-dependence and this dependence is not shared arbitrarily across other kernels. Compatibility with $\mathcal{M}$ means that the observed statistics arise as marginals of a joint model that factorizes according to the DAG and respects this prescribed modular parameterization. The crucial observation is that Fisher information is a functional of the probability model; thus, the graph-induced conditional independences and modular factorization translate into model-dependent constraints on FI.

\begin{theorem}[CFII principle]
Let $\mathcal{M}$ be a classical causal model class defined by a DAG, its implied conditional independences, and regularity assumptions ensuring that parameter dependence factorizes consistently with the graph. Then, there exists a family of inequalities among the Fisher informations associated with the relevant contexts or parameter segments, and every operational model compatible with $\mathcal{M}$ must satisfy them. We call these constraints causal Fisher-information inequalities (CFIIs). Any observed violation falsifies $\mathcal{M}$.
\end{theorem}

\begin{proof}[Proof sketch]---The DAG factorization constrains the joint score structure. The modules that are conditionally independent contribute scores that are orthogonal in expectation, while the hidden or discarded variables appear only through the parameter-independent stochastic post-processing. The first property bounds how Fisher information can compose across modules; the second prevents discarded variables from increasing endpoint information. Any inequality obtained from these two ingredients is therefore a necessary condition for membership in the model class.
\end{proof}

\noindent Formal statement and extended discussion: Supplementary Theorem~S1 and Supplementary Notes~II.D--II.E.

Theorem 1 reframes the familiar FI inequalities as instances of a broader causal logic. Rather than asking whether a process really follows a trajectory, we specify the variables, conditional independences, and modular parameter dependences that define a trajectory-like classical causal class, and then test if the observed family of distributions belongs to that classical causal model. The algebraic FI bound may look identical to a traditional series inequality, but its antecedent is now explicit: if the measured FIs fall outside the allowed region, then no model with that DAG, those independences, and that modular parameterization can explain the data.

\subsection*{Causal-path CFII: inverse FI as information resistance}

The archetypal case is a classical causal path $A \to C \to B$, where $A$ is an input or preparation label, $C$ is an intermediate classical mediator, and $B$ is the endpoint record. The path hypothesis asserts both conditional independence and parameter modularity:
\begin{eqnarray}
 p(c,b | a,\theta_{ac}, \theta_{cb})=p(c | a,\theta_{ac}) \, p(b | c,\theta_{cb}),
\label{eq:pathfactor}
\end{eqnarray}
with an additive total parameter
\begin{eqnarray}
\theta_{ab} = \theta_{ac} + \theta_{cb}.
\label{eq:additive}
\end{eqnarray}
This covers the natural situation in which the parameter corresponds to accumulated time, phase, action, or another additive generator. The modularity statement in Eq.~\eqref{eq:pathfactor} is important: $\theta_{ac}$ appears only in the upstream kernel, $\theta_{cb}$ appears only in the downstream kernel, and the endpoint record is obtained through the mediator rather than through an additional direct parameter-dependent channel.

\begin{theorem}[Causal-path CFII]
Under the path factorization in Eq.~\eqref{eq:pathfactor} and the additive causal parameter in Eq.~\eqref{eq:additive}, the effective Fisher information for estimating $\theta_{ab}$ from the endpoint record $B$ alone obeys
\begin{eqnarray}
\left(F_{ab}^{(B)}\right)^{-1}\geq F_{ac}^{-1}+F_{cb}^{-1}.
\label{eq:pathCFII}
\end{eqnarray}
Here, $F_{ac}$ is the FI of the upstream module, $F_{cb}$ is the downstream FI averaged over the upstream distribution, and $F_{ab}^{(B)}$ is the effective FI for the sum parameter extracted from the endpoint data. All three quantities are local FIs evaluated at the chosen parameter values; the displayed inverse form assumes positive FI, with singular cases understood by the usual limiting convention.
\end{theorem}

\begin{proof}[Proof sketch]---First, imagine that the intermediate value $C$ is accessible together with the endpoint record. Because the logarithm of the joint likelihood splits into an upstream term and a downstream term, the two score components are uncorrelated, so the Fisher information matrix of $(C,B)$ is diagonal in the module parameters $(\theta_{ac},\theta_{cb})$. Applying the multi-parameter Cram\'er--Rao bound to the sum direction $u=(1,1)^T$ yields a harmonic bound: the information about $\theta_{ab}$ in the joint record is at most $(F_{ac}^{-1}+F_{cb}^{-1})^{-1}$. Second, the actual endpoint record $B$ is obtained from $(C,B)$ by forgetting $C$, which is a parameter-independent coarse-graining. The Fisher information cannot increase under such causal data processing, so the same harmonic upper bound applies to $B$ alone.
\end{proof}

\noindent Full proof: Supplementary Theorem~S5. The $K$-step extension is Supplementary Corollary~S2.

Eq.~\eqref{eq:pathCFII} is the core series law of this work. If we define an ``information resistance'' as
\begin{eqnarray}
R:=F^{-1}, 
\label{eq:info_resistance}
\end{eqnarray}
then a classical causal path behaves like passive resistors in series:
\begin{eqnarray}
R_{ab}^{(B)}\geq R_{ac}+R_{cb}.
\label{eq:resistance}
\end{eqnarray}
The same logic extends immediately to a $K$-step chain $X_0 \to X_1 \to \cdots \to X_K$ with modular segment parameters $\theta_{j-1,j}$ and additive total $\theta_{0K}=\sum_{j=1}^{K}\theta_{j-1,j}$, giving
\begin{eqnarray}
\left(F_{0K}^{(X_K)}\right)^{-1} \geq \sum_{j=1}^{K}F_{j-1,j}^{-1}.
\label{eq:chainCFII}
\end{eqnarray}
Thus, the inverse Fisher information accumulates across every classical bottleneck in the chain.

\subsection*{Violation means impossibility of the classical causal model}

A CFII is not merely a useful bound; it is a logical constraint defining the feasible Fisher-information region of a causal model class. This makes falsification immediate.

\begin{theorem}[Violation implies causal-model impossibility]
Fix a classical causal model class $\mathcal{M}$ and a CFII that holds for all operational models compatible with $\mathcal{M}$. If an experiment produces Fisher informations that violate the inequality at some parameter value, then the observed statistics are incompatible with $\mathcal{M}$.
\end{theorem}

\begin{proof}[Proof sketch]---The proof is the contrapositive of the definition. A CFII is required to hold for every member of the model class. Therefore, a point outside the CFII-allowed region cannot come from any model in the class.
\end{proof}

\noindent Full proof: Supplementary Theorem~S7. For the causal-path witness, see also Supplementary Theorem~S8.

For the path hypothesis, it is convenient to package the test into the witness
\begin{eqnarray}
V_{\mathrm{path}}:=\left(F_{ab}^{(B)}\right)^{-1}-F_{ac}^{-1}-F_{cb}^{-1}.
\label{eq:Vpath}
\end{eqnarray}
Every classical causal-path model satisfies $V_{\mathrm{path}}\geq 0$, so a negative value means that the endpoint precision is better than any series composition of two independent classical modules would permit. Importantly, the conclusion is specific: one has not proved ``quantumness in general,'' but rather ruled out the entire conjunction of the assumptions that defined the classical causal-path class---i.e., the existence of the mediator, path factorization, and modular parameter dependence.

The falsification program remains statistically meaningful with finite data. Suppose that in each context $s$, one can evaluate the local score $\ell_s(x|\theta)=\partial_\theta\ln p(x|s,\theta)$. The empirical average of the squared score,
\begin{eqnarray}
\widetilde F_s(\theta)=\frac{1}{N_s}\sum_{i=1}^{N_s}\ell_s(x_{s,i}|\theta)^2,
\label{eq:pluginFI}
\end{eqnarray}
is unbiased for $F_s(\theta)=\E[\ell_s(X|\theta)^2]$. The required moment condition is the concrete one $\E[\ell_s(X|\theta)^4]<\infty$, which implies
\begin{eqnarray*}
\sqrt{N_s}\,[\widetilde F_s(\theta)-F_s(\theta)]\Rightarrow
\mathcal N\bigl(0,\Var[\ell_s(X|\theta)^2]\bigr).
\end{eqnarray*}
Thus the finite-shot uncertainty is controlled by the variance of the squared score. By combining this context-wise central limit theorem with the delta method, one obtains asymptotic confidence intervals for any smooth CFII witness.

Let $\mathbf F(\theta)=(F_s(\theta))_s$ denote the vector of true context-wise Fisher informations entering a chosen CFII witness. The index $s$ runs over the endpoint and all segment or context experiments used by the test. We write
\begin{eqnarray}
G(\theta)=g[\mathbf F(\theta)], \quad \widetilde G(\theta)=g[\widetilde{\mathbf F}(\theta)],
\label{eq:Gdef_main}
\end{eqnarray}
where $g$ is the scalar function that builds the relevant witness from those FIs. For example, the two-step path witness uses $g(F_{ab}^{(B)},F_{ac},F_{cb})=(F_{ab}^{(B)})^{-1}-F_{ac}^{-1}-F_{cb}^{-1}$, and the $K$-step chain witness uses the corresponding endpoint resistance minus the sum of the $K$ segment resistances. The plug-in estimate $\widetilde G$ is obtained by replacing every entry of $\mathbf F$ by the estimate in Eq.~\eqref{eq:pluginFI}. Our sign convention is $G(\theta)<0$ for a CFII violation, so a negative plug-in value is interpreted only after its standard error, computed from the gradient $\nabla g$, has been taken into account.

\begin{proposition}[Finite-data certification]
If each score has finite fourth moment, then the smooth plug-in CFII witness $\widetilde G(\theta)$ is asymptotically normal around its true value $G(\theta)$. Hence, a negative estimate separated from zero by several standard errors rejects the corresponding classical causal model class with controlled significance.
\end{proposition}

\begin{proof}[Proof sketch]---Each $\widetilde F_s$ is an average of i.i.d. squared scores, so the central limit theorem applies context by context. Independence across contexts yields a multivariate normal limit for the vector of estimated Fisher informations. A first-order Taylor expansion of the witness then gives asymptotic normality and an explicit variance formula through the gradient of $G$.
\end{proof}

\noindent Full proof: Supplementary Theorems~S9--S10.

\subsection*{From witness to resource: violation-certified gain and Fisher synergy}

Because FI controls the Cram\'er--Rao precision bound, every CFII immediately induces a classical precision frontier. For the path model, the corresponding benchmark FI is the harmonic composition
\begin{eqnarray}
F_{\mathrm{cl}}^{(\mathrm{path})}:=\left(F_{ac}^{-1}+F_{cb}^{-1}\right)^{-1}.
\label{eq:Fcl}
\end{eqnarray}
The path CFII is precisely the statement $F_{ab}^{(B)}\leq F_{\mathrm{cl}}^{(\mathrm{path})}$. Equivalently, with $R_{ab}^{(B)}=(F_{ab}^{(B)})^{-1}$ and $R_{\mathrm{cl}}=(F_{\mathrm{cl}}^{(\mathrm{path})})^{-1}=F_{ac}^{-1}+F_{cb}^{-1}$, the witness is $V_{\mathrm{path}}=R_{ab}^{(B)}-R_{\mathrm{cl}}$.

\begin{theorem}[Violation-certified metrological gain]
If $V_{\mathrm{path}}<0$, then
\begin{eqnarray}
F_{ab}^{(B)}>F_{\mathrm{cl}}^{(\mathrm{path})},
\label{eq:gainFI}
\end{eqnarray}
and therefore the observed model has a tighter asymptotic Cram\'er--Rao lower bound for regular unbiased estimation of $\theta_{ab}$ than the lower bound imposed by the classical causal-path class:
\begin{eqnarray}
\frac{1}{N F_{ab}^{(B)}}<\frac{1}{N F_{\mathrm{cl}}^{(\mathrm{path})}}.
\label{eq:gainCRB}
\end{eqnarray}
Under the standard efficiency conditions, this tighter bound can be attained asymptotically, for example by maximum-likelihood estimation. Equivalently, the ratio of information resistances satisfies
\begin{eqnarray}
\frac{R_{\mathrm{cl}}}{R_{ab}^{(B)}}=\frac{R_{\mathrm{cl}}}{R_{\mathrm{cl}}+V_{\mathrm{path}}}>1.
\label{eq:improvementfactor}
\end{eqnarray}
\end{theorem}

\begin{proof}[Proof sketch]---The witness is exactly the difference between the actual resistance and the classical benchmark resistance. If the witness is negative, the actual resistance is smaller; because both resistances are positive, taking reciprocals yields a strict FI gain and thus a strict reduction in the asymptotic error lower bound.
\end{proof}

\noindent Full proof: Supplementary Theorem~S12. The estimator-level achievability is developed in Supplementary Theorem~S13.

This witness-to-resource transition is not merely formal. Under the usual regularity assumptions, the maximum-likelihood estimation is asymptotically efficient, so the improved FI can be harvested by an explicit estimator; see Supplementary Theorem~S13. The more revealing question is how the gain appears. The answer is \emph{Fisher-information synergy}.

Consider a two-parameter description $\theta=(\theta_1,\theta_2)^T$ of a process whose Fisher matrix with respect to an observation $Y$ is
\begin{eqnarray}
F_Y(\theta)=
\begin{pmatrix}
F_1 & J \\
J & F_2
\end{pmatrix},
\label{eq:FIM2x2}
\end{eqnarray}
where $J$ is the covariance between the two score contributions associated with the two local parameter directions. In a classical modular causal decomposition, one has $J=0$ by construction because the centered module scores are orthogonal. When the score contributions are correlated, the effective FI for the additive parameter $\Theta=\theta_1+\theta_2$ becomes
\begin{eqnarray}
F_Y^{(u)}=\frac{F_1F_2-J^2}{F_1+F_2-2J}, \quad u=(1,1)^T.
\label{eq:synergy}
\end{eqnarray}

\begin{theorem}[Synergy formula]
For the Fisher matrix in Eq.~\eqref{eq:FIM2x2}, assume $F_1>0$, $F_2>0$, and $F_Y$ is positive definite. The effective FI for $\Theta=\theta_1+\theta_2$ satisfies
\begin{eqnarray}
F_Y^{(u)} > \left(F_1^{-1} + F_2^{-1}\right)^{-1}
\label{eq:beatsseries}
\end{eqnarray}
if and only if
\begin{eqnarray}
0<J<\frac{2F_1F_2}{F_1+F_2}.
\label{eq:Jwindow}
\end{eqnarray}
Moreover, at fixed $F_1$ and $F_2$, the largest achievable effective FI is $\min\{F_1,F_2\}$.
\end{theorem}

\begin{proof}[Proof sketch]---Invert the $2 \times 2$ Fisher matrix explicitly and evaluate the quadratic form $\left(u^TF_Y^{-1}u\right)^{-1}$. Comparing the resulting resistance to the classical harmonic resistance yields the exact window in Eq.~\eqref{eq:Jwindow}. The optimization over $J$ is elementary because the effective resistance is a rational function of a single variable.
\end{proof}

\noindent Full proof: Supplementary Theorem~S14.

Equivalently, the information resistance associated with the sum direction is the inverse effective FI,
\begin{eqnarray}
R_Y^{(u)} := \bigl[ F_Y^{(u)} \bigr]^{-1} = u^T F_Y^{-1}u = \frac{F_1+F_2-2J}{F_1F_2-J^2}.
\end{eqnarray}
The modular classical value, obtained when $J=0$, is $R_{\rm ser}=F_1^{-1}+F_2^{-1}$. A direct comparison gives
\begin{eqnarray}
R_Y^{(u)}-R_{\rm ser} = \frac{J \left[ (F_1+F_2)J - 2 F_1 F_2 \right]}{F_1 F_2 (F_1 F_2 - J^2)},
\end{eqnarray}
so $R_Y^{(u)} < R_{\rm ser}$ exactly in the window of Eq.~\eqref{eq:Jwindow}.

Theorem 5 identifies the physical mechanism behind CFII violation. In a modular classical path, the upstream and downstream likelihood factors contain separate parameter dependences; after centering, their score contributions are orthogonal, so $J=0$ and the harmonic series penalty follows. A violation requires a different score geometry: the same observed record must make the two would-be segment scores fluctuate together, so that an outcome that is informative about one segment is also informative about the other. Thus, in the window of Eq.~\eqref{eq:Jwindow}, the positive covariance $J$ makes the effective information resistance $R_Y^{(u)}=\bigl[ F_Y^{(u)} \bigr]^{-1}$ for the additive parameter smaller than the modular classical series resistance $R_{\rm ser}=F_1^{-1}+F_2^{-1}$.

The same mechanism explains the long-chain scaling. If the Fisher matrix of a $K$-segment description takes the equicorrelated form $F[(1-\varepsilon)I_K+\varepsilon J_K]$, where $I_K$ is the identity matrix, $J_K$ is the $K\times K$ all-ones matrix, and $\varepsilon$ is the common normalized off-diagonal score correlation, then the effective FI for the total parameter is
\begin{eqnarray}
F^{(u)} = F\left(\varepsilon + \frac{1-\varepsilon}{K}\right),
\label{eq:equicorrelated}
\end{eqnarray}
so any fixed $\varepsilon>0$ changes the large-$K$ scaling from the classical $F/K$ law to the nonzero limit $F\varepsilon$. The corresponding design rule is concrete: avoid experimental operations that reset, reveal, or classically separate the intermediate modules, and instead preserve coherent or otherwise shared features of the record that make the putative segment scores positively correlated.

\subsection*{Coherent single-qubit dynamics beyond the classical causal-path frontier}

The general theory is realized by the simplest coherent qubit dynamics. Consider a single qubit evolving under
\begin{eqnarray}
\hat{H} = \frac{1}{2}\hat{\sigma}_x, \quad \hat{U}(\theta)=e^{-i\theta\hat{\sigma}_x/2},
\label{eq:qubitU}
\end{eqnarray}
prepared in
\begin{eqnarray}
\ket{\psi(\vartheta,\phi)} = \cos\frac{\vartheta}{2}\ket{0} + e^{i\phi}\sin\frac{\vartheta}{2}\ket{1},
\label{eq:qubitstate}
\end{eqnarray}
and measured with the projective readout $\{\ketbra{0}{0}, \ketbra{1}{1}\}$. Writing the binary statistics as $p_0(\theta)=\frac{1+z(\theta)}{2}$ and $p_1(\theta)=\frac{1-z(\theta)}{2}$, one finds
\begin{eqnarray}
z(\theta)=\cos\vartheta\cos\theta+\sin\vartheta\sin\phi\sin\theta,
\label{eq:ztheta}
\end{eqnarray}
and therefore,
\begin{eqnarray}
F(\theta)=\frac{(\partial_\theta z)^2}{1-z^2}.
\label{eq:binaryFI}
\end{eqnarray}

\begin{result}[Deterministic single-qubit collapse]
At the coherent point $\phi=\pi/2$, the readout FI is constant for all $\theta$:
\begin{eqnarray}
F(\theta)=1.
\label{eq:Fone}
\end{eqnarray}
Consequently, for every nontrivial split $\theta_{ab}=\theta_{ac}+\theta_{cb}$, one has
\begin{eqnarray}
V(\theta_{ac},\theta_{cb}) = -1,
\label{eq:Vminus1}
\end{eqnarray}
so the classical causal-path benchmark is $F_{\mathrm{cl}}^{(\mathrm{path})}=1/2$ and the achieved FI is larger by an exact factor of two.
\end{result}

\begin{proof}[Proof sketch]---When $\phi=\pi/2$, Eq.~\eqref{eq:ztheta} reduces to $z(\theta)=\cos(\theta-\vartheta)$, so Eq.~\eqref{eq:binaryFI} gives $F(\theta)=1$ identically. Substituting $F(\theta_{ab})=F(\theta_{ac})=F(\theta_{cb})=1$ into the witness yields $V=1-1-1=-1$.
\end{proof}

\noindent Full derivation: Supplementary Result~S1. The estimator-level achievability is illustrated in Supplementary Fig.~2.

This point is useful because the causal meaning is fully transparent: no entanglement, no adaptive measurements, and no measurement optimization are required. A fixed qubit fringe already outperforms every classical path model with a single intermediate bottleneck. Moreover, the endpoint model is a simple binary fringe, so the maximum-likelihood estimation asymptotically reaches the $1/\sqrt{N}$ standard deviation implied by the actual FI, whereas the classical causal-path frontier predicts a minimum error larger by $\sqrt{2}$.

For generic $(\vartheta,\phi)$, the witness landscape contains broad regions where the path CFII is violated and the gain indicator is negative, exactly as shown in Supplementary Fig.~1. By Supplementary Theorem~S15, the equivalence between the negative witness and the negative log-error ratio means that the same landscape simultaneously maps model falsification and precision enhancement.

A conservative critic might try to rescue the classical path model by optimizing the location of the intermediate split. This motivates the split-optimized classical benchmark
\begin{eqnarray}
F_{\mathrm{cl}}^{(\mathrm{opt})}(\theta_{ab})=\max_{0<\theta_{ac}<\theta_{ab}}\left( \frac{1}{F(\theta_{ac})} + \frac{1}{F(\theta_{ab}-\theta_{ac})}\right)^{-1}
\label{eq:optbenchmark}
\end{eqnarray}
The maximization makes the benchmark conservative: the classical path is allowed to choose the most favorable intermediate split after the total parameter interval has been fixed. We use the dimensionless gain ratio $\Gamma$ to mean the achieved endpoint FI divided by the relevant classical benchmark FI. For the split-optimized comparison, this is
\begin{eqnarray}
\Gamma(\theta_{ab})=\frac{F(\theta_{ab})}{F_{\mathrm{cl}}^{(\mathrm{opt})}(\theta_{ab})}.
\label{eq:GammaOpt_main}
\end{eqnarray}

\begin{result}[Split-robust collapse]
If $\Gamma(\theta_{ab})>1$, then the causal-path CFII is violated for every admissible split of the same total parameter. No choice of intermediate time can restore a classical causal-path explanation at that $\theta_{ab}$.
\end{result}

\begin{proof}[Proof sketch]---For each fixed split, the corresponding classical benchmark is bounded above by the optimized benchmark in Eq.~\eqref{eq:optbenchmark}. Therefore, beating the optimized benchmark automatically implies beating every split-dependent benchmark.
\end{proof}

\noindent Full proof and adversarial-benchmark analysis: Supplementary Result~S2 and Supplementary Fig.~3.

The optimized comparison is physically important because it gives the classical model its best possible chance. The generic qubit landscapes in the SI show wide regions where $\Gamma>1$ even after this adversarial optimization; see Supplementary Fig.~3. Thus, the collapse is not a split artifact; it is a genuine incompatibility between the coherent dynamics and the classical path modularization.

The same logic strengthens dramatically when the classical narrative is refined into a long chain. If one insists on a $K$-step causal decomposition of the same total evolution, the classical benchmark becomes the harmonic series composition of all segments.

\begin{result}[Chain-amplified gain]
Whenever the chosen readout yields constant FI $F(\theta)=F_0$, every $K$-step classical causal-path decomposition satisfies
\begin{eqnarray}
V_K=-\frac{K-1}{F_0}<0,
\label{eq:VK}
\end{eqnarray}
the classical benchmark is
\begin{eqnarray}
F_{\mathrm{cl}}^{(K)}=\frac{F_0}{K},
\label{eq:FclK}
\end{eqnarray}
and the gain ratio is exactly
\begin{eqnarray}
\Gamma_K=\frac{F(\theta_{ab})}{F_{\mathrm{cl}}^{(K)}}=K.
\label{eq:GammaK}
\end{eqnarray}
Hence, the standard deviation beats the $K$-step classical frontier by a factor $\sqrt{K}$.
\end{result}

\begin{proof}[Proof sketch]---With constant FI, every segment contributes the same resistance $1/F_0$, so the classical chain adds them to $K/F_0$. The endpoint experiment, however, still has resistance $1/F_0$. Direct substitution gives Eqs.~\eqref{eq:VK}--\eqref{eq:GammaK}.
\end{proof}

\noindent Full proof: Supplementary Result~S3. The general chain inequality is Supplementary Corollary~S2.

Result 3 makes the resource interpretation especially vivid. Increasing $K$ does not add laboratory resources; it only tightens the classical causal claim by inserting more putative passive mediators. A coherent process that keeps $O(1)$ endpoint FI therefore becomes increasingly incompatible with the classical series law as the chain is refined.

\subsection*{AI-assisted adversarial finite-data stress test}

\begin{figure*}[t]
\centering
\includegraphics[width=0.80\linewidth]{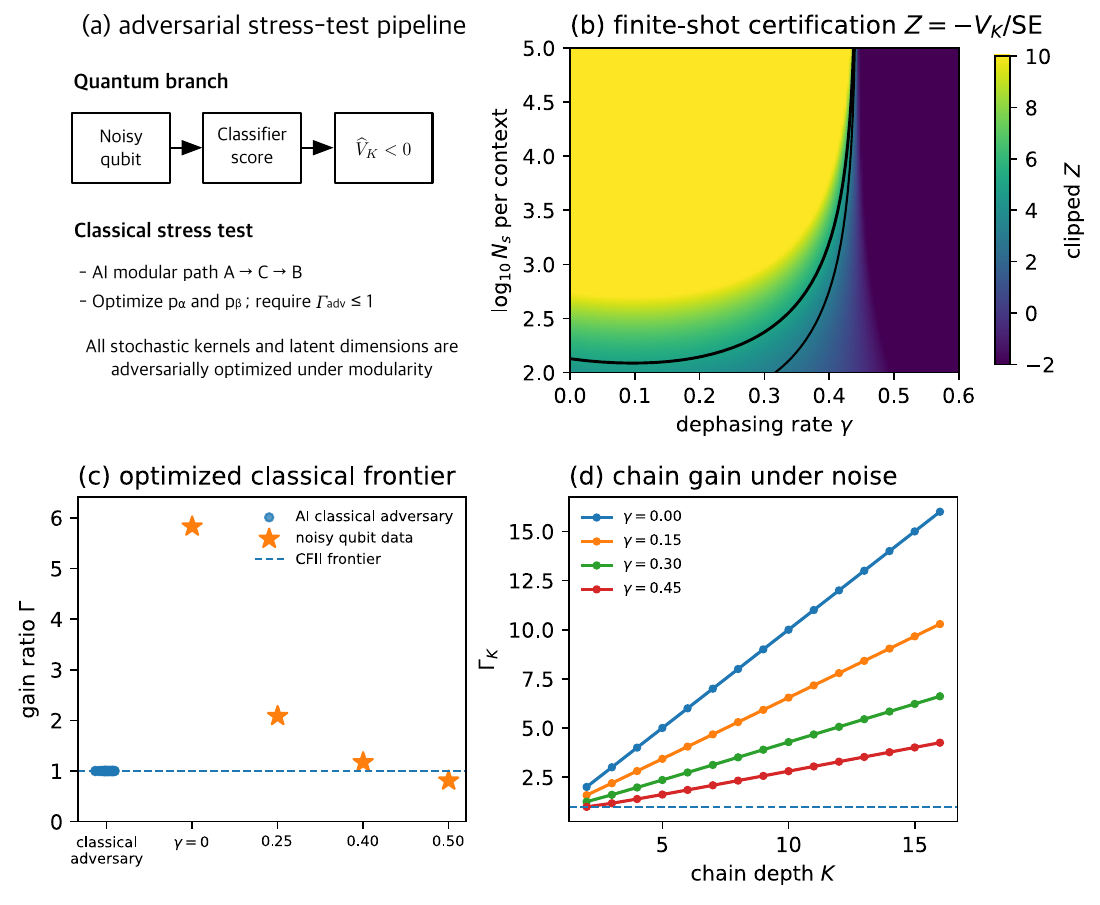}
\caption{AI-assisted adversarial finite-data stress test. (a) Noisy coherent-qubit samples are processed by a classifier score estimator to certify $\widehat V_K<0$, while an AI-optimized modular classical causal adversary is trained to maximize its gain ratio under the same causal-path constraint. (b) Finite-shot certification phase diagram for $K=4$, $T=\pi/2$, and readout error $\epsilon_{\rm r}=0.02$; $N_s$ denotes the number of independent shots assigned to each context. The color shows the delta-method significance $Z=-\widehat V_K/{\rm SE}(\widehat V_K)$; black contours mark $Z=3$ and $Z=5$. (c) The optimized classical adversary saturates the CFII frontier but remains at $\Gamma_{\rm adv}\leq 1$ within numerical precision; the adversary points represent independent random restarts of the same modular optimization, whereas the noisy coherent-qubit points show the gain at selected dephasing rates. (d) In a mid-fringe visibility model, the chain gain $\Gamma_K\simeq K\exp[-2\gamma T(1-1/K)]$ remains above unity over a finite noise window.}
\label{fig:ai_stress_main}
\end{figure*}

The analytic coherent-rotation example establishes the mechanism in its cleanest form. To test whether the certificate survives the imperfections that usually dominate an experiment, we perform an AI-assisted adversarial stress test using noisy finite samples. The simulated endpoint statistics are a visibility-reduced binary fringe
\begin{eqnarray}
p_0^{(\gamma)}(\theta)=\frac{1+z_\gamma(\theta)}{2},\quad
p_1^{(\gamma)}(\theta)=\frac{1-z_\gamma(\theta)}{2},
\label{eq:noisyfringe_main}
\end{eqnarray}
with
\begin{eqnarray}
z_\gamma(\theta)=\eta_{\rm r} e^{-\gamma\theta}\cos(\theta-\vartheta_0),
\label{eq:noisyz_main}
\end{eqnarray}
where $\vartheta_0$ is the chosen preparation angle, $\gamma$ is a dephasing-rate parameter, and $\eta_{\rm r}=1-2\epsilon_{\rm r}$ models symmetric readout error. We use a $K=4$ equal-partition causal-chain test of a total interval $T=\pi/2$ with $\epsilon_{\rm r}=0.02$. The finite-data witness is
\begin{eqnarray}
\widehat V_K=\widehat F(T)^{-1}-\sum_{j=1}^{K}\widehat F(T/K)^{-1},
\label{eq:VKhat_main}
\end{eqnarray}
with independent samples for the endpoint and segment contexts; the hats denote finite-sample score-based FI estimates. In Eq.~\eqref{eq:VKhat_main}, each quantity $\widehat F(\theta_s)$ is computed from an estimated local score, $\widehat F(\theta_s)=N_s^{-1}\sum_i\widehat\ell_s(x_{s,i}|\theta_s)^2$, with $\theta_s=T$ for the endpoint context and $\theta_s=T/K$ for a segment context. When the analytic score is not supplied, we estimate $\widehat\ell_s$ by a classifier likelihood-ratio method. Samples generated at $\theta_s+\delta$ and $\theta_s-\delta$ are given equal class priors; the optimal classifier logit approaches $\log[p(x | \theta_s+\delta) / p(x | \theta_s-\delta)]$, and division by $2\delta$ gives the finite-difference score estimate. Thus the classifier is the practical route by which the hatted FIs in Eq.~\eqref{eq:VKhat_main} are obtained. This likelihood-free step is useful because the same pipeline extends to high-dimensional records where the analytic scores are unavailable~\cite{Cranmer2020SBI}.

We then give the classical model an adversarial advantage. Here the term AI adversary means a differentiable trainable family of stochastic kernels used to search the classical model class as aggressively as possible; it does not relax the modular causal constraint. A differentiable adversary, with trainable parameters $\phi=(\alpha,\beta)$, is constrained only to obey the modular causal-path form
\begin{eqnarray}
p_\phi(c,b|\theta_{ac},\theta_{cb})=
p_{\alpha}(c|\theta_{ac})p_{\beta}(b|c,\theta_{cb}),
\label{eq:adversary_main}
\end{eqnarray}
The latent dimension, stochastic kernels, and local score structure are optimized to maximize the ratio
\begin{eqnarray}
\Gamma_{\rm adv}:=
\frac{F_B^{(u)}}{\left(F_{ac}^{-1}+F_{cb}^{-1}\right)^{-1}},
\qquad u=(1,1)^T.
\label{eq:Gammaadv_main}
\end{eqnarray}
Here, $F_B^{(u)}$ is the effective endpoint FI for the sum direction induced by the adversary's endpoint Fisher matrix. The Adam optimizer~\cite{KingmaBa2015} is restarted from many random initializations, explicitly searching for a modular classical model that crosses the CFII frontier.

The results are given in Fig.~\ref{fig:ai_stress_main}. Here, $Z=-\widehat V_K/{\rm SE}(\widehat V_K)$ denotes the delta-method certification significance, and $N_s$ is the number of independent shots assigned to each context in the finite-data estimate. For $\gamma=0.25$ and $N_s=10^3$ shots per context, the noisy coherent data give
\begin{eqnarray}
\widehat V_K = -2.58, \quad Z \simeq 12.2, \quad \Gamma_K \simeq 2.08,
\label{eq:numerical_summary_main}
\end{eqnarray}
so the negative witness is far from a finite-sample fluctuation. Fig.~\ref{fig:ai_stress_main}(b) shows how this certification changes with both dephasing and shot number: for the chosen readout error, the violation remains certifiable up to a dephasing threshold near $\gamma\simeq 0.44$. Fig.~\ref{fig:ai_stress_main}(c) separates the quantum data from the adversarial classical search. The orange points are noisy coherent-qubit gains at selected $\gamma$ values, while the blue adversary points correspond to 36 independent random restarts of the modular optimization in Eq.~\eqref{eq:adversary_main}; these restarts probe whether different initial stochastic kernels can find a classical model with $\Gamma_{\rm adv}>1$. The best restart reaches $\Gamma_{\rm adv}=0.9999999998$, which appears in the plot as saturation of the dashed CFII frontier but not a crossing of it. Fig.~\ref{fig:ai_stress_main}(d) shows the complementary chain-depth dependence: even when visibility decays with the total interval, increasing $K$ can keep $\Gamma_K$ above unity over a finite noise window because the classical benchmark decreases as a series resistance. The stress test therefore strengthens the causal interpretation: the observed FI geometry is not rescued by flexible hidden mediators or AI-optimized stochastic kernels as long as the classical modular path assumption is maintained. Full algorithms, calibration checks, and additional figures are given in the Methods and Supplementary Note~VII.

\section*{DISCUSSION}

The present framework sharpens what a Fisher-information inequality can mean. Historically, the inequalities such as Eq.~\eqref{eq:pathCFII} were often motivated as a test of discrete trajectories or intermediate-state hypotheses~\cite{Zamir1998,Tan2021,ArvidssonShukur2020,LupuGladstein2022}. Our point is more systematic. The essential object is not the word ``trajectory'' but the classical causal model class whose graph and modularity assumptions imply the inequality. Once this is recognized, the FI inequalities become a principled route from the causal hypotheses to experimental tests.

This also clarifies what is, and is not, falsified by a violation. A negative witness rules out the specific classical model class that generated the CFII. It does not prove that all classical explanations fail, nor does it claim ``quantumness'' without qualification. A larger model class---for example, one with extra latent memory, explicit context dependence, or relaxed conditional independences---may still explain the data. The CFIIs therefore map the space of viable explanations with metrological precision.

The resource interpretation is specific in the same model-relative sense. The certified gain is not measured against an unspecified classical standard; it is measured against the FI frontier forced by the tested DAG, conditional independences, and modular parameter dependence. A CFII violation is therefore metrologically meaningful because it beats a benchmark that is explicit and internally certified by the same data. The design task is then not merely to maximize a generic FI, but to identify a target classical causal narrative and engineer the experiment so that the corresponding witness becomes negative. The synergy formula shows what this means experimentally: create positive score correlations between the would-be modules while keeping the inference problem regular enough for efficient estimation.

The framework also helps position the contextual quantum metrology (coQM)~\cite{Jae2024}. In the coQM setting, the relevant classicality condition is no-signaling in time (NSIT), which can be written as a conditional-independence statement $B \perp S | \theta$ for the context variable $S$ that decides whether an earlier measurement is performed~\cite{Halliwell2017,Jae2024}. NSIT violation is therefore one particular causal-model collapse, and the resulting Fisher-information gain fits naturally into our witness-to-resource logic. However, the NSIT viewpoint does not subsume the general CFII program: one can construct models that satisfy NSIT perfectly yet still violate the causal-path CFII maximally, because NSIT constrains the context dependence while the path CFII constrains the modular mediation and additive series composition. See Supplementary Appendix~A, especially Supplementary Proposition~S3.

Several limitations should be kept in view. First, FI is a local quantity, so the interpretation is sharpest in regular asymptotic regimes. Second, failure to observe a violation only shows consistency with the tested model class; it does not prove the model. Third, finite-data implementations require reliable score calibration or equivalent local estimation tools. None of these caveats undercuts the central message. They delimit the regime in which causal-model falsification and violation-certified gain are read most cleanly from experimental data.

In summary, the causal Fisher-information inequalities unify the three themes that are often treated separately: causal inference, nonclassicality witnessing, and precision metrology. A classical causal hypothesis defines an FI frontier; a violation of the frontier falsifies the hypothesis and simultaneously certifies a resource. The single-qubit example and the adversarial finite-data stress test show that this logic is not an abstract reinterpretation but an experimentally concrete program. Once inverse FI is recognized as the natural resistance of classical mediation, every negative CFII witness marks a place where nature refuses to pay the classical series penalty.

\section*{Acknowledgments}

This work was supported by the Ministry of Science, ICT and Future Planning (MSIP) through the National Research Foundation of Korea (RS-2024-00432214, RS-2025-03532992, and RS-2025-18362970) and the Institute of Information and Communications Technology Planning and Evaluation grant funded by the Korean government (RS-2019-II190003, ``Research and Development of Core Technologies for Programming, Running, Implementing and Validating of Fault-Tolerant Quantum Computing System''), the Korean ARPA-H Project through the Korea Health Industry Development Institute (KHIDI), funded by the Ministry of Health \& Welfare, Republic of Korea (RS-2025-25456722), and the Ministry of Trade, Industry and Resources (MOTIR), Korea, under the project ``Industrial Technology Infrastructure Program'' (RS-2024-00466693).

\section*{METHODS}

\subsection*{Numerical stress-test implementation}

The numerical stress test in Fig.~\ref{fig:ai_stress_main} is based on the Bernoulli family in Eqs.~\eqref{eq:noisyfringe_main}--\eqref{eq:noisyz_main}. The corresponding FI is evaluated as
\begin{eqnarray}
F_\gamma(\theta)=\frac{[\partial_\theta z_\gamma(\theta)]^2}{1-z_\gamma(\theta)^2},
\label{eq:noisyFI_methods}
\end{eqnarray}
with the same readout-error and dephasing parameters used in the Results. For finite-shot certification, the independent samples are drawn for the endpoint context and the $K$ segment contexts. The context-wise score estimates give squared-score samples $q_{s,i}=\widehat\ell_s(x_{s,i}|\theta)^2$, from which
\begin{eqnarray}
\widehat{F}_s &=& \frac{1}{N_s}\sum_{i=1}^{N_s}q_{s,i}, \nonumber \\
\widehat{\mathrm{Var}}(\widehat{F}_s) &=& \frac{1}{N_s(N_s-1)}\sum_{i=1}^{N_s}(q_{s,i}-\widehat{F}_s)^2.
\label{eq:FIvariance_methods}
\end{eqnarray}
The standard error of $\widehat{V}_K$ is then obtained by first-order propagation of these context-wise variances through Eq.~\eqref{eq:VKhat_main}. The modular classical adversary is implemented by softmax-normalized stochastic kernels constrained to the factorization in Eq.~\eqref{eq:adversary_main}; only the kernel parameters are optimized, and the modular dependence on $\theta_{ac}$ and $\theta_{cb}$ is never relaxed. Further hyperparameters and calibration checks are reported in Supplementary Note~VII.

\subsection*{Where to find full proofs (Supplementary Information guide)}

The complete proofs, extended derivations, and additional examples are provided in the Supplementary Information (SI). The main text keeps the assumptions and operational witnesses explicit; the SI serves as the proof-complete companion rather than as a substitute for the basic definitions. For transparency, the main text claims map to the SI as follows.
\begin{itemize}
    \setlength{\itemsep}{2pt}
    \setlength{\parskip}{0pt}
    \setlength{\parsep}{0pt}
    \item \textit{General operational framework and definition of classical causal-model compatibility:} Supplementary Notes~II.A--II.E and Supplementary Definitions~S1--S2.
    \item \textit{Additivity of Fisher information across contexts and multi-context sampling:} Supplementary Note~II.B and Supplementary Theorem~S2.
    \item \textit{Effective Fisher information for causal parameters and the multi-parameter Cram\'er--Rao geometry:} Supplementary Note~III.A and Supplementary Theorem~S3.
    \item \textit{Causal data processing for the Fisher information matrix:} Supplementary Note~III.B, Supplementary Theorem~S4, and Supplementary Corollary~S1.
    \item \textit{Causal-path CFII, information-resistance picture, and the $K$-step chain generalization:} Supplementary Notes~III.C--III.D, Supplementary Theorem~S5, Supplementary Remark~S1, and Supplementary Corollary~S2.
    \item \textit{Model falsification, feasible FI regions, and the causal-path witness:} Supplementary Notes~IV.A--IV.C, Supplementary Definition~S3, Supplementary Theorems~S7--S8, and Supplementary Definition~S4.
    \item \textit{Finite-data certification of CFII violation:} Supplementary Note~IV.D and Supplementary Theorems~S9--S10.
    \item \textit{Witness-to-resource transition, classical precision frontiers, and estimator-level achievability:} Supplementary Notes~V.A--V.C and Supplementary Theorems~S11--S13.
    \item \textit{Fisher-information synergy and long-chain scaling intuition:} Supplementary Notes~V.D--V.E and Supplementary Theorem~S14.
    \item \textit{Single-qubit coherent example, generic landscapes, split-optimized benchmarks, and chain-amplified gain:} Supplementary Notes~VI.A--VI.D, Supplementary Theorem~S15, and Supplementary Results~S1--S3.
    \item \textit{AI-assisted adversarial finite-data stress test, classifier score estimation, noisy certification, and optimized modular classical adversaries:} Supplementary Note~VII.
    \item \textit{Contextual quantum metrology as an NSIT-based special case, together with its scope and limits:} Supplementary Appendix~A, Supplementary Proposition~S1, Supplementary Theorems~S16--S17, and Supplementary Proposition~S3.
\end{itemize}

\subsection*{Conventions}

Throughout, the main text focuses on a single real parameter for clarity, while the SI develops the corresponding multi-parameter geometry needed for additive causal parameters and effective Fisher information. All statements about precision are local and asymptotic unless stated otherwise, and all finite-data certification claims assume the regularity conditions made explicit in the SI.


\clearpage
\newpage
\onecolumngrid

\makeatletter
\let\addcontentsline\origaddcontentsline
\makeatother

\newcommand{\suppnote}[1]{%
  \refstepcounter{section}%
  \setcounter{subsection}{0}%
  \setcounter{subsubsection}{0}%
  \section*{Supplementary Note~\Roman{section}. #1}%
}
\newcommand{\suppappendix}[1]{%
  \refstepcounter{section}%
  \setcounter{subsection}{0}%
  \setcounter{subsubsection}{0}%
  \section*{Supplementary Appendix~\Alph{section}. #1}%
}
\renewcommand{\thetheorem}{S\arabic{theorem}}
\renewcommand{\theproposition}{S\arabic{proposition}}
\renewcommand{\thelemma}{S\arabic{lemma}}
\renewcommand{\thecorollary}{S\arabic{corollary}}
\renewcommand{\thedefinition}{S\arabic{definition}}
\renewcommand{\theremark}{S\arabic{remark}}
\renewcommand{\theresult}{S\arabic{result}}

\part{Supplementary Information}

This supplementary information (SI) is written as a self-contained and proof-complete companion to our main submitting manuscript. It provides full derivations, rigorous proofs, and extended methodological details required for a complete understanding of our work.

\tableofcontents

\suppnote{Introduction: The work summary}

Precision parameter estimation is central to modern science and engineering. In a metrological task, an unknown parameter $\theta$ (e.g., a time interval, phase shift, rotation angle, field strength, or concentration) is encoded into measurement statistics, and one aims to infer $\theta$ from observed data. A standard operational quantifier of statistical sensitivity is the Fisher information (FI)~\cite{Frieden2004}. For a measurement outcome $x$ with conditional probability $p(x|\theta)$, the FI is defined by
\begin{eqnarray}
F(\theta) = \sum_x p(x|\theta)\left( \frac{\partial}{\partial\theta}\ln{p(x|\theta)} \right)^2.
\label{eq:FI_def_intro}
\end{eqnarray}
In the asymptotic limit, FI governs achievable precision through the Cram\'er--Rao bound~\cite{Cramer2016, Helstrom1976}.

Quantum metrology seeks to exploit genuine quantum features to enhance estimation precision beyond what is achievable by ``most classical'' strategies. In the conventional framework, the ultimate bound is often characterized in terms of quantum Fisher information (QFI), which is defined by optimizing over all measurements~\cite{BC1994, Paris2009}. However, the measurement that attains the bound can be difficult to identify or experimentally infeasible to implement in realistic settings, motivating complementary approaches based on operational and experimentally accessible criteria. A particularly appealing direction is to employ FI-based criteria, including FI inequalities, as experimentally accessible witnesses of nonclassical resources and metrological advantage~\cite{Zamir,Tan2021FIResources,ArvidssonShukur2020PostselectedMetrology,LupuGladstein2022NegativeQuasi}.

A representative example is a classical FI inequality that appears when one estimates a total parameter either directly or by dividing the process into intermediate segments. Let us consider a parameter decomposition as
\begin{eqnarray}
\theta_{ab}=\theta_{ac}+\theta_{cb},
\label{eq:param_add_intro}
\end{eqnarray}
which is naturally associated with an intermediate description of a process. Under the classical assumptions that amount to a causal-path structure together with statistical independence of the estimation parameters, the following inequality is satisfied:
\begin{eqnarray}
F(\theta_{ab})^{-1}\geq F(\theta_{ac})^{-1}+F(\theta_{cb})^{-1}.
\label{eq:FI_ineq_intro}
\end{eqnarray}
One may define the violation statistic
\begin{eqnarray}
V:=F(\theta_{ab})^{-1}-F(\theta_{ac})^{-1}-F(\theta_{cb})^{-1},
\label{eq:V_def_intro}
\end{eqnarray}
and an observed value $V < 0$ indicates a violation of Eq.~(\ref{eq:FI_ineq_intro}). Such violations arise, for instance, in simple quantum dynamics when one attempts to characterize the evolution with a specific sequence of intermediate states. Crucially, the correct interpretation of a violation is not ``quantumness in general'' but the failure of the underlying classical assumptions used to derive Eq.~(\ref{eq:FI_ineq_intro}).

In this work, we elevate this observation into a general principle and propose a unified framework: FI inequalities of the type in Eq.~(\ref{eq:FI_ineq_intro}) are most naturally understood as tests of classical causal models. The bridge is provided by Bayesian networks and conditional independence. A classical causal model is specified by a directed acyclic graph (DAG) (possibly with latent variables), which implies a family of conditional-independence relations~\cite{Pearl2009Causality}. When a metrological protocol is assumed to admit a description within such a model class, those conditional-independence relations translate into constraints on the achievable FI. Thus, violating an FI inequality is equivalent to falsifying the corresponding classical causal model class. To emphasize the generality, we introduce the notion of \emph{causal Fisher-information inequalities} (CFIIs): given a classical causal model class $\mathcal{M}$ defined by a DAG-induced set of conditional independences, CFIIs are the FI constraints that must hold for any experiment whose statistics admit a realization within $\mathcal{M}$. The trajectory-based inequality in Eq.~(\ref{eq:FI_ineq_intro}) is then a special case associated with a causal-path (Markov-chain-like) model class.
\begin{theorem}[Causal Fisher-information inequalities (informal statement)]
\label{thm:CFII_informal_intro}
Consider a family of parameterized experiments specified by conditional distributions $\{p(x | s,\theta)\}$, where $s$ denotes an experimental context (measurement choice, intermediate intervention, segmentation strategy, or control setting). Let $\mathcal{M}$ be a classical causal model class defined by a DAG that imposes conditional-independence constraints among the relevant variables, together with regularity assumptions ensuring that the parameter dependence factorizes consistently with the model. Then, there exists a family of inequalities among the corresponding Fisher informations that must be satisfied by any data compatible with $\mathcal{M}$. Any experimental violation of these inequalities certifies that no model in the class $\mathcal{M}$ can reproduce the observed statistics.
\end{theorem}
This causal-model perspective has two immediate consequences. First, it provides a systematic route to generalization: rather than focusing on a specific hypothesis such as a discrete trajectory decomposition, one may start from an arbitrary causal model class $\mathcal{M}$, derive its associated CFIIs, and use experimental data to test (and potentially falsify) $\mathcal{M}$. Second, it reframes nonclassical resources in metrology: what appears as contextuality, non-Markovianity, or incompatibility can be interpreted as the collapse of a classical causal description, and such collapse can be harnessed as an information-theoretic resource.

\suppnote{General framework}\label{sec:framework}

Here we present a general operational framework that unifies (i) standard parameter estimation, (ii) experiments involving multiple contexts (measurement choices, interventions, segmentations), and (iii) classical causal model classes specified by conditional-independence constraints. Throughout, we focus on a single real parameter $\theta$ for clarity. The generalization to multi-parameter estimation can be formulated by replacing FI with the Fisher information matrix~\cite{Demkowicz2020,Liu2020}.

\subsection{Operational setting: parameterized experiments with contexts}\label{subsec:operational_setting}

We consider an experiment in which an unknown parameter $\theta$ is encoded into outcome statistics. The experimenter can choose an experimental context (or setting) $s \in {\cal S}$, where $s$ may represent a measurement choice, a control operation, an intermediate intervention, or a segmentation strategy. For each context $s$, the experiment produces an outcome $x$ in an outcome set ${\cal X}_s$ according to a conditional distribution
\begin{eqnarray}
p(x | s,\theta).
\label{eq:model_px_intro2}
\end{eqnarray}
We allow ${\cal X}_s$ to depend on $s$ to incorporate, for example, different output alphabets for different measurement settings. For notational simplicity, we present the formulas for discrete outcomes; continuous outcomes are treated by replacing sums with integrals.

A data set consists of independent samples collected under possibly multiple contexts. Let ${\cal S}_D\subseteq{\cal S}$ denote the set of contexts used in the protocol. For each $s \in {\cal S}_D$, we collect $N_s$ independent outcomes
\begin{eqnarray}
x_{s,1},x_{s,2},\cdots,x_{s,N_s}.
\label{eq:data_per_context}
\end{eqnarray}
The total number of samples is
\begin{eqnarray}
N = \sum_{s \in {\cal S}_D} N_s.
\label{eq:total_samples}
\end{eqnarray}
Assuming conditional independence of samples given $(s,\theta)$, the likelihood function is
\begin{eqnarray}
{\cal L}(\theta|D) = \prod_{s \in {\cal S}_D} \prod_{i=1}^{N_s} p(x_{s,i} | s,\theta),
\label{eq:likelihood}
\end{eqnarray}
and the log-likelihood is
\begin{eqnarray}
\ell(\theta|D) = \sum_{s \in {\cal S}_D} \sum_{i=1}^{N_s} \ln{p(x_{s,i} | s,\theta)}.
\label{eq:loglikelihood}
\end{eqnarray}
An estimator is a function of the data, denoted as $\hat{\theta} = \hat{\theta}(D)$.

\subsection{Fisher information under multi-context sampling}\label{subsec:FI_multicontext}

For each context $s$, the Fisher information is defined by~\cite{Frieden2004}
\begin{eqnarray}
F_s(\theta) = \sum_{x \in {\cal X}_s} p(x | s,\theta)\left(\frac{\partial}{\partial\theta} \ln{p(x | s,\theta)} \right)^2.
\label{eq:FI_context}
\end{eqnarray}
The quantity $F_s(\theta)$ characterizes the local sensitivity of the outcome distribution for context $s$ with respect to $\theta$.

When data are collected under multiple contexts, the relevant statistical model is the product distribution associated with Eq.~(\ref{eq:likelihood}). The corresponding Fisher information of the full data set, denoted by $F_D(\theta)$, is obtained from the score $\partial_\theta\ell(\theta|D)$. A basic property of FI is the additivity for independent samples, which yields a simple decomposition across contexts~\cite{Frieden2004}.
\begin{theorem}[Additivity of Fisher information across contexts]
\label{thm:FI_additivity}
Assume that conditioned on $(s,\theta)$ all outcomes $\{x_{s,i}\}$ are independent and distributed as $p(x|s,\theta)$. Then, the Fisher information of the full data set $D$ satisfies
\begin{eqnarray}
F_D(\theta)=\sum_{s \in {\cal S}_D} N_s F_s(\theta).
\label{eq:FI_total}
\end{eqnarray}
\end{theorem}

\begin{proof}---From Eq.~(\ref{eq:loglikelihood}),
\begin{eqnarray}
\frac{\partial}{\partial\theta}\ell(\theta | D) = \sum_{s \in {\cal S}_D} \sum_{i=1}^{N_s}\frac{\partial}{\partial\theta}\ln{p(x_{s,i} | s,\theta)}.
\label{eq:score_sum}
\end{eqnarray}
Taking the expectation over the product distribution and using the independence, the cross terms vanish because the score has zero mean under standard regularity conditions. Thus, one obtains Eq.~(\ref{eq:FI_total}).
\end{proof}

{\bf Theorem~\ref{thm:FI_additivity}} shows that the sample allocation $\{N_s\}$ determines the effective FI through a linear combination of $\{F_s(\theta)\}$. It is convenient to introduce weights $q_s:=N_s/N$, so that
\begin{eqnarray}
\frac{1}{N}F_D(\theta) = \sum_{s \in {\cal S}_D} q_s F_s(\theta).
\label{eq:FI_per_sample}
\end{eqnarray}
In the asymptotic regime, an unbiased estimator $\hat{\theta}$ satisfies the Cram\'er--Rao bound~\cite{Cramer2016, Helstrom1976}
\begin{eqnarray}
{\rm Var}(\hat{\theta}) \geq \frac{1}{F_D(\theta)}.
\label{eq:CRB_total}
\end{eqnarray}
Hence, for a fixed total budget $N$, the performance comparisons between protocols must specify the sample allocation across contexts.

\subsection{Quantum realization of the conditional model}\label{subsec:quantum_realization}

The operational model in Eq.~(\ref{eq:model_px_intro2}) includes the quantum experiments as a special case. Let the probe be prepared in a quantum state $\hat{\rho}(\theta)$ that depends on $\theta$. A context $s$ specifies a measurement (and possibly additional control operations), represented by a POVM $\{\hat{M}_{x | s}\}_{x \in {\cal X}_s}$ satisfying $\hat{M}_{x | s} \geq 0$ and $\sum_{x \in {\cal X}_s}\hat{M}_{x | s} = \hat{\mathds{1}}$. Here, the outcome statistics are given by the Born rule,
\begin{eqnarray}
p(x | s,\theta)=\tr{\left[\hat{M}_{x | s}\hat{\rho}(\theta)\right]}.
\label{eq:Born_rule_context}
\end{eqnarray}
For a pure-state model, $\hat{\rho}(\theta) = \ketbra{\psi(\theta)}{\psi(\theta)}$.

Sequential or adaptive protocols can also be included by enlarging the context and outcome spaces. For example, a consecutive measurement producing an outcome pair $(a,b)$ can be treated as a single outcome $x=(a,b)$ for a context $s$ that specifies the corresponding measurement instrument. Thus, the multi-context framework in Sec.~\ref{subsec:operational_setting} naturally accommodates both conventional single-setting metrology and protocols that integrate data from multiple settings.

\subsection{Classical causal models with conditional independence}\label{subsec:causal_models}

We next formalize the notion of a classical causal explanation of experimental statistics. A classical causal model is specified by a directed acyclic graph (DAG) $G$ whose nodes correspond to random variables~\cite{Pearl2009Causality}. Let ${\cal V} = \{X_1, X_2, \cdots, X_n\}$ denote the set of model variables, and let ${\rm pa}(X_j)$ denote the set of parents of node $X_j$ in $G$. The defining property of a Bayesian network is the factorization of the joint distribution:
\begin{eqnarray}
p(x_1, x_2, \cdots, x_n)=\prod_{j=1}^n p(x_j | {\rm pa}(x_j)).
\label{eq:DAG_factorization}
\end{eqnarray}
The factorization implies conditional-independence (CI) relations. For three sets of variables $U$, $V$ and $W$, the CI statement $U \perp V|W$ means
\begin{eqnarray}
p(u, v | w)=p(u | w)p(v | w),
\label{eq:CI_def}
\end{eqnarray}
for all values $u$, $v$ and $w$ with $p(w)>0$. A DAG $G$ induces a collection of CI constraints (e.g., via d-separation), and the set of distributions satisfying Eq.~(\ref{eq:DAG_factorization}) is equivalently characterized by these constraints together with positivity conditions.

In metrology, the experimenter controls the context $S$ and the parameter $\theta$ is unknown. To connect the causal models to the operational distributions $p(x | s,\theta)$, we treat $S$ as an exogenous variable that can be set by the experimenter, and we allow latent variables $\Lambda$ to represent unobserved degrees of freedom. A parameterized causal model generates a family of conditional distributions $p(v | s,\theta)$ over ${\cal V}$. In the most general form considered here, the model assumes the existence of $\Lambda$ such that
\begin{eqnarray}
p(v, \lambda | s,\theta)=p(\lambda)\prod_{j=1}^n p(x_j | {\rm pa}(x_j), s, \theta, \lambda),
\label{eq:param_causal_factorization}
\end{eqnarray}
where $v=(x_1,\cdots,x_n)$ and $\lambda$ is a value of $\Lambda$. The operational distribution $p(x|s,\theta)$ of an observed outcome variable $X$ is obtained by marginalization over unobserved variables.
\begin{definition}[Causal model compatibility]
\label{def:compatibility}
Let $X$ be the observed outcome variable and let $p(x|s,\theta)$ be an operational model. We say that $p(x | s,\theta)$ is compatible with a classical causal model class ${\cal M}(G)$ if there exist (possibly latent) variables ${\cal V}\setminus\{X\}$ and $\Lambda$ and a family of joint distributions $p(v, \lambda | s,\theta)$ that factorize as in Eq.~(\ref{eq:param_causal_factorization}) for the DAG $G$, such that the induced marginal equals the operational model for all contexts and parameter values,
\begin{eqnarray}
p(x | s,\theta)=\sum_{v \setminus x}\sum_{\lambda} p(v, \lambda|s, \theta).
\label{eq:marginal_compatibility}
\end{eqnarray}
\end{definition}
{\bf Definition~\ref{def:compatibility}} formalizes the statement that observed statistics admit a classical causal explanation within a specified model class. Different choices of $G$ and different restrictions on how $\theta$ enters Eq.~(\ref{eq:param_causal_factorization}) define different notions of classicality~\cite{Pearl2009Causality,Allen2017QuantumCommonCauses,Barrett2021CyclicQuantumCausal}. In particular, trajectory-based (Markov-chain-like) models and context-free (noncontextual) models correspond to different CI structures and lead to different testable constraints.

\subsection{Causal Fisher-information inequality framework}\label{subsec:CFII_def}

Given a causal model class ${\cal M}$, the set of compatible operational models $\{p(x | s,\theta)\}$ is constrained by the CI relations implied by ${\cal M}$. Since FI is a functional of $p(x | s,\theta)$, these constraints induce relations among the Fisher informations $\{F_s(\theta)\}$. This motivates the following general definition.
\begin{definition}[Causal Fisher-information inequality]
\label{def:CFII}
Fix a causal model class ${\cal M}$ and a finite set of contexts ${\cal S}_0\subseteq{\cal S}$. A causal Fisher-information inequality (CFII) is an inequality of the form
\begin{eqnarray}
{\cal G}\left(\{F_s(\theta)\}_{s \in {\cal S}_0}\right) \geq 0
\label{eq:CFII_general_form}
\end{eqnarray}
that holds for all $\theta$ and for every operational model $p(x|s,\theta)$ compatible with ${\cal M}$ in the sense of Definition~\ref{def:compatibility}.
\end{definition}
CFIIs are model-dependent: the function ${\cal G}$ depends on the causal assumptions encoded in ${\cal M}$.  In the following, we derive explicit CFIIs for several broad classes of causal models. The central operational message is that if experimentally estimated Fisher informations violate a CFII, then the observed statistics are incompatible with the assumed classical causal model class. This provides a general route to interpret FI-inequality violations as falsification of classical causal explanations, and it sets the stage for viewing such ``causal-model collapse'' as a metrological resource.

\suppnote{Causal Fisher-information inequalities (CFII)}\label{sec:cfii}

The general framework in Sec.~\ref{sec:framework} emphasizes a simple but powerful idea: once a statistical model is required to arise from a classical causal structure, the likelihood is no longer an arbitrary function of the parameter. It must factorize according to the causal graph, and this factorization enforces quantitative restrictions on the Fisher information that can be observed at the endpoints of the experiment. We call such restrictions ``causal Fisher-information inequalities (CFIIs).''

Here, we derive the CFIIs that constitute the backbone of our framework. The derivation is intentionally constructive: we isolate two universal principles---a multi-parameter Cram\'er--Rao geometry and a causal data-processing law---and then show how they fuse into a striking ``series law'' for Fisher information along a causal path. The conclusion is conceptually sharp: inverse Fisher information behaves as an information-resistance that must accumulate along classical causal bottlenecks. Hence, whenever an experiment exhibits a precision that beats this accumulation rule, the assumed classical causal model is not merely implausible; this is impossible by mathematics alone.

\subsection{Effective Fisher information for causal parameters}\label{subsec:effectiveFI}

CFIIs are most naturally expressed in terms of the causal parameters, which are often not the primitive parameters of the model but rather functions of them. In particular, in causal networks, the parameter of interest is frequently an additive quantity accumulated along edges (e.g., a total delay, a total phase, a total action), while each causal module contributes its own parameter.

Let $X$ be an observation (possibly multivariate) generated by a regular parametric model $p(x|\boldsymbol{\theta})$ with a $d$-dimensional parameter $\boldsymbol{\theta} = (\theta_1, \ldots, \theta_d)^T$. We define the score vector and Fisher information matrix as
\begin{eqnarray}
\mathbf{s}_X(x | \boldsymbol{\theta}) &:=& \nabla_{\boldsymbol{\theta}} \ln{p(x | \boldsymbol{\theta})}, \label{eq:score_vector_def} \\
\mathbf{F}_X(\boldsymbol{\theta}) &:=& \mathbb{E}_{x|\boldsymbol{\theta}}\Big[\mathbf{s}_X(x | \boldsymbol{\theta})\mathbf{s}_X(x | \boldsymbol{\theta})^T\Big]. \label{eq:FIM_def}
\end{eqnarray}
We assume standard regularity conditions so that $\mathbb{E}[\mathbf{s}_X]=\mathbf{0}$ and differentiation can be interchanged with integration.

The causal quantity of interest is modeled as a scalar function $g(\boldsymbol{\theta})$. In our applications the crucial case is the linear causal parameter
\begin{eqnarray}
\Theta = g(\boldsymbol{\theta}) = \mathbf{u}^T \boldsymbol{\theta},
\label{eq:linear_causal_parameter}
\end{eqnarray}
for some fixed vector $\mathbf{u} \in \mathbb{R}^d$. The appropriate single-number measure of information about $\Theta$ in a multi-parameter model is the effective Fisher information defined by~\cite{Demkowicz2020,Liu2020}
\begin{eqnarray}
F_X^{(\mathbf{u})}(\boldsymbol{\theta}) := \Big(\mathbf{u}^T \mathbf{F}_X(\boldsymbol{\theta})^{-1}\mathbf{u}\Big)^{-1},
\label{eq:effective_FI_def}
\end{eqnarray}
whenever $\mathbf{F}_X(\boldsymbol{\theta})$ is invertible on the support of $\mathbf{u}$. When $d=1$ and $\mathbf{u}=1$, this reduces to the standard Fisher information. Now we can state that:
\begin{theorem}[Cram\'er--Rao bound for a linear causal parameter]
\label{thm:CRB_linear}
Let $X \sim p(x|\boldsymbol{\theta})$ and let $\Theta=\mathbf{u}^T\boldsymbol{\theta}$ be the causal parameter in Eq.~(\ref{eq:linear_causal_parameter}). For any unbiased estimator $\hat{\Theta}(X)$, one has
\begin{eqnarray}
{\rm Var}(\hat{\Theta}) \ge \mathbf{u}^T \mathbf{F}_X(\boldsymbol{\theta})^{-1}\mathbf{u} = \frac{1}{F_X^{(\mathbf{u})}(\boldsymbol{\theta})}.
\label{eq:CRB_linear}
\end{eqnarray}
\end{theorem}

\begin{proof}---Let $r(X):=\hat{\Theta}(X)-\Theta$. The unbiasedness implies $\mathbb{E}[r]=0$. Under the regularity assumptions,
\begin{eqnarray}
\frac{\partial}{\partial \theta_i}\mathbb{E}[\hat{\Theta}] = \int dx \, \hat{\Theta}(x)\frac{\partial}{\partial \theta_i}p(x | \boldsymbol{\theta}) = \mathbb{E}\Big[\hat{\Theta} \, s_i(X | \boldsymbol{\theta})\Big],
\label{eq:unbiased_derivative_identity}
\end{eqnarray}
where $s_i$ is the $i$th component of the score vector in Eq.~(\ref{eq:score_vector_def}). Since $\mathbb{E}[\hat{\Theta}] = \Theta = \mathbf{u}^T\boldsymbol{\theta}$, we have $\partial_{\theta_i}\mathbb{E}[\hat{\Theta}]=u_i$. Using $\mathbb{E}[s_i] = 0$, Eq.~(\ref{eq:unbiased_derivative_identity}) yields
\begin{eqnarray}
\mathbb{E}\Big[r(X) \, s_i(X|\boldsymbol{\theta})\Big] = u_i.
\label{eq:cov_error_score}
\end{eqnarray}
Then, by letting $\mathbf{s} = \mathbf{s}_X(X | \boldsymbol{\theta})$ and $\mathbf{u}=(u_1, \ldots, u_d)^T$, consider the nonnegative quantity
\begin{eqnarray}
\mathbb{E}\Big[(r-\mathbf{a}^T\mathbf{s})^2\Big] = \mathbb{E}[r^2] - 2\mathbf{a}^T \mathbb{E}[r \mathbf{s}] + \mathbf{a}^T \mathbb{E}[\mathbf{s} \mathbf{s}^T] \mathbf{a} \ge 0.
\label{eq:csquare_expand}
\end{eqnarray}
By Eq.~(\ref{eq:cov_error_score}) and Eq.~(\ref{eq:FIM_def}), this becomes
\begin{eqnarray}
{\rm Var}(\hat{\Theta}) - 2\mathbf{a}^T\mathbf{u} + \mathbf{a}^T\mathbf{F}_X(\boldsymbol{\theta})\mathbf{a} \ge 0.
\label{eq:quadratic_form_bound}
\end{eqnarray}
Minimization of the right-hand side over $\mathbf{a}$ gives the choice $\mathbf{a}=\mathbf{F}_X(\boldsymbol{\theta})^{-1}\mathbf{u}$ and yields
\begin{eqnarray}
{\rm Var}(\hat{\Theta}) \ge \mathbf{u}^T\mathbf{F}_X(\boldsymbol{\theta})^{-1}\mathbf{u}.
\label{eq:CRB_linear_proved}
\end{eqnarray}
This is Eq.~(\ref{eq:CRB_linear}), and the identity with Eq.~(\ref{eq:effective_FI_def}) is immediate.
\end{proof}

\subsection{Causal data processing for Fisher information}\label{subsec:dp_fi}

A causal model is not only a factorization; it is also a directional constraint: information can only flow forward through the causal arrows, and any omission of variables corresponds to a stochastic post-processing that is independent of the parameter. Such processing cannot create information. We provide the following theorem:
\begin{theorem}[Data-processing inequality for the Fisher information matrix]
\label{thm:DPI_FIM}
Let $X\sim p(x|\boldsymbol{\theta})$ and let $Y$ be generated from $X$ through a Markov kernel $p(y|x)$ that does not depend on $\boldsymbol{\theta}$:
\begin{eqnarray}
p(y | \boldsymbol{\theta}) = \int dx \, p(y|x) p(x | \boldsymbol{\theta}).
\label{eq:theta_independent_channel}
\end{eqnarray}
Then, the Fisher information matrices satisfy
\begin{eqnarray}
\mathbf{F}_Y(\boldsymbol{\theta}) \preceq \mathbf{F}_X(\boldsymbol{\theta}),
\label{eq:DPI_FIM_statement}
\end{eqnarray}
where $\preceq$ denotes the positive-semidefinite order.
\end{theorem}

\begin{proof}---Let $\mathbf{s}_X = \nabla_{\boldsymbol{\theta}}\ln{p(x | \boldsymbol{\theta})}$ and $\mathbf{s}_Y = \nabla_{\boldsymbol{\theta}} \ln{p(y | \boldsymbol{\theta})}$. Differentiating Eq.~(\ref{eq:theta_independent_channel}), we have
\begin{eqnarray}
\nabla_{\boldsymbol{\theta}} \, p(y | \boldsymbol{\theta}) = \int dx \, p(y|x)\nabla_{\boldsymbol{\theta}} p(x | \boldsymbol{\theta}) = \int dx \, p(y|x) p(x | \boldsymbol{\theta})\mathbf{s}_X(x | \boldsymbol{\theta}).
\label{eq:diff_marginal}
\end{eqnarray}
Dividing by $p(y|\boldsymbol{\theta})$, we obtain
\begin{eqnarray}
\mathbf{s}_Y(y | \boldsymbol{\theta}) = \frac{1}{p(y | \boldsymbol{\theta})}\int dx \, p(y|x) p(x | \boldsymbol{\theta})\mathbf{s}_X(x | \boldsymbol{\theta}) = \mathbb{E}\bigl[ \mathbf{s}_X(X | \boldsymbol{\theta}) \,|\, Y=y \bigr],
\label{eq:score_conditional_expectation}
\end{eqnarray}
where the conditional expectation is taken with respect to the joint distribution $p(x,y | \boldsymbol{\theta})=p(y|x) p(x | \boldsymbol{\theta})$.

For any fixed vector $\mathbf{v}\in\mathbb{R}^d$,
\begin{eqnarray}
\mathbf{v}^T\mathbf{F}_Y(\boldsymbol{\theta})\mathbf{v} &=& \mathbb{E}\bigl[\big(\mathbf{v}^T\mathbf{s}_Y(Y | \boldsymbol{\theta})\big)^2\bigr] = \mathbb{E}\bigl[\big(\mathbf{v}^T\mathbb{E}[\mathbf{s}_X | Y]\big)^2\bigr] \nonumber\\
	&\le& \mathbb{E}\bigl[\mathbb{E}\big[\big(\mathbf{v}^T\mathbf{s}_X\big)^2 \,|\, Y\big]\bigr] = \mathbb{E}\bigl[\big(\mathbf{v}^T\mathbf{s}_X(X | \boldsymbol{\theta})\big)^2\bigr] = \mathbf{v}^T\mathbf{F}_X(\boldsymbol{\theta})\mathbf{v},
\label{eq:DPI_quadratic}
\end{eqnarray}
where the inequality follows from the fact that conditional expectation is an $L^2$-contraction. Since Eq.~(\ref{eq:DPI_quadratic}) holds for all $\mathbf{v}$, we obtain Eq.~(\ref{eq:DPI_FIM_statement}).
\end{proof}

An immediate consequence is that effective Fisher information for any causal parameter cannot increase under a causal coarse-graining.
\begin{corollary}[Monotonicity of effective Fisher information]
\label{cor:monotone_effectiveFI}
Under the assumptions of {\bf Theorem~\ref{thm:DPI_FIM}}, for any $\mathbf{u}$ one has
\begin{eqnarray}
F_Y^{(\mathbf{u})}(\boldsymbol{\theta}) \le F_X^{(\mathbf{u})}(\boldsymbol{\theta}).
\label{eq:monotone_effectiveFI}
\end{eqnarray}
\end{corollary}

\begin{proof}---By {\bf Theorem~\ref{thm:DPI_FIM}}, $\mathbf{F}_Y \preceq \mathbf{F}_X$. For positive definite matrices, the inversion reverses the order, hence $\mathbf{F}_Y^{-1} \succeq \mathbf{F}_X^{-1}$. Therefore, $\mathbf{u}^T\mathbf{F}_Y^{-1}\mathbf{u} \ge \mathbf{u}^T\mathbf{F}_X^{-1}\mathbf{u}$, and Eq.~(\ref{eq:monotone_effectiveFI}) follows from Eq.~(\ref{eq:effective_FI_def}).
\end{proof}

\subsection{The CFII for a causal path: inverse FI accumulates}\label{subsec:cfii_path}

We now arrive at the core inequality. Consider three nodes $A \rightarrow C \rightarrow B$ forming a causal path. Operationally, $A$ labels the preparation (or initial condition), $C$ labels an intermediate state (possibly unobserved), and $B$ labels the final measurement record. A classical causal-path hypothesis asserts that the joint conditional distribution factorizes as
\begin{eqnarray}
p(c,b | a,\theta_{ac},\theta_{cb}) = p(c|a,\theta_{ac})\,p(b|c,\theta_{cb}),
\label{eq:causal_path_factorization}
\end{eqnarray}
which is equivalent to the conditional independence $A \perp B \,|\, C$ together with modular dependence of parameters: $\theta_{ac}$ appears only in the $A \to C$ module, and $\theta_{cb}$ appears only in the $C\to B$ module.

The causal parameter of interest is the \emph{total} parameter accumulated along the path,
\begin{eqnarray}
\theta_{ab} = \theta_{ac} + \theta_{cb}.
\label{eq:additive_path_parameter}
\end{eqnarray}
This is precisely the form encountered whenever the parameter corresponds to an additive generator under composition (time translations, phase shifts, accumulated action, etc.). We will show that Eq.~(\ref{eq:causal_path_factorization}) and Eq.~(\ref{eq:additive_path_parameter}) force a stringent bound on the precision available at the endpoint.

To express the bound in a way that is both causal and operational, we introduce the local Fisher informations
\begin{eqnarray}
F_{ac}(\theta_{ac}) &:=& \mathbb{E}_{c | a,\theta_{ac}}\bigl[\big(\partial_{\theta_{ac}} \ln{p(c|a,\theta_{ac})}\big)^2 \bigr], \label{eq:F_ac_def}\\
F_{cb}(\boldsymbol{\theta}) &:=& \mathbb{E}_{c | a,\theta_{ac}}\bigl[\mathbb{E}_{b | c,\theta_{cb}}\big[\big(\partial_{\theta_{cb}}\ln{p(b|c,\theta_{cb})}\big)^2\big]\bigr]. \label{eq:F_cb_def}
\end{eqnarray}
Here, the second definition explicitly acknowledges that the information in the $C \to B$ module may depend on which intermediate value $c$ occurs, and hence must be averaged over the distribution induced by the upstream module (so in general it can also depend on $\theta_{ac}$ through $p(c|a,\theta_{ac})$).

We also define the endpoint effective Fisher information for estimating $\theta_{ab}$ from $B$ alone, in the two-parameter model in Eq.~(\ref{eq:causal_path_factorization}). Let $\boldsymbol{\theta}=(\theta_{ac},\theta_{cb})^T$ and $\mathbf{u}=(1,1)^T$, so that $\theta_{ab}=\mathbf{u}^T\boldsymbol{\theta}$. Then, we have
\begin{eqnarray}
F_{ab}^{(B)}(\boldsymbol{\theta}) := F_{B}^{(\mathbf{u})}(\boldsymbol{\theta}) = \Big(\mathbf{u}^T \mathbf{F}_{B}(\boldsymbol{\theta})^{-1}\mathbf{u}\Big)^{-1}.
\label{eq:F_ab_endpoint_effective}
\end{eqnarray}
When the model is effectively one-parameter (no nuisance degrees of freedom), $F_{ab}^{(B)}$ reduces to the usual Fisher information from $p(b|a,\theta_{ab})$. The advantage of Eq.~(\ref{eq:F_ab_endpoint_effective}) is that it remains meaningful for the genuinely causal situation where the path is described by modular parameters. We now construct the following theorem:
\begin{theorem}[CFII for a classical causal path]
\label{thm:CFII_path}
Assume the causal-path factorization Eq.~(\ref{eq:causal_path_factorization}) and the additive causal parameter in Eq.~(\ref{eq:additive_path_parameter}). Then, the effective Fisher information at the endpoint satisfies
\begin{eqnarray}
\Big(F_{ab}^{(B)}(\boldsymbol{\theta})\Big)^{-1} \ge \Big(F_{ac}(\theta_{ac})\Big)^{-1} + \Big(F_{cb}(\boldsymbol{\theta})\Big)^{-1}.
\label{eq:CFII_path}
\end{eqnarray}
\end{theorem}

\begin{proof}---We proceed in two steps: first we evaluate the information about $\theta_{ab}$ when the intermediate variable $C$ is accessible, and then we invoke causal data processing when $C$ is discarded.

\medskip\noindent
{\bf Step 1: joint record $(C,B)$ yields the ``series law'' bound.} 

\medskip
From Eq.~(\ref{eq:causal_path_factorization}), we have
\begin{eqnarray}
\ln{p(c,b | a,\theta_{ac},\theta_{cb})} = \ln{p(c | a,\theta_{ac})} + \ln{p(b | c,\theta_{cb})}.
\label{eq:loglik_split}
\end{eqnarray}
Hence, the score components are
\begin{eqnarray}
\partial_{\theta_{ac}} \ln{p(c,b | a,\theta_{ac},\theta_{cb})} &=& \partial_{\theta_{ac}} \ln{p(c | a,\theta_{ac})}, \nonumber \\
\partial_{\theta_{cb}}\ln{p(c,b | a,\theta_{ac},\theta_{cb})} &=& \partial_{\theta_{cb}} \ln{p(b | c,\theta_{cb})}.
\label{eq:score_cb}
\end{eqnarray}
The Fisher information matrix for the joint record $(C,B)$ has the entries
\begin{eqnarray}
\mathbf{F}_{(C,B)}(\boldsymbol{\theta}) = 
\begin{pmatrix}
F_{ac}(\theta_{ac}) & 0 \\
0 & F_{cb}(\boldsymbol{\theta})
\end{pmatrix}
\label{eq:FIM_joint_diagonal}
\end{eqnarray}
because the cross term vanishes: i.e., 
\begin{eqnarray}
\mathbb{E}\bigl[\big(\partial_{\theta_{ac}} \ln{p}\big)\big(\partial_{\theta_{cb}}\ln{p}\big)\bigr] = \mathbb{E}_{c | a,\theta_{ac}}\bigl[\big(\partial_{\theta_{ac}}\ln{p}(c | a,\theta_{ac})\big)\mathbb{E}_{b | c,\theta_{cb}}\big[\partial_{\theta_{cb}}\ln p(b|c,\theta_{cb})\big]\bigr] = 0,
\label{eq:cross_term_zero}
\end{eqnarray}
since the conditional expectation of a score is zero. Now set $\mathbf{u}=(1,1)^T$ so that $\theta_{ab}=\mathbf{u}^T\boldsymbol{\theta}$. By {\bf Theorem~\ref{thm:CRB_linear}}, the following holds:
\begin{eqnarray}
{\rm Var}(\hat{\theta}_{ab}) \ge \mathbf{u}^T\mathbf{F}_{(C,B)}(\boldsymbol{\theta})^{-1}\mathbf{u} = \frac{1}{F_{ac}(\theta_{ac})} + \frac{1}{F_{cb}(\boldsymbol{\theta})}.
\label{eq:CRB_series_step}
\end{eqnarray}
Equivalently, the effective Fisher information for $\theta_{ab}$ contained in the \emph{joint} record $(C,B)$ is upper bounded by
\begin{eqnarray}
F_{ab}^{(C,B)}(\boldsymbol{\theta})
\;:=\; F_{(C,B)}^{(\mathbf{u})}(\boldsymbol{\theta})
\;=\;
\Big(\frac{1}{F_{ac}(\theta_{ac})}+\frac{1}{F_{cb}(\boldsymbol{\theta})}\Big)^{-1}.
\label{eq:effectiveFI_joint_series}
\end{eqnarray}

\medskip\noindent
{\bf Step 2: discarding $C$ cannot increase information about $\theta_{ab}$.}

\medskip
The endpoint record $B$ is obtained from $(C,B)$ by the $\boldsymbol{\theta}$-independent coarse-graining that simply forgets $C$. By {\bf Corollary~\ref{cor:monotone_effectiveFI}},
\begin{eqnarray}
F_{ab}^{(B)}(\boldsymbol{\theta}) \le F_{ab}^{(C,B)}(\boldsymbol{\theta}).
\label{eq:monotone_apply}
\end{eqnarray}
By taking the inverse of both sides and using Eq.~(\ref{eq:effectiveFI_joint_series}), we obtain
\begin{eqnarray}
\Big(F_{ab}^{(B)}(\boldsymbol{\theta})\Big)^{-1} \ge \frac{1}{F_{ac}(\theta_{ac})}+\frac{1}{F_{cb}(\boldsymbol{\theta})},
\label{eq:CFII_path_proved}
\end{eqnarray}
which is Eq.~(\ref{eq:CFII_path}). The proof is completed.
\end{proof}

The inequality in Eq.~(\ref{eq:CFII_path}) is the archetypal CFII: it converts a causal claim---``the process admits a classical intermediate node $C$ with modular parameters''---into a quantitative constraint on observable precision. Here, we remark:
\begin{remark}[Inverse Fisher information as information-resistance]
\label{rem:resistance}
Eq.~(\ref{eq:CFII_path}) admits a physical reading that will guide the rest of this study. We here define the information-resistance $R:=F^{-1}$. Then, Eq.~(\ref{eq:CFII_path}) states that along a classical causal path, the resistances must add:
\begin{eqnarray}
R_{ab}^{(B)} \ge R_{ac} + R_{cb}.
\label{eq:resistance_series}
\end{eqnarray}
In a classical causal model, the information must traverse the intermediate bottleneck $C$, and each module imposes an unavoidable resistance. Therefore, if an experiment exhibits an endpoint precision such that $R_{ab}^{(B)} < R_{ac} + R_{cb}$, the conclusion is unambiguous: the assumed classical causal-path model cannot reproduce the statistics. In this sense, a violation of Eq.~(\ref{eq:CFII_path}) is not a mere metrological curiosity; it is a certificate of causal-model impossibility.
\end{remark}

The causal-path CFII extends immediately to longer chains.
\begin{corollary}[CFII for an $N$-step causal chain]
\label{cor:CFII_chain}
Consider a chain $X_0 \rightarrow X_1\rightarrow \cdots \rightarrow X_N$ with modular parameters $\theta_{j-1,j}$ and factorization
\begin{eqnarray}
p(x_1,\ldots,x_N \,|\, x_0,\boldsymbol{\theta}) = \prod_{j=1}^N p(x_j \,|\, x_{j-1},\theta_{j-1,j}),
\label{eq:chain_factorization}
\end{eqnarray}
and define the additive causal parameter $\theta_{0N}=\sum_{j=1}^N\theta_{j-1,j}$. Then,
\begin{eqnarray}
\Big(F_{0N}^{(X_N)}(\boldsymbol{\theta})\Big)^{-1} \ge \sum_{j=1}^N \Big(F_{j-1,j}(\theta_{j-1,j})\Big)^{-1},
\label{eq:CFII_chain}
\end{eqnarray}
where $F_{0N}^{(X_N)}$ is the endpoint effective Fisher information for $\theta_{0N}$ from $X_N$ alone.
\end{corollary}

\begin{proof}---Apply {\bf Theorem~\ref{thm:CFII_path}} iteratively, or proceed by induction by grouping the last two links into a single effective module. Each application adds one more term of the form $F_{j-1,j}^{-1}$ on the right-hand side.
\end{proof}

\subsection{Series-parallel composition and the network view}\label{subsec:cfii_network}

The path inequality already reveals the essential mechanism: modularity forces the orthogonality of score contributions, and causal coarse-graining prevents the hidden information from being resurrected at the endpoint. However, the realistic causal models often combine modules not only in series but also in parallel (multiple conditionally independent branches providing information about the same causal parameter). The resulting picture is that Fisher information obeys composition laws reminiscent of electrical networks: Fisher informations add in parallel, while inverse Fisher informations add in series. This analogy is not decorative; it is the algebraic heart of how CFIIs scale to complex causal architectures.

We record the parallel law in a form suited for later use:
\begin{theorem}[Parallel composition: additivity of Fisher information]
\label{thm:parallel_additivity}
Let $Y_1$ and $Y_2$ be conditionally independent given $\theta$, i.e., $p(y_1,y_2 | \theta)=p(y_1 | \theta)p(y_2|\theta)$. Then, the Fisher information satisfies
\begin{eqnarray}
F_{(Y_1,Y_2)}(\theta) = F_{Y_1}(\theta) + F_{Y_2}(\theta).
\label{eq:parallel_additivity}
\end{eqnarray}
Moreover, any coarse-graining $Z=\Gamma(Y_1,Y_2)$ produced by a $\theta$-independent map obeys
\begin{eqnarray}
F_{Z}(\theta) \le F_{Y_1}(\theta) + F_{Y_2}(\theta).
\label{eq:parallel_additivity_coarse}
\end{eqnarray}
\end{theorem}

\begin{proof}---From $\ln{p(y_1,y_2 | \theta)}=\ln{p(y_1 | \theta)} + \ln{p(y_2 | \theta)}$, the score is the sum of independent score contributions. Thus, by taking the expectation of the square, we obtain Eq.~(\ref{eq:parallel_additivity}). The coarse-grained bound in Eq.~(\ref{eq:parallel_additivity_coarse}) is an immediate application of {\bf Theorem~\ref{thm:DPI_FIM}} in the one-parameter case.
\end{proof}

Taken together, {\bf Theorem~\ref{thm:CFII_path}} (series law) and {\bf Theorem~\ref{thm:parallel_additivity}} (parallel law) allow one to derive CFIIs for any causal architecture that can be built from series and parallel compositions of independent modules. In particular, one can assign each module an information-resistance $R=F^{-1}$ and reduce the causal graph exactly as one reduces a resistor network: series resistances add, parallel resistances combine harmonically. The endpoint Fisher information is then upper bounded by the effective conductance of this reduced network, and the corresponding endpoint resistance is lower bounded by the effective resistance. This network viewpoint will be exploited in later sections to construct families of testable inequalities tailored to specific causal hypotheses, and to interpret their violation as a collapse of classical causal explanation.

\suppnote{Violation = impossibility of the classical causal model (model falsification)}\label{sec:falsification}

The inequalities derived in Sec.~\ref{sec:cfii} look, at first glance, like technical bounds on Fisher information. Their true significance is far sharper. A CFII is not merely a limit on precision; it is a logical constraint imposed by a classical causal explanation. Once a causal model class is fixed, the inequality becomes a necessary condition for compatibility. Consequently, a violation is not an invitation to interpret; it is a mathematical verdict. Here, we formalize this verdict and explain why it is physically dramatic: a CFII violation means that the classical causal narrative literally cannot generate the observed statistics.

\subsection{From inequality to falsification: the logical structure of CFIIs}\label{subsec:logic_falsification}

Recall from {\bf Definition~\ref{def:compatibility}} that a classical causal model class ${\cal M}$ specifies which families of operational distributions $\{p(x | s,\theta)\}$ are admissible, by requiring that they arise from a DAG factorization (possibly with latent variables) together with the associated conditional-independence relations. {\bf Definition~\ref{def:CFII}} then introduced a CFII as an inequality
\begin{eqnarray}
{\cal G}\left(\left\{F_s(\theta)\right\}_{s \in {\cal S}_0}\right) \ge 0
\label{eq:CFII_form_falsification}
\end{eqnarray}
that holds for all $\theta$ and for every operational model compatible with ${\cal M}$.

To turn this into a falsification principle, we separate two objects. First, we fix a finite set of contexts ${\cal S}_0$ and consider the vector of Fisher informations
\begin{eqnarray}
\mathbf{F}(\theta) := \left(F_s(\theta)\right)_{s \in {\cal S}_0}.
\label{eq:FI_vector}
\end{eqnarray}
Second, we define the feasible region of FI-vectors admitted by the model class.
\begin{definition}[Feasible FI region]
\label{def:feasible_region}
For a fixed $\theta$, the feasible FI region of a model class ${\cal M}$ is the set
\begin{eqnarray}
\Omega_{\cal M}(\theta) := \left\{ \mathbf{F}(\theta) : \{p(x | s,\theta)\}_{s \in {\cal S}_0}\ \mbox{is compatible with}\ {\cal M} \right\}.
\label{eq:Omega_def}
\end{eqnarray}
\end{definition}
A CFII is precisely the statement that $\Omega_{\cal M}(\theta)$ is contained in the half-space defined by ${\cal G} \ge 0$.

\begin{theorem}[CFII violation implies causal-model impossibility]
\label{thm:falsification_general}
Fix a model class ${\cal M}$, a context set ${\cal S}_0$, and a function ${\cal G}$ such that Eq.~(\ref{eq:CFII_form_falsification}) holds for all operational models compatible with ${\cal M}$. Let an experiment produce operational distributions $\{p_{\rm exp}(x | s,\theta)\}_{s \in {\cal S}_0}$ with Fisher informations $\mathbf{F}_{\rm exp}(\theta)$. If for some $\theta$ one has
\begin{eqnarray}
{\cal G}\left(\mathbf{F}_{\rm exp}(\theta)\right) < 0,
\label{eq:G_negative}
\end{eqnarray}
then $\{p_{\rm exp}(x | s,\theta)\}_{s \in {\cal S}_0}$ is not compatible with ${\cal M}$. Equivalently, $\mathbf{F}_{\rm exp}(\theta) \notin \Omega_{\cal M}(\theta)$.
\end{theorem}

\begin{proof}---Assume, toward contradiction, that $\{p_{\rm exp}(x | s,\theta)\}_{s \in {\cal S}_0}$ is compatible with ${\cal M}$. Then, by the definition of CFII, Eq.~(\ref{eq:CFII_form_falsification}) must hold for this operational model at the same $\theta$. Hence, ${\cal G}(\mathbf{F}_{\rm exp}(\theta)) \ge 0$, contradicting Eq.~(\ref{eq:G_negative}). Therefore, the operational model cannot be compatible with ${\cal M}$, and the FI-vector lies outside the feasible region.
\end{proof}

{\bf Theorem~\ref{thm:falsification_general}} is logically simple, but its meaning is profound. The experimental task is to estimate FI (or effective FI) from observed statistics; the theoretical task is to derive CFIIs from explicit causal assumptions; the conclusion is model falsification.

\subsection{The causal-path witness: when information seems to appear from nowhere}\label{subsec:path_witness}

The abstract falsification principle becomes physically vivid in the archetypal case of a causal path $A \rightarrow C \rightarrow B$. Here, the classical claim is that the endpoint record $B$ can only depend on the past through the intermediate bottleneck $C$, and that parameter dependence is modular, as encoded in Eq.~(\ref{eq:causal_path_factorization}). {\bf Theorem~\ref{thm:CFII_path}} then states that the endpoint effective Fisher information for the additive causal parameter $\theta_{ab}=\theta_{ac}+\theta_{cb}$ obeys
\begin{eqnarray}
\left(F_{ab}^{(B)}(\boldsymbol{\theta})\right)^{-1} \ge \left(F_{ac}(\theta_{ac})\right)^{-1} + \left(F_{cb}(\boldsymbol{\theta})\right)^{-1}.
\label{eq:path_CFII_repeat}
\end{eqnarray}
This inequality is the mathematical form of a narrative: information must traverse $C$, and every traversal adds an irreducible information-resistance.

For experimental falsification, it is convenient to package Eq.~(\ref{eq:path_CFII_repeat}) into a single witness.
\begin{definition}[Causal-path violation statistic]
\label{def:V_path}
For a causal-path hypothesis with parameters $\boldsymbol{\theta}=(\theta_{ac},\theta_{cb})^T$, let us define
\begin{eqnarray}
V_{\rm path}(\boldsymbol{\theta}) := \left(F_{ab}^{(B)}(\boldsymbol{\theta})\right)^{-1} - \left(F_{ac}(\theta_{ac})\right)^{-1} - \left(F_{cb}(\boldsymbol{\theta})\right)^{-1}.
\label{eq:V_path_def}
\end{eqnarray}
\end{definition}
Then, we can construct the theorem:
\begin{theorem}[Causal-path falsification criterion]
\label{thm:path_falsification}
If an operational model is compatible with the classical causal-path assumptions of Sec.~\ref{subsec:cfii_path}, then
\begin{eqnarray}
V_{\rm path}(\boldsymbol{\theta})\ge 0
\label{eq:V_nonnegative}
\end{eqnarray}
for all admissible $\boldsymbol{\theta}$.
Therefore, if an experiment yields $V_{\rm path}(\boldsymbol{\theta}) < 0$ for some $\boldsymbol{\theta}$, then no classical causal-path model in that class can reproduce the observed statistics.
\end{theorem}

\begin{proof}---The nonnegativity in Eq.~(\ref{eq:V_nonnegative}) is algebraically equivalent to Eq.~(\ref{eq:path_CFII_repeat}), which holds by {\bf Theorem~\ref{thm:CFII_path}} under the causal-path assumptions. The contrapositive implication is an instance of {\bf Theorem~\ref{thm:falsification_general}}.
\end{proof}

At this point the physical implication can be stated without qualification. A negative value of $V_{\rm path}$ means that, at the endpoint $B$, one can estimate the total causal parameter more precisely than would be permitted if the parameter were truly accumulated through two independent causal modules connected in series. In the information-resistance language of {\bf Remark~\ref{rem:resistance}}, the experiment behaves as if the series connection had developed an ``active element'' that cancels the resistance. Classically, such behavior is forbidden because every causal module is a passive channel: it may preserve information, but it cannot conjure additional information about $\theta$ that was not already present in its input. A violation therefore confronts us with a dramatic conclusion: either the intermediate bottleneck $C$ does not exist as a classical variable, or the process does not decompose into independent causal modules, or some hidden influence bypasses the assumed causal constraints. In every case, the classical causal story collapses.

\subsection{What exactly is falsified?}\label{subsec:what_fails}

A central virtue of the causal interpretation is the precision about what is, and is not, concluded from a violation. A CFII is derived under the explicit assumptions: a DAG factorization, its implied conditional-independence relations, and a modular parameter dependence consistent with the causal decomposition. Hence, a violation falsifies the conjunction of these assumptions. However, it does not claim that the world is ``quantum'' in the abstract; it claims something sharper:
\begin{quote}
\emph{the observed statistics cannot be generated by any member of the specified classical causal model class.}
\end{quote}
This distinction is not semantic. The different notions of ``classicality'' correspond to different model classes ${\cal M}$, and thus to different feasible regions $\Omega_{\cal M}(\theta)$. A violation identifies which classical causal explanations are ruled out, while leaving open the possibility that a larger class of models (with weaker conditional-independence requirements, additional latent memory nodes, or explicit context dependence) may still explain the data. This is exactly why CFIIs are powerful scientific instruments: they do not merely quantify precision, they map which causal explanations remain tenable.

In the causal-path setting, for instance, $V_{\rm path} < 0$ rules out the existence of any intermediate classical bottleneck $C$ that simultaneously (i) renders $A$ and $B$ conditionally independent given $C$ and (ii) supports a modular decomposition of the parameter into $(\theta_{ac}, \theta_{cb})$ with additive causal parameter $\theta_{ab}$. In dynamical quantum experiments, this naturally resonates with the failure of a trajectory-like description by a fixed series of intermediate states, where the true evolution may involve coherence across segments that cannot be represented by a single classical node $C$. In such a case, a coherent superposition, for example, a qubit state $\ket{\psi}$ evolving under a Hamiltonian $\hat{H}$, can exhibit statistics that invalidate the classical path decomposition precisely, because the causal bottleneck cannot capture phase-sensitive interference effects. The inequality detects this failure at the level of Fisher information: the putative resistance of a classical bottleneck is overcome by interference, and the only consistent conclusion is that the classical bottleneck model is impossible. We will explicitly show this example later.

\subsection{Statistical certification from finite data}\label{subsec:finite_data}

{\bf Theorems~\ref{thm:falsification_general}} and {\bf \ref{thm:path_falsification}} are the statements about true Fisher informations, i.e., about the underlying distributions. In experiment, Fisher information must be estimated from a finite set of samples. Here, we present a statistically principled route to certify the CFII violations with controlled confidence.

We assume that, for each context $s \in {\cal S}_0$ and a fixed parameter value $\theta$ (e.g., a calibrated or controlled setting), we can evaluate the score
\begin{eqnarray}
\ell_s(x | \theta) := \frac{\partial}{\partial\theta} \ln{p(x | s,\theta)}.
\end{eqnarray}
This is standard in the metrological models where the encoding $p(x | s,\theta)$ is known up to the value of $\theta$, and it can also be achieved operationally by local calibration. Given i.i.d. outcomes $\{x_{s,i}\}_{i=1}^{N_s}$ sampled from $p(x | s,\theta)$, let us consider the estimator
\begin{eqnarray}
\tilde{F}_s(\theta) := \frac{1}{N_s}\sum_{i=1}^{N_s}\ell_s(x_{s,i} | \theta)^2.
\label{eq:Fhat_score_sq}
\end{eqnarray}
By construction, $\mathbb{E}[\tilde{F}_s(\theta)]=F_s(\theta)$ whenever the score has finite second moment.
\begin{theorem}[Consistency and asymptotic normality of $\tilde{F}_s(\theta)$]
\label{thm:FI_estimation_CLT}
Assume $\mathbb{E}[\ell_s(X|\theta)^4]<\infty$ for each $s\in{\cal S}_0$, where $X\sim p(x|s,\theta)$. Then:
\begin{eqnarray}
\tilde{F}_s(\theta) \to F_s(\theta)
\label{eq:FI_consistency}
\end{eqnarray}
in probability as $N_s\to\infty$, and moreover
\begin{eqnarray}
\sqrt{N_s}\left(\tilde{F}_s(\theta) - F_s(\theta)\right) \Rightarrow {\cal N}\left(0,\sigma_s^2(\theta)\right),
\label{eq:FI_CLT}
\end{eqnarray}
where
\begin{eqnarray}
\sigma_s^2(\theta) := {\rm Var}\left(\ell_s(X | \theta)^2\right).
\label{eq:sigma_def}
\end{eqnarray}
\end{theorem}

\begin{proof}---Define the i.i.d. random variables $Y_{s,i} := \ell_s(x_{s,i} | \theta)^2$. Then, $\tilde{F}_s(\theta)$ is the sample mean of $Y_{s,i}$:
\begin{eqnarray}
\tilde{F}_s(\theta) = \frac{1}{N_s}\sum_{i=1}^{N_s}Y_{s,i}.
\label{eq:Fhat_as_mean}
\end{eqnarray}
Since $\mathbb{E}[Y_{s,i}]=\mathbb{E}[\ell_s(X | \theta)^2]=F_s(\theta)$ and $\mathbb{E}[Y_{s,i}^2]=\mathbb{E}[\ell_s(X | \theta)^4] < \infty$, the weak law of large numbers implies Eq.~(\ref{eq:FI_consistency}). The central limit theorem for i.i.d. variables with finite variance yields
\begin{eqnarray}
\sqrt{N_s}\left(\tilde{F}_s(\theta) - F_s(\theta)\right) = \frac{1}{\sqrt{N_s}}\sum_{i=1}^{N_s}\left(Y_{s,i} - \mathbb{E}[Y_{s,i}]\right) \Rightarrow {\cal N}\left(0,{\rm Var}(Y_{s,1})\right),
\label{eq:CLT_applied}
\end{eqnarray}
which is Eq.~(\ref{eq:FI_CLT}) with Eq.~(\ref{eq:sigma_def}).
\end{proof}

{\bf Theorem~\ref{thm:FI_estimation_CLT}} provides an operational route to confidence intervals for each $F_s(\theta)$, and, crucially, for any smooth function ${\cal G}$ of the FI-vector. We now translate it to a violation test. Let $\tilde{\mathbf{F}}(\theta) := (\tilde{F}_s(\theta))_{s \in {\cal S}_0}$, and define the plug-in estimator
\begin{eqnarray}
\tilde{{\cal G}}(\theta) := {\cal G}\left(\tilde{\mathbf{F}}(\theta)\right).
\label{eq:Ghat_def}
\end{eqnarray}
Assume that outcomes are sampled independently across contexts, so that the estimators $\{\tilde{F}_s(\theta)\}$ are asymptotically independent.
\begin{theorem}[Asymptotic certification of CFII violation]
\label{thm:delta_method_general}
Assume the conditions of {\bf Theorem~\ref{thm:FI_estimation_CLT}} for all $s \in {\cal S}_0$, and assume ${\cal G}$ is continuously differentiable on an open neighborhood of $\mathbf{F}(\theta)$. Let $N_s\to\infty$ with fixed proportions $q_s := N_s/N$ where $N=\sum_{s \in {\cal S}_0} N_s$. Then,
\begin{eqnarray}
\sqrt{N}\left(\tilde{{\cal G}}(\theta) - {\cal G}(\mathbf{F}(\theta))\right) \Rightarrow {\cal N}\left(0,\tau^2(\theta)\right),
\label{eq:Ghat_CLT}
\end{eqnarray}
where
\begin{eqnarray}
\tau^2(\theta) = \sum_{s \in {\cal S}_0}\frac{1}{q_s}\left(\frac{\partial{\cal G}}{\partial F_s}\Big|_{\mathbf{F}(\theta)}\right)^2\sigma_s^2(\theta).
\label{eq:tau_def}
\end{eqnarray}
\end{theorem}

\begin{proof}---From {\bf Theorem~\ref{thm:FI_estimation_CLT}}, the vector
\begin{eqnarray}
\sqrt{N}\left(\tilde{\mathbf{F}}(\theta) - \mathbf{F}(\theta)\right) = \left(\sqrt{N}\left(\tilde{F}_s(\theta) - F_s(\theta)\right)\right)_{s \in {\cal S}_0}
\label{eq:vector_scaling}
\end{eqnarray}
has asymptotically normal components with variances scaled by $q_s$:
\begin{eqnarray}
\sqrt{N}\left(\tilde{F}_s(\theta) - F_s(\theta)\right) = \sqrt{\frac{N}{N_s}}\sqrt{N_s}\left(\tilde{F}_s(\theta) - F_s(\theta)\right) \Rightarrow {\cal N}\left(0,\frac{\sigma_s^2(\theta)}{q_s}\right).
\label{eq:component_scaling}
\end{eqnarray}
The independence across contexts yields an asymptotically normal vector with diagonal covariance matrix whose $s$th diagonal entry is $\sigma_s^2(\theta)/q_s$.

Since ${\cal G}$ is continuously differentiable, a first-order Taylor expansion gives
\begin{eqnarray}
\tilde{{\cal G}}(\theta) - {\cal G}(\mathbf{F}(\theta)) = \nabla{\cal G}(\mathbf{F}(\theta))^T\left(\tilde{\mathbf{F}}(\theta) - \mathbf{F}(\theta)\right) + r_N,
\label{eq:taylor_G}
\end{eqnarray}
where the remainder $r_N$ satisfies $r_N=o_p(\|\tilde{\mathbf{F}}-\mathbf{F}\|)$. Multiplying by $\sqrt{N}$ and using the multivariate central limit theorem together with Slutsky's theorem, we obtain Eq.~(\ref{eq:Ghat_CLT}). The variance is the quadratic form of the gradient with the asymptotic covariance matrix, which reduces to Eq.~(\ref{eq:tau_def}), because the covariance is diagonal.
\end{proof}

{\bf Theorem~\ref{thm:delta_method_general}} makes the falsification program experimentally concrete. Under the null hypothesis that the data are compatible with ${\cal M}$, one must have ${\cal G}(\mathbf{F}(\theta)) \ge 0$. If an experiment yields $\tilde{{\cal G}}(\theta)$ significantly below zero compared to its standard error $\tau(\theta)/\sqrt{N}$, then the classical causal model class is rejected with controlled significance. In other words, one can attach to a CFII violation the same statistical seriousness as to any hypothesis test in physics.

For the causal-path witness, ${\cal G}$ can be chosen as ${\cal G} = V_{\rm path}$ in Eq.~(\ref{eq:V_path_def}), and the derivative factors entering Eq.~(\ref{eq:tau_def}) take a simple form. Writing $F_1 = F_{ab}^{(B)}$, $F_2 = F_{ac}$, and $F_3 = F_{cb}$, one has
\begin{eqnarray}
\frac{\partial V_{\rm path}}{\partial F_1}=-\frac{1}{F_1^2},\quad
\frac{\partial V_{\rm path}}{\partial F_2}=\frac{1}{F_2^2},\quad
\frac{\partial V_{\rm path}}{\partial F_3}=\frac{1}{F_3^2}.
\label{eq:grad_V}
\end{eqnarray}
Substituting Eq.~(\ref{eq:grad_V}) into Eq.~(\ref{eq:tau_def}) yields an explicit asymptotic variance for $\hat{V}_{\rm path}$. Thus, the inequality, originally a statement about causal structure, becomes an experimentally testable boundary with quantitative confidence.

\subsection{Why this matters}\label{subsec:why_matters}

What makes CFII-based falsification distinctive is that the rejected object is not a phenomenological fit but a ``causal explanation'' of how information about $\theta$ may propagate through an experiment. The content of a violation is therefore structural: it tells us that no classical model with the assumed conditional independences can account for the data, even if that model is allowed arbitrary internal parameters and arbitrary hidden variables. When a CFII is violated, the classical causal narrative does not merely require refinement; it ceases to exist. This is precisely the point at which the story becomes metrological. Once a classical causal model collapses, the ``extra'' information that defeats the classical resistance law cannot be treated as an accident. It is the signature of a nonclassical resource that bypasses the classical bottleneck.

We will show, in later sections, how such violations can be interpreted not only as witnesses but as resources, and how they encompass as special instances contextuality-based protocols in which the failure of a noncontextual causal condition is converted into enhanced precision. The dramatic moral is that, in metrology, impossibility results are not dead ends. They are signposts that point to where Nature stores her most useful nonclassical advantages.

\suppnote{From witness to resource: ``causal-model collapse'' as a metrological resource}\label{sec:resource}

Sections~\ref{sec:cfii} and \ref{sec:falsification} established the conceptual backbone of this work: a CFII is a necessary condition for the existence of a classical causal explanation, and its violation is a rigorous falsification of that explanation. At this stage, however, we ask a practical question: if a CFII violation certifies that a classical causal story fails, what does that buy us in metrology? Is it merely a diagnostic statement about ``nonclassicality,'' or does it translate into a concrete improvement of estimation performance?

Here, we answer with an engineering-level message. A CFII violation is not only a witness; it is a resource certificate. It certifies that the experiment possesses an information flow that cannot be routed through the assumed classical bottlenecks, and this surplus information is directly convertible into smaller estimation error. The conversion is quantitative, and it is operational: it does not require philosophical interpretation, only a well-defined estimator. The same inequality that collapses the classical causal story simultaneously exposes a metrological advantage that the classical story would have declared impossible.

\subsection{Classical causal precision bounds from CFIIs}\label{subsec:classical_precision_bounds}

A CFII is an inequality among Fisher informations. Since Fisher information is the curvature controlling Cram\'er--Rao-type bounds, each CFII immediately induces a precision bound that holds whenever the classical causal model class is valid. This is the sense in which CFIIs serve as causal benchmarks for metrology. We formulate the statement in a general form first, and then specialize it to the causal-path inequality. 
\begin{theorem}[Classical causal precision bound induced by a CFII]
\label{thm:classical_precision_bound}
Fix a model class ${\cal M}$ and a finite set of contexts ${\cal S}_0$. Assume that a CFII of the form
\begin{eqnarray}
{\cal G}\left(\left\{F_s(\theta)\right\}_{s\in{\cal S}_0}\right) \ge 0
\label{eq:CFII_resource_general}
\end{eqnarray}
holds for all operational models compatible with ${\cal M}$. Let an estimator $\hat{\theta}$ be constructed from $N$ i.i.d. repetitions of an experiment that uses only contexts in ${\cal S}_0$. Then, for any operational model compatible with ${\cal M}$ and for any unbiased $\hat{\theta}$, the mean-square error satisfies a model-dependent lower bound
\begin{eqnarray}
{\rm Var}(\hat{\theta}) \ge \frac{1}{N}\frac{1}{F_{\rm cl}(\theta)},
\label{eq:resource_bound_general}
\end{eqnarray}
where $F_{\rm cl}(\theta)$ is any quantity that upper bounds the achievable per-sample Fisher information under the constraint in Eq.~(\ref{eq:CFII_resource_general}). In particular, if the CFII can be algebraically rearranged into
\begin{eqnarray}
F_{\rm eff}(\theta) \le F_{\rm cl}(\theta),
\label{eq:Feff_le_Fcl}
\end{eqnarray}
for a suitable effective Fisher information $F_{\rm eff}(\theta)$ relevant to the estimation strategy, then Eq.~(\ref{eq:resource_bound_general}) follows with $F_{\rm cl}$.
\end{theorem}

\begin{proof}
For $N$ i.i.d. repetitions, the Fisher information is additive, hence the total Fisher information equals $N$ times the per-sample Fisher information ({\bf Theorem~\ref{thm:FI_additivity}} in Sec.~\ref{sec:framework}). For any unbiased estimator, the Cram\'er--Rao bound gives
\begin{eqnarray}
{\rm Var}(\hat{\theta}) \ge \frac{1}{N}\frac{1}{F_{\rm eff}(\theta)},
\label{eq:CRB_per_sample_resource}
\end{eqnarray}
where $F_{\rm eff}$ denotes the per-sample Fisher information appropriate to the statistical model used in the protocol. If the operational model is compatible with ${\cal M}$, then its Fisher informations must satisfy the CFII as in Eq.~(\ref{eq:CFII_resource_general}). If Eq.~(\ref{eq:CFII_resource_general}) implies an upper bound in Eq.~(\ref{eq:Feff_le_Fcl}) on $F_{\rm eff}$, then by substituting Eq.~(\ref{eq:Feff_le_Fcl}) into Eq.~(\ref{eq:CRB_per_sample_resource}), we obtain Eq.~(\ref{eq:resource_bound_general}), completing the proof.
\end{proof}

{\bf Theorem~\ref{thm:classical_precision_bound}} is conceptually simple but practically sharp: it tells us that a causal model class defines a \emph{precision frontier}. Any protocol consistent with the model must live above that frontier. The moment an experiment falls below it, the classical model is not merely disfavored; it is logically incompatible with the achieved precision.

The causal-path case makes this concrete. Under the classical causal-path hypothesis of Sec.~\ref{subsec:cfii_path}, {\bf Theorem~\ref{thm:CFII_path}} states that
\begin{eqnarray}
\left(F_{ab}^{(B)}(\boldsymbol{\theta})\right)^{-1} \ge \left(F_{ac}(\theta_{ac})\right)^{-1} + \left(F_{cb}(\boldsymbol{\theta})\right)^{-1}.
\label{eq:path_CFII_resource_repeat}
\end{eqnarray}
Here, we define the classical benchmark Fisher information by the harmonic composition
\begin{eqnarray}
F_{\rm cl}^{\rm(path)}(\boldsymbol{\theta}) := \left( \left(F_{ac}(\theta_{ac})\right)^{-1} + \left(F_{cb}(\boldsymbol{\theta})\right)^{-1} \right)^{-1}.
\label{eq:Fcl_path_def}
\end{eqnarray}
Then, Eq.~(\ref{eq:path_CFII_resource_repeat}) is precisely the upper bound
\begin{eqnarray}
F_{ab}^{(B)}(\boldsymbol{\theta}) \le F_{\rm cl}^{\rm(path)}(\boldsymbol{\theta})
\label{eq:Feff_path_le}
\end{eqnarray}
that {\bf Theorem~\ref{thm:classical_precision_bound}} requires. Thus, if a process genuinely decomposes into two independent causal modules in series, the achievable per-sample Fisher information about the total parameter can never exceed the harmonic mean in Eq.~(\ref{eq:Fcl_path_def}). This is the classical ``series penalty'' expressed in metrological language.

\subsection{Violation-certified gain: witness-to-resource transition}\label{subsec:violation_certified_gain}

We now convert the previous statement into a resource claim. Recall the causal-path violation statistic
\begin{eqnarray}
V_{\rm path}(\boldsymbol{\theta}) = \left(F_{ab}^{(B)}(\boldsymbol{\theta})\right)^{-1} - \left(F_{ac}(\theta_{ac})\right)^{-1} - \left(F_{cb}(\boldsymbol{\theta})\right)^{-1}.
\label{eq:V_path_resource_repeat}
\end{eqnarray}
The inequality in Eq.~(\ref{eq:path_CFII_resource_repeat}) is equivalent to $V_{\rm path}(\boldsymbol{\theta}) \ge 0$. Here, a negative value $V_{\rm path} < 0$ can be interpreted as the impossibility of the assumed classical causal-path model. Now we show that the same quantity also certifies a quantitative metrological gain.

To state the gain in the most operational form, we consider $N$ i.i.d. repetitions of the endpoint experiment producing outcomes $B_1, \ldots, B_N$ distributed as $p(b | \boldsymbol{\theta})$, and let $\hat{\theta}_{ab}$ be an estimator of $\theta_{ab}=\theta_{ac}+\theta_{cb}$ constructed from $\{B_i\}_{i=1}^N$. Here we provide the following theorem:
\begin{theorem}[Violation-certified precision enhancement]
\label{thm:gain_from_violation}
Assume that $F_{ab}^{(B)}(\boldsymbol{\theta})$ is the per-sample effective Fisher information for estimating $\theta_{ab}$ from endpoint data $B$ in the model under consideration. Define the classical benchmark $F_{\rm cl}^{\rm(path)}(\boldsymbol{\theta})$ by Eq.~(\ref{eq:Fcl_path_def}). If $V_{\rm path}(\boldsymbol{\theta}) < 0$, then
\begin{eqnarray}
F_{ab}^{(B)}(\boldsymbol{\theta}) > F_{\rm cl}^{\rm(path)}(\boldsymbol{\theta}),
\label{eq:gain_strict_FI}
\end{eqnarray}
and consequently the Cram\'er--Rao bound for any unbiased estimator satisfies
\begin{eqnarray}
\frac{1}{N}\frac{1}{F_{ab}^{(B)}(\boldsymbol{\theta})} < \frac{1}{N}\frac{1}{F_{\rm cl}^{\rm(path)}(\boldsymbol{\theta})}.
\label{eq:gain_strict_CRB}
\end{eqnarray}
Moreover, defining the resistance quantities
\begin{eqnarray}
R_{\rm ab}(\boldsymbol{\theta}) = \left(F_{ab}^{(B)}(\boldsymbol{\theta})\right)^{-1},
\quad
R_{\rm cl}(\boldsymbol{\theta}) = \left(F_{\rm cl}^{\rm(path)}(\boldsymbol{\theta})\right)^{-1},
\label{eq:R_defs}
\end{eqnarray}
the violation directly yields the certified improvement factor
\begin{eqnarray}
\frac{R_{\rm cl}(\boldsymbol{\theta})}{R_{\rm ab}(\boldsymbol{\theta})} = \frac{R_{\rm cl}(\boldsymbol{\theta})}{R_{\rm cl}(\boldsymbol{\theta}) + V_{\rm path}(\boldsymbol{\theta})} >1.
\label{eq:gain_factor_resistance}
\end{eqnarray}
\end{theorem}

\begin{proof}---If $V_{\rm path}(\boldsymbol{\theta}) < 0$, then by Eq.~(\ref{eq:V_path_resource_repeat}), we have
\begin{eqnarray}
\left(F_{ab}^{(B)}(\boldsymbol{\theta})\right)^{-1} < \left(F_{\rm cl}^{\rm(path)}(\boldsymbol{\theta})\right)^{-1}.
\label{eq:inverse_gain}
\end{eqnarray}
Taking reciprocals preserves the strict inequality because both sides are positive, proving Eq.~(\ref{eq:gain_strict_FI}). By substituting these informations into the Cram\'er--Rao bound, we obtain Eq.~(\ref{eq:gain_strict_CRB}). Finally, Eq.~(\ref{eq:gain_factor_resistance}) follows from the definition $V_{\rm path}=R_{\rm ab}-R_{\rm cl}$.
\end{proof}

{\bf Theorem~\ref{thm:gain_from_violation}} is the operational heart of the ``witness-to-resource'' transition. It says that the same numerical quantity that falsifies a classical causal explanation also provides a quantitative \emph{performance guarantee}: the endpoint experiment achieves a precision that no member of the classical causal-path model class can match, even in principle.

\begin{remark}[Negative $V_{\rm path}$ is a metrology enhancement resource]
The result is dramatic when expressed in the resistance language. In a classical causal narrative, resistances add. The total estimation error behaves as if the parameter had to pass through two passive resistors in series. A negative $V_{\rm path}$ means that the experiment contains an element that cancels part of this series resistance. In classical circuitry that would require an active device. Here the ``active device'' is not an additional energy source; it is the collapse of a classical causal description. Metrologically, the collapse is a resource because it converts structural impossibility into reduced error.
\end{remark}

\subsection{Achievability: efficient estimators harvest the resource}\label{subsec:achievability_mle}

A resource certificate is meaningful only if the improved bound is not a mirage. {\bf Theorem~\ref{thm:gain_from_violation}} compares two Cram\'er--Rao bounds, but does not yet guarantee that the smaller bound can be approached by an explicit estimator. We therefore record a standard but essential fact: under regularity conditions, maximum likelihood estimation asymptotically achieves the Fisher information bound. This is the mechanism by which causal-model collapse becomes practically harvestable.

\begin{theorem}[Asymptotic efficiency of maximum likelihood estimation]
\label{thm:MLE_efficiency}
Consider $X_1,\ldots,X_N$, the i.i.d. samples drawn from a regular one-parameter family $p(x|\theta)$ with true parameter value $\theta$. Let
\begin{eqnarray}
\ell_N(\theta)=\sum_{i=1}^N \ln p(X_i | \theta)
\label{eq:loglik_N}
\end{eqnarray}
be the log-likelihood and define the maximum likelihood estimator $\hat{\theta}_{\rm ML}$ by the score equation
\begin{eqnarray}
\frac{\partial}{\partial\theta}\ell_N(\hat{\theta}_{\rm ML})=0.
\label{eq:score_eq}
\end{eqnarray}
Assume that $\hat{\theta}_{\rm ML} \to \theta$ in probability and that the following identities hold with finite moments:
\begin{eqnarray}
\mathbb{E}\left[\frac{\partial}{\partial\theta}\ln p(X | \theta)\right]=0,
\quad
\mathbb{E}\left[\left(\frac{\partial}{\partial\theta}\ln{p(X | \theta)}\right)^2\right]=F(\theta),
\quad \text{and} \quad
\mathbb{E}\left[-\frac{\partial^2}{\partial\theta^2}\ln p(X | \theta)\right]=F(\theta),
\label{eq:regularity_identities}
\end{eqnarray}
where $F(\theta)$ is the per-sample Fisher information. Then,
\begin{eqnarray}
\sqrt{N}\left(\hat{\theta}_{\rm ML} - \theta\right) \Rightarrow {\cal N}\left(0,\frac{1}{F(\theta)}\right),
\label{eq:MLE_asymp_normal}
\end{eqnarray}
and consequently,
\begin{eqnarray}
{\rm Var}(\hat{\theta}_{\rm ML}) = \frac{1}{N}\frac{1}{F(\theta)} + o\left(\frac{1}{N}\right).
\label{eq:MLE_var}
\end{eqnarray}
\end{theorem}

\begin{proof}---Define the score and observed information
\begin{eqnarray}
S_N(\theta) &=& \frac{\partial}{\partial\theta}\ell_N(\theta) = \sum_{i=1}^N \frac{\partial}{\partial\theta}\ln p(X_i | \theta), \nonumber \\
J_N(\theta) &=& -\frac{\partial^2}{\partial\theta^2}\ell_N(\theta) = \sum_{i=1}^N \left(-\frac{\partial^2}{\partial\theta^2}\ln p(X_i | \theta)\right).
\label{eq:obs_info_JN}
\end{eqnarray}
The defining Eq.~(\ref{eq:score_eq}) reads $S_N(\hat{\theta}_{\rm ML})=0$. By Taylor's theorem,
\begin{eqnarray}
0 = S_N(\hat{\theta}_{\rm ML}) =S_N(\theta)-J_N(\tilde{\theta}_N)\left(\hat{\theta}_{\rm ML} - \theta\right),
\label{eq:taylor_score}
\end{eqnarray}
where $\tilde{\theta}_N$ lies between $\hat{\theta}_{\rm ML}$ and $\theta$. By rearranging, we have
\begin{eqnarray}
\sqrt{N}\left(\hat{\theta}_{\rm ML} - \theta\right) = \left(\frac{1}{N}J_N(\tilde{\theta}_N)\right)^{-1} \left(\frac{1}{\sqrt{N}}S_N(\theta)\right).
\label{eq:MLE_decomposition}
\end{eqnarray}

We then analyze the two factors. First, by Eq.~(\ref{eq:regularity_identities}) and the central limit theorem, we can find
\begin{eqnarray}
\frac{1}{\sqrt{N}}S_N(\theta) = \frac{1}{\sqrt{N}}\sum_{i=1}^N s(X_i | \theta) \Rightarrow {\cal N}\left(0,F(\theta)\right),
\label{eq:score_CLT}
\end{eqnarray}
where $s(X|\theta)=\partial_\theta \ln{p(X | \theta)}$.

Second, by the law of large numbers and Eq.~(\ref{eq:regularity_identities}),
\begin{eqnarray}
\frac{1}{N}J_N(\theta) = \frac{1}{N}\sum_{i=1}^N j(X_i | \theta) \to F(\theta)
\label{eq:info_LL}
\end{eqnarray}
in probability, where $j(X | \theta)=-\partial_\theta^2\ln p(X | \theta)$. Since $\hat{\theta}_{\rm ML} \to \theta$ in probability by assumption, also $\tilde{\theta}_N \to \theta$ in probability. Under continuity of $j(X | \theta)$ in $\theta$ and dominated convergence sufficient to exchange limits and expectations, Eq.~(\ref{eq:info_LL}) extends to
\begin{eqnarray}
\frac{1}{N}J_N(\tilde{\theta}_N) \to F(\theta)
\label{eq:info_tilde}
\end{eqnarray}
in probability. Therefore, by Slutsky's theorem,
\begin{eqnarray}
\left(\frac{1}{N}J_N(\tilde{\theta}_N)\right)^{-1} \to \frac{1}{F(\theta)}
\label{eq:info_inverse_limit}
\end{eqnarray}
in probability.

Finally, by combining Eq.~(\ref{eq:score_CLT}), Eq.~(\ref{eq:info_inverse_limit}), and Eq.~(\ref{eq:MLE_decomposition}) via Slutsky's theorem, we can attain
\begin{eqnarray}
\sqrt{N}\left(\hat{\theta}_{\rm ML}-\theta\right) \Rightarrow {\cal N}\left(0,\frac{1}{F(\theta)}\right),
\label{eq:MLE_limit_final}
\end{eqnarray}
which is Eq.~(\ref{eq:MLE_asymp_normal}). The variance statement Eq.~(\ref{eq:MLE_var}) follows.
\end{proof}

{\bf Theorem~\ref{thm:MLE_efficiency}} supplies the missing bridge from principle to practice. If an experiment exhibits a violation-certified increase in Fisher information, then an explicit estimator (maximum likelihood, or any asymptotically efficient alternative) can harvest that increase into reduced error. The point is practical: ``causal-model collapse'' certifies that there exists a \emph{working estimation procedure} whose asymptotic performance beats the classical causal benchmark.

\subsection{Mechanism: Fisher-information synergy from correlated score contributions}\label{subsec:synergy_mechanism}

We established that the CFII violation implies a gain and that the gain is achievable. We now expose the structural mechanism behind the gain in a way that is independent of any particular physical platform. The mechanism is the emergence of ``synergy'' between the causal modules, visible as off-diagonal terms of the Fisher information matrix. These off-diagonal terms are precisely what classical modularity forbids, as made explicit by the vanishing cross term in Eq.~(\ref{eq:cross_term_zero}).

Consider a two-parameter description $\boldsymbol{\theta}=(\theta_1,\theta_2)^T$ of a process and an observation $Y$ distributed as $p(y | \theta_1,\theta_2)$. Let the Fisher information matrix be
\begin{eqnarray}
\mathbf{F}_Y(\boldsymbol{\theta}) =
\begin{pmatrix}
F_1(\boldsymbol{\theta}) & J(\boldsymbol{\theta}) \\
J(\boldsymbol{\theta}) & F_2(\boldsymbol{\theta})
\end{pmatrix},
\label{eq:FIM_2x2}
\end{eqnarray}
where
\begin{eqnarray}
F_1(\boldsymbol{\theta}) &=& \mathbb{E}\bigl[\left(\partial_{\theta_1} \ln{p(Y | \theta_1,\theta_2)}\right)^2\bigr], \nonumber \\
F_2(\boldsymbol{\theta}) &=& \mathbb{E}\bigl[\left(\partial_{\theta_2} \ln{p(Y | \theta_1,\theta_2)}\right)^2\bigr], \nonumber \\
J(\boldsymbol{\theta}) &=& \mathbb{E}\bigl[\left(\partial_{\theta_1} \ln{p(Y | \theta_1,\theta_2)}\right)\left(\partial_{\theta_2} \ln{p(Y | \theta_1,\theta_2)}\right)\bigr].
\label{eq:defs_FJ}
\end{eqnarray}
The quantity $J(\boldsymbol{\theta})$ is the covariance between the score contributions of the two parameters. In a classical causal-path model with modular parameter dependence, this covariance vanishes by construction (as shown explicitly in Eq.~(\ref{eq:cross_term_zero})). When it does not vanish, the two modules are not independent in the sense required by the classical causal decomposition.

We now compute the effective Fisher information for the additive parameter
\begin{eqnarray}
\Theta=\theta_1+\theta_2.
\label{eq:Theta_sum}
\end{eqnarray}
Let $\mathbf{u}=(1,1)^T$. By the definition of Eq.~(\ref{eq:effective_FI_def}),
\begin{eqnarray}
F_Y^{(\mathbf{u})}(\boldsymbol{\theta}) = \left(\mathbf{u}^T\mathbf{F}_Y(\boldsymbol{\theta})^{-1}\mathbf{u}\right)^{-1}.
\label{eq:Feff_sum_def}
\end{eqnarray}
We then provide the theorem:
\begin{theorem}[Synergy formula and the condition for beating the series law]
\label{thm:synergy_formula}
Assume $\mathbf{F}_Y(\boldsymbol{\theta})$ in Eq.~(\ref{eq:FIM_2x2}) is positive definite, i.e., $F_1F_2 - J^2 > 0$. Then the effective Fisher information for $\Theta=\theta_1+\theta_2$ is
\begin{eqnarray}
F_Y^{(\mathbf{u})}(\boldsymbol{\theta}) =\frac{F_1 F_2 - J^2}{F_1 + F_2 - 2J}.
\label{eq:Feff_sum_closed}
\end{eqnarray}
Moreover, comparing to the classical series benchmark
\begin{eqnarray}
F_{\rm series}(\boldsymbol{\theta}) = \left(\frac{1}{F_1} + \frac{1}{F_2}\right)^{-1} = \frac{F_1F_2}{F_1+F_2},
\label{eq:Fseries_def}
\end{eqnarray}
one has
\begin{eqnarray}
F_Y^{(\mathbf{u})}(\boldsymbol{\theta}) > F_{\rm series}(\boldsymbol{\theta})
\label{eq:synergy_advantage_condition}
\end{eqnarray}
if and only if
\begin{eqnarray}
0 < J(\boldsymbol{\theta}) < \frac{2F_1(\boldsymbol{\theta}) F_2(\boldsymbol{\theta})} {F_1(\boldsymbol{\theta}) + F_2(\boldsymbol{\theta})}.
\label{eq:J_range}
\end{eqnarray}
Finally, for fixed $F_1$ and $F_2$, the supremum of $F_Y^{(\mathbf{u})}$ over $J$ under $F_1 F_2 - J^2 > 0$ is
\begin{eqnarray}
\sup_J F_Y^{(\mathbf{u})}(\boldsymbol{\theta}) = \min\left\{F_1(\boldsymbol{\theta}), F_2(\boldsymbol{\theta})\right\}.
\label{eq:max_Feff_min}
\end{eqnarray}
\end{theorem}

\begin{proof}---We first compute $\mathbf{F}_Y^{-1}$. For a $2\times 2$ matrix in Eq.~(\ref{eq:FIM_2x2}) with the determinant $F_1 F_2 - J^2$, one has
\begin{eqnarray}
\mathbf{F}_Y(\boldsymbol{\theta})^{-1} = \frac{1}{F_1F_2-J^2}
\begin{pmatrix}
F_2 & -J \\
-J & F_1
\end{pmatrix}.
\label{eq:FIM_inverse_2x2}
\end{eqnarray}
Therefore,
\begin{eqnarray}
\mathbf{u}^T\mathbf{F}_Y^{-1}\mathbf{u} = \frac{1}{F_1F_2-J^2} \left(F_1+F_2-2J\right),
\label{eq:uFu}
\end{eqnarray} and Eq.~(\ref{eq:Feff_sum_closed}) follows by inversion.

To compare with the series benchmark in Eq.~(\ref{eq:Fseries_def}), we define the resistance quantities
\begin{eqnarray}
R_{\rm eff}(\boldsymbol{\theta}) &=& \left(F_Y^{(\mathbf{u})}(\boldsymbol{\theta})\right)^{-1} = \frac{F_1+F_2-2J}{F_1F_2-J^2}, \nonumber \\
R_{\rm series}(\boldsymbol{\theta}) &=& \left(F_{\rm series}(\boldsymbol{\theta})\right)^{-1}=\frac{F_1+F_2}{F_1F_2}.
\label{eq:R_eff_series}
\end{eqnarray}
The inequality $F_Y^{(\mathbf{u})} > F_{\rm series}$ is equivalent to $R_{\rm eff}<R_{\rm series}$. A direct algebraic manipulation yields
\begin{eqnarray}
R_{\rm series}-R_{\rm eff} = \frac{J\left(J(F_1 + F_2) - 2F_1 F_2\right)}{F_1F_2\left(J^2 - F_1 F_2\right)}.
\label{eq:Rdiff_algebra}
\end{eqnarray}
Since $F_1F_2-J^2>0$, the denominator in Eq.~(\ref{eq:Rdiff_algebra}) is negative. Therefore, $R_{\rm series} - R_{\rm eff}>0$ holds if and only if
\begin{eqnarray}
J \left( J(F_1 + F_2) - 2F_1 F_2 \right)<0,
\label{eq:J_condition_intermediate}
\end{eqnarray}
which is equivalent to the strict range in Eq.~(\ref{eq:J_range}). This proves the condition for beating the series law.

For the optimization over $J$, observe from Eq.~(\ref{eq:R_eff_series}) that maximizing $F_Y^{(\mathbf{u})}$ is equivalent to minimizing $R_{\rm eff}$ over $J$ with $J^2 < F_1 F_2$. By differentiating $R_{\rm eff}$ with respect to $J$, we have
\begin{eqnarray}
\frac{\partial}{\partial J}R_{\rm eff} = \frac{-2(J - F_1)(J - F_2)}{\left(J^2 - F_1 F_2\right)^2}.
\label{eq:dR_dJ}
\end{eqnarray}
If $F_1 \neq F_2$, the unique stationary point in the allowed interval $\abs{J} < \sqrt{F_1 F_2}$ is $J=\min\{F_1, F_2\}$, and substituting this value into Eq.~(\ref{eq:Feff_sum_closed}) yields $F_Y^{(\mathbf{u})}=\min\{F_1, F_2\}$. In the symmetric case $F_1=F_2=:F$, Eq.~(\ref{eq:Feff_sum_closed}) simplifies to $F_Y^{(\mathbf{u})}=(F+J)/2$, whose supremum over $\abs{J}<F$ is $F$, approached as $J \to F^-$. Hence, in all cases, the supremum equals Eq.~(\ref{eq:max_Feff_min}).
\end{proof}

{\bf Theorem~\ref{thm:synergy_formula}} reveals a unifying mechanism. In a classical modular causal decomposition, score contributions are orthogonal and $J=0$, forcing the harmonic series penalty. Causal-model collapse manifests as a nonzero $J$, and when $J$ is positive in the range in Eq.~(\ref{eq:J_range}), the effective resistance for the total parameter drops below the classical series value. The gain is therefore not mysterious: it is the information-theoretic shadow of an interference-like correlation between modules that classical causality forbids.

\begin{remark}[The synergy in the Fisher information geometry]
This interpretation is precisely what makes the framework practical. The designer does not need to guess which ``quantum resource'' is present in a vague sense. The design target is concrete: engineer the experiment so that the score contributions for the would-be causal modules become positively correlated. In quantum dynamics, such correlations naturally arise when an attempted segmentation into independent modules is inconsistent with coherent evolution. In contextual schemes, they arise when the statistics cannot be embedded into a context-free causal model. In both cases, the resource is the same mathematical object: synergy in the Fisher information geometry.
\end{remark}

\subsection{Scaling intuition: cancelling the series penalty in long chains}\label{subsec:scaling_long_chain}

The causal-path CFII already carries a message that is easy to miss if one reads it only as a bound on precision: the familiar ``series penalty'' is not a law of nature, but a law of a model class. In a strictly classical causal chain, the information about an additive parameter must traverse a sequence of causal bottlenecks, and the CFII forces the inverse Fisher information to accumulate. Once the classical causal modularization collapses (as certified by a CFII violation), the bottleneck accounting rule stops being applicable, and the very same experiment can exhibit an endpoint sensitivity that would be mathematically impossible under the classical causal narrative.

\subsubsection*{Why the $1/N$ dilution appears classically} 

To isolate the scaling, consider an $N$-parameter description
$\boldsymbol{\theta}=(\theta_1,\ldots,\theta_N)^T$ and an additive causal quantity
\begin{eqnarray}
\Theta=\sum_{j=1}^N\theta_j = \mathbf{u}^T\boldsymbol{\theta}, \quad \mathbf{u} = (1,\ldots,1)^T.
\end{eqnarray}
The appropriate single-number information about $\Theta$ in a multi-parameter model is the effective Fisher information
\begin{eqnarray}
F^{(\mathbf{u})}(\boldsymbol{\theta})=\left(\mathbf{u}^T\mathbf{F}(\boldsymbol{\theta})^{-1}\mathbf{u}\right)^{-1},
\end{eqnarray}
where $\mathbf{F}(\boldsymbol{\theta})$ is the Fisher information matrix (FIM). This definition is not cosmetic: it is precisely the Cram\'er--Rao geometry for estimating the sum direction in parameter space.

In a strictly modular classical chain, the segments behave as independent ``causal modules'' in the precise sense used to derive the causal-path CFII: score contributions associated with different segments are orthogonal, hence the FIM is diagonal (or block-diagonal in more general modular networks). For the simplest equal-strength case, one may write
\begin{eqnarray}
\mathbf{F}_{\rm cl}(\boldsymbol{\theta}) = F \mathds{1}_N,
\end{eqnarray}
so that $\mathbf{F}_{\rm cl}^{-1} = (1/F) \mathds{1}_N$ (Here, $\hat{\mathds{1}}_N$ is the $N \times N$ identity matrix) and therefore
\begin{eqnarray}
\mathbf{u}^T\mathbf{F}_{\rm cl}^{-1}\mathbf{u}=\frac{N}{F}, \quad F^{(\mathbf{u})}_{\rm cl}=\frac{F}{N}.
\end{eqnarray}
This is exactly the classical ``series penalty'': when one tries to estimate the \emph{total} parameter accumulated over $N$ modules whose information contents are comparable, the effective information about the total parameter scales as $1/N$. Equivalently, the information-resistance $R:=F^{-1}$ grows linearly with $N$. The physical picture is the one we emphasized in the causal-path CFII: a classical chain is a sequence of passive mediators, and passive resistances in series must add.

\subsubsection*{What changes under causal-model collapse: off-diagonal Fisher synergy} 

The CFII derivation identifies the microscopic reason behind the series penalty: modularity forces score orthogonality and therefore removes off-diagonal entries of $\mathbf{F}$. Under causal-model collapse, those orthogonality relations need not hold. The score contributions from different segments become correlated, and the FIM develops off-diagonal structure. This is the multi-segment analogue of the two-segment ``synergy'' mechanism discussed in Sec.~\ref{subsec:synergy_mechanism}: the gain is the information-geometric shadow of correlations between score contributions that the classical causal model forbids.

To see how such correlations can cancel the series penalty in a transparent and analytically solvable way, it is useful to introduce a minimal toy model that cleanly separates diagonal ``local information'' from off-diagonal ``collective synergy.'' Thus, let us consider an ``equal-information, equal-correlation'' FIM of the form
\begin{eqnarray}
\mathbf{F}(\boldsymbol{\theta}) = F\left((1-\varepsilon) \mathds{1}_N + \varepsilon\mathbf{J}_N\right),
\label{eq:equicorrelated_FIM}
\end{eqnarray}
where $\mathbf{J}_N$ is the all-ones matrix. This corresponds to a score vector $\mathbf{s}=(s_1,\ldots,s_N)^T$ whose covariance is
$\mathbf{F}=\mathbb{E}[\mathbf{s}\mathbf{s}^T]$ with
\begin{eqnarray}
\mathbb{E}[s_j^2] = F, \quad \mathbb{E}[s_i s_j] = F\varepsilon \quad (i \neq j),
\end{eqnarray}
so that $\varepsilon$ directly quantifies the uniform pairwise correlation of segment-wise score contributions. For $0 \le \varepsilon < 1$, $\mathbf{F}$ is positive definite (and hence invertible), while $\varepsilon \to 1$ corresponds to an almost perfectly correlated score structure.

\medskip
There are two equally instructive ways to compute $F^{(\mathbf{u})}$: an eigenmode decomposition and a direct inversion.

\medskip\noindent
{\bf Eigenmode viewpoint}: The vector $\mathbf{u}$ is an eigenvector of $\mathbf{J}_N$ with eigenvalue $N$, and any vector orthogonal to $\mathbf{u}$ is an eigenvector of $\mathbf{J}_N$ with eigenvalue $0$. Hence, $\mathbf{u}$ is an eigenvector of $\mathbf{F}$ with eigenvalue $F(1 - \varepsilon + \varepsilon N)$, while the $(N-1)$ orthogonal directions have eigenvalue $F(1-\varepsilon)$. Since the total parameter $\Theta$ lives exactly in the $\mathbf{u}$ direction, the relevant information for estimating $\Theta$ is controlled by the inverse eigenvalue along $\mathbf{u}$, yielding
\begin{eqnarray}
\mathbf{u}^T\mathbf{F}(\boldsymbol{\theta})^{-1}\mathbf{u} = \frac{N}{F(1-\varepsilon+\varepsilon N)}.
\label{eq:uFu_N}
\end{eqnarray}
Therefore, the effective Fisher information becomes
\begin{eqnarray}
F^{(\mathbf{u})}(\boldsymbol{\theta}) = \left(\mathbf{u}^T\mathbf{F}(\boldsymbol{\theta})^{-1}\mathbf{u}\right)^{-1} = \frac{F(1 - \varepsilon + \varepsilon N)}{N}
=F\left(\varepsilon+\frac{1-\varepsilon}{N}\right).
\label{eq:Feff_N_equicorr}
\end{eqnarray}

\medskip\noindent
{\bf Direct inversion viewpoint}: One may also verify Eq.~(\ref{eq:uFu_N}) by explicitly inverting Eq.~(\ref{eq:equicorrelated_FIM}).
By using $\mathbf{J}_N^2=N\mathbf{J}_N$ and a short algebra, we have
\begin{eqnarray}
\mathbf{F}(\boldsymbol{\theta})^{-1} = \frac{1}{F(1-\varepsilon)}\left(\mathds{1}_N - \frac{\varepsilon}{ 1 - \varepsilon + \varepsilon N}\mathbf{J}_N\right),
\end{eqnarray}
and contracting with $\mathbf{u}$ yields $\mathbf{u}^T\mathbf{F}^{-1}\mathbf{u}$ in Eq.~(\ref{eq:uFu_N}).

Eq.~(\ref{eq:Feff_N_equicorr}) makes the scaling mechanism completely explicit. When $\varepsilon=0$ (no synergy), one recovers the classical harmonic scaling $F^{(\mathbf{u})}=F/N$. When $\varepsilon>0$ is fixed and $N$ becomes large, the second term $(1-\varepsilon)/N$ dies out and one finds the saturation
\begin{eqnarray}
\lim_{N\to\infty}F^{(\mathbf{u})}(\boldsymbol{\theta})=F\varepsilon,
\end{eqnarray}
namely, an $O(1)$ effective Fisher information for the total parameter $\Theta$ even for arbitrarily long chains. In the extreme synergy limit $\varepsilon \to 1$ (maximal correlation compatible with invertibility), Eq.~(\ref{eq:Feff_N_equicorr}) yields $F^{(\mathbf{u})} \to F$, i.e., the chain behaves as if there were no series penalty at all.

It is worth pausing on what this means physically. The equicorrelated matrix has a simple interpretation: it decomposes the information geometry into one ``collective mode'' along $\mathbf{u}$ and $(N-1)$ ``relative modes'' orthogonal to it. As $\varepsilon \to 1$, the collective mode becomes highly informative, while the relative modes become poorly conditioned (the eigenvalue $F(1-\varepsilon)$ tends to $0$). This captures an important metrological moral: the same correlations that make the sum parameter $\Theta$ easy to estimate can make the individual segment parameters hard to disentangle. In other words, the synergy is not magic that creates information about everything, but it is a geometric reallocation of information into the direction that the estimation task actually cares about.

\subsubsection*{From scaling intuition to operational design} 

The relevance to our causal perspective is immediate. A long chain is precisely the regime where classical causal reasoning would insist that precision must become hopelessly diluted, because each causal module adds an irreducible information-resistance. A large CFII violation is precisely the regime where Nature declares that this dilution narrative is not merely pessimistic but structurally wrong for the observed statistics. In such a regime, the causal-model collapse is not only a philosophical diagnosis; it is a design instruction: engineer the experiment so that the would-be segment-wise score contributions become positively correlated, thereby generating off-diagonal Fisher synergy and converting a classically series-limited architecture into an effectively collective sensor.

\subsection{CFII-guided metrological design}\label{subsec:design_principle}

The results above suggest a practical methodology for metrological design. A CFII is derived from an explicit classical causal hypothesis. If the hypothesis were true, it would impose a precision frontier (Sec.~\ref{subsec:classical_precision_bounds}). If the experiment violates the CFII, the violation simultaneously (i) falsifies the hypothesis (Sec.~\ref{sec:falsification}) and (ii) certifies a precision gain that can be harvested by an efficient estimator (Sec.~\ref{subsec:violation_certified_gain} and Sec.~\ref{subsec:achievability_mle}). Thus, a CFII is more than a test; it is a metrology design compass.

The most important practical implication is that the ``resource'' is \emph{self-certified}. To claim a metrological advantage, one typically must compare to a theoretical bound (such as a QFI benchmark) and argue that the implementation truly adheres to the assumptions behind that bound. In contrast, a CFII-based certification proceeds internally: it compares the achieved endpoint information to a causal benchmark derived from an explicit classical model. When the inequality is violated, the classical benchmark is not merely beaten; it is proven inapplicable. The very data that achieve the enhanced precision also certify that the enhancement arises from causal-model collapse. This is precisely the kind of internal certification that is valuable in practical sensing, where device imperfections often blur the meaning of idealized theoretical limits.

Equally important is that the mechanism identified in Sec.~\ref{subsec:synergy_mechanism} translates into a concrete control objective. A classical causal decomposition forces the score covariance $J$ to vanish. A metrological protocol that aims to exploit causal-model collapse should therefore be designed to generate positive score correlations between the would-be causal modules, while maintaining estimator stability (regularity of the likelihood and positivity of the effective statistical model used by the inference procedure). When these conditions are met, the improvement factor in Eq.~(\ref{eq:gain_factor_resistance}) becomes the experimentally measurable figure of merit: by tuning preparation, controls, and measurement contexts to push $V_{\rm path}$ deeper below zero, one directly tunes the achievable estimation error.

\suppnote{Examples: coherent dynamics beyond the classical causal-path frontier}\label{sec:examples_coherence_entanglement}

In the preceding sections, we established that a Fisher-information inequality becomes a \emph{causal} inequality once one commits to a classical causal model class. In particular, under the causal-path hypothesis $A \rightarrow C \rightarrow B$ with a modular (segment-wise) parameterization, Fisher information must obey a strict series law. When the series law is violated, the conclusion is that the entire classical causal-path explanation is falsified.

Here we illustrate these messages in an elementary coherent-dynamics setting. A single qubit undergoing a coherent rotation violates the causal-path CFII even with a fixed, simple projective readout. We then show that the resulting causal-path collapse is visible in Fisher-information landscapes, estimator-level performance, and adversarial classical benchmarks that optimize the intermediate split time.

\subsection{Unified benchmark: CFII gap and its metrological meaning}\label{subsec:ex_unified_benchmark}

We start with an additive parameter decomposition
\begin{eqnarray}
\theta_{ab}=\theta_{ac}+\theta_{cb}, \quad (\theta_{ac}, \theta_{cb} > 0),
\label{eq:sec6_additive}
\end{eqnarray}
and let $F(\theta)$ denote the (classical) Fisher information of the chosen readout statistics for estimating $\theta$ after an evolution of duration $\theta$.

Under the classical causal-path model class formalized in Sec.~\ref{sec:cfii}, the causal-path CFII reads
\begin{eqnarray}
F(\theta_{ab})^{-1} \ge F(\theta_{ac})^{-1} + F(\theta_{cb})^{-1},
\label{eq:sec6_cfii_path}
\end{eqnarray}
and the violation statistic is defined by
\begin{eqnarray}
V(\theta_{ac},\theta_{cb}) = F(\theta_{ab})^{-1}-F(\theta_{ac})^{-1} - F(\theta_{cb})^{-1}.
\label{eq:sec6_V_def}
\end{eqnarray}
The causal-path hypothesis requires $V(\theta_{ac},\theta_{cb}) \ge 0$ for all admissible splits. Thus, $V<0$ is a model-falsification certificate in the precise sense of Sec.~\ref{sec:falsification}.

To connect the witness to the operational precision, we define the classical causal-path benchmark Fisher information for a given split,
\begin{eqnarray}
F_{\rm cl}^{\rm(path)}(\theta_{ac},\theta_{cb}) = \left( F(\theta_{ac})^{-1} + F(\theta_{cb})^{-1} \right)^{-1}.
\label{eq:sec6_Fcl_path}
\end{eqnarray}
Then, Eq.~(\ref{eq:sec6_cfii_path}) is equivalent to $F(\theta_{ab}) \le F_{\rm cl}^{\rm(path)}$. Since the Cram\'er--Rao bound gives $\Delta\theta\ge 1/\sqrt{N F(\theta)}$ for $N$ i.i.d. samples~\cite{Cramer2016,Helstrom1976}, the causal-path model predicts an error floor
\begin{eqnarray}
\Delta\theta_{\rm cl} \ge \frac{1}{\sqrt{N F_{\rm cl}^{\rm(path)}(\theta_{ac},\theta_{cb})}}.
\label{eq:sec6_CRB_cl}
\end{eqnarray}
In contrast, the actual experiment yields
\begin{eqnarray}
\Delta\theta \ge \frac{1}{\sqrt{N F(\theta_{ab})}}.
\label{eq:sec6_CRB_actual}
\end{eqnarray}
Here we package the advantage by the error-ratio indicator
\begin{eqnarray}
G(\theta_{ac},\theta_{cb}) := \log\left(\frac{\Delta\theta}{\Delta\theta_{\rm cl}}\right) = \frac{1}{2}\log\left(\frac{F_{\rm cl}^{\rm(path)}(\theta_{ac},\theta_{cb})}{F(\theta_{ab})}\right).
\label{eq:sec6_G_def}
\end{eqnarray}
Note that $G<0$ is exactly the metrological-gain regime. Thus, we can summarize as:
\begin{theorem}[Witness--resource equivalence for the causal-path CFII]
\label{thm:sec6_witness_resource_equiv}
For $F(\theta_{ab}), F(\theta_{ac}), F(\theta_{cb}) > 0$, the following statements are equivalent:
\begin{eqnarray}
V(\theta_{ac},\theta_{cb})<0, \quad F(\theta_{ab})>F_{\rm cl}^{\rm(path)}(\theta_{ac},\theta_{cb}), \quad G(\theta_{ac},\theta_{cb})<0.
\label{eq:sec6_equiv}
\end{eqnarray}
\end{theorem}

\begin{proof}
By definition,
\begin{eqnarray}
V(\theta_{ac},\theta_{cb})<0
&\Longleftrightarrow& F(\theta_{ab})^{-1} < F(\theta_{ac})^{-1} + F(\theta_{cb})^{-1} \nonumber \\
&\Longleftrightarrow& F(\theta_{ab}) > \left(F(\theta_{ac})^{-1} + F(\theta_{cb})^{-1}\right)^{-1} = F_{\rm cl}^{\rm(path)}(\theta_{ac},\theta_{cb}).
\end{eqnarray}
By substituting this inequality into Eq.~(\ref{eq:sec6_G_def}), we obtain
\begin{eqnarray}
F(\theta_{ab}) > F_{\rm cl}^{\rm(path)}(\theta_{ac},\theta_{cb}) 
\Longleftrightarrow 
\frac{F_{\rm cl}^{\rm(path)}(\theta_{ac},\theta_{cb})}{F(\theta_{ab})} < 1
\Longleftrightarrow
G(\theta_{ac},\theta_{cb}) < 0.
\end{eqnarray}
The proof is completed.
\end{proof}

{\bf Theorem~\ref{thm:sec6_witness_resource_equiv}} is the operational core and summary of Sec.~\ref{sec:resource}. Thus, a CFII violation is simultaneously a falsification of the classical causal-path model class and a guaranteed precision advantage relative to that class.

\subsection{Single-qubit coherent rotation}\label{subsec:sec6_single_qubit}

We begin with a single qubit evolving under a coherent unitary rotation. We cast a simple model:
\begin{eqnarray}
\hat{H}=\frac{1}{2}\hat{\sigma}_x,
\quad
\hat{U}(\theta) = e^{-i\theta \hat{H}} = e^{-i\theta \hat{\sigma}_x/2}.
\label{eq:sec6_single_U}
\end{eqnarray}
We prepare a (pure) probe state
\begin{eqnarray}
\ket{\psi(\vartheta,\varphi)} = \cos\frac{\vartheta}{2}\ket{0} + e^{i\varphi}\sin\frac{\vartheta}{2}\ket{1},
\label{eq:sec6_single_probe}
\end{eqnarray}
and perform a fixed projective measurement in the $\hat{\sigma}_z$ basis,
\begin{eqnarray}
\hat{M}_0=\ketbra{0}{0},
\quad
\hat{M}_1=\ketbra{1}{1}.
\label{eq:sec6_single_readout}
\end{eqnarray}
Writing the binary probabilities as $p_0(\theta)=\tfrac{1+z(\theta)}{2}$ and $p_1(\theta)=\tfrac{1-z(\theta)}{2}$, one finds
\begin{eqnarray}
z(\theta) = \cos{\vartheta}\cos{\theta} + \sin{\vartheta}\sin{\varphi}\sin{\theta}.
\label{eq:sec6_single_z}
\end{eqnarray}
For binary statistics, the Fisher information is given by
\begin{eqnarray}
F(\theta) = \sum_{m=0,1}p_m(\theta)\left(\frac{\partial}{\partial\theta}\ln{p_m(\theta)}\right)^2 = \frac{\left(\frac{\partial}{\partial\theta}z(\theta)\right)^2}{1-z(\theta)^2},
\label{eq:sec6_single_F_general}
\end{eqnarray}
where the removable singularities at $z(\theta)^2=1$ are treated by continuity.

Then, we provide the result that shows that the causal-path CFII can be violated in a maximally clean way: no entanglement, no post-selection, no measurement optimization. A single fixed projective measurement suffices, yet the classical causal-path model collapses deterministically.
\begin{result}[Deterministic single-qubit collapse of the causal-path model]
\label{result:sec6_single_deterministic}
Fix $\varphi=\pi/2$ in Eq.~(\ref{eq:sec6_single_probe}). Then, the Fisher information of the readout in Eq.~(\ref{eq:sec6_single_readout}) is constant,
\begin{eqnarray}
F(\theta)=1 \quad (\forall \, \theta),
\label{eq:sec6_single_F_const}
\end{eqnarray}
and the causal-path CFII gap in Eq.~(\ref{eq:sec6_V_def}) satisfies
\begin{eqnarray}
V(\theta_{ac},\theta_{cb})=-1 \quad (\forall \, \theta_{ac}, \theta_{cb} > 0).
\label{eq:sec6_single_V_const}
\end{eqnarray}
Equivalently, $F_{\rm cl}^{\rm(path)}=\tfrac{1}{2}$ and the achieved Fisher information exceeds the classical causal-path benchmark such that:
\begin{eqnarray}
\frac{F(\theta_{ab})}{F_{\rm cl}^{\rm(path)}(\theta_{ac},\theta_{cb})} = 2.
\label{eq:sec6_single_gain}
\end{eqnarray}
\end{result}

The deterministic point makes the causal content unmistakable. A classical causal-path narrative demands that information about $\theta_{ab}$ must traverse an intermediate classical mediator $C$ and therefore pay a series penalty. The coherent qubit rotation refuses: it yields twice the Fisher information permitted by any such bottleneck.

This advantage is not merely a bound-level statement. At the deterministic point, the endpoint model is a simple binary fringe $p_0(\theta)=\cos^2((\theta-\vartheta)/2)$.
With $N$ i.i.d. samples, a maximum-likelihood estimate is obtained by matching the frequency $\widehat{p}_0=n_0/N$ and inverting
\begin{eqnarray}
\widehat{\theta} = \vartheta + 2\arccos\sqrt{\widehat{p}_0},
\label{eq:sec6_single_MLE}
\end{eqnarray}
on an interval where the mapping is one-to-one. Asymptotically, the estimator achieves $\Delta\theta \simeq 1/\sqrt{N}$, whereas the classical causal-path frontier predicts $\Delta\theta_{\rm cl} \ge \sqrt{2}/\sqrt{N}$.

\subsection{Generic regimes, numerical landscapes, and adversarial benchmarks}\label{subsec:sec6_numerics}

The deterministic result above might appear as a special symmetric point. Its true role is conceptual: at that point the Fisher information is constant, and the classical series law fails uniformly for every split. The natural question is what happens away from such a symmetry. The answer is that the uniform collapse deforms into structured regions of collapse rather than disappearing. Thus, here we show: (i) a representative landscape where $V<0$ and the metrological advantage $G<0$ persist broadly, (ii) estimator-level achievability of the advantage, and (iii) robustness against an ``adversarial'' classical benchmark that optimizes the split time to help the classical model as much as possible.

\subsubsection*{\emph{(i)} Representative landscape: $V<0$ and $G<0$ in a generic single-qubit setting} 

\begin{figure}[t]
\centering
\includegraphics[width=0.80\textwidth]{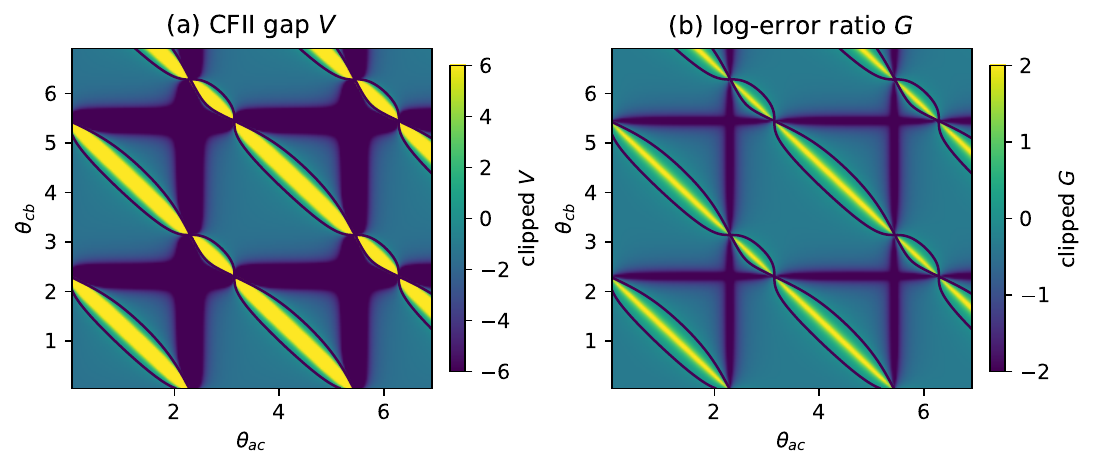}
\caption{Representative single-qubit landscapes for the generic setting $(\vartheta, \varphi)=(0.7\pi, 0.3\pi)$ with $\hat{\sigma}_z$ readout. (a) $V(\theta_{ac},\theta_{cb})$ in Eq.~(\ref{eq:sec6_V_def}) (clipped). (b) $G(\theta_{ac},\theta_{cb})$ in Eq.~(\ref{eq:sec6_G_def}) (clipped), where $G<0$ certifies an operational precision advantage over the classical causal-path benchmark.}
\label{fig:sec6_landscapes}
\end{figure}

The CFII gap $V$ and the gain indicator $G$ for a representative nontrivial single-qubit setting are visualized in Fig.~\ref{fig:sec6_landscapes}: a generic probe $(\vartheta,\varphi)=(0.7\pi,0.3\pi)$ with $\hat{\sigma}_z$ readout.

{\bf By Theorem~\ref{thm:sec6_witness_resource_equiv}}, the negative region in the $V$ landscape is exactly the negative region in the $G$ landscape. Hence, the figure simultaneously shows where the classical causal-path explanation fails and where the metrological advantage over that explanation is guaranteed.

\subsubsection*{\emph{(ii)} Achievability: the gain is not only a bound, but an estimator-level reality}  

\begin{figure}[t]
\centering
\includegraphics[width=0.43\textwidth]{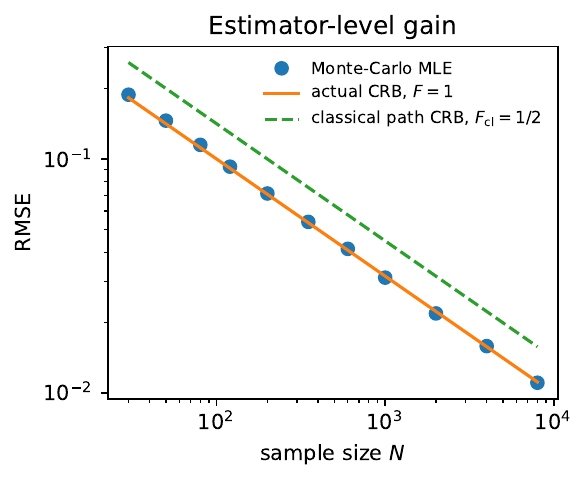}
\caption{Achievability of the advantage for the deterministic single-qubit point. Monte-Carlo RMSE of a simple MLE (markers) versus sample size $N$ (log--log), compared to the quantum CRB $1/\sqrt{N F(\theta_{ab})}$ (solid) and the classical causal-path bound $1/\sqrt{N F_{\rm cl}^{\rm(path)}}$ (dashed). Here $F=1$ and $F_{\rm cl}=1/2$.}
\label{fig:sec6_rmse}
\end{figure}

Fig.~\ref{fig:sec6_rmse} demonstrates that the gain is operationally attainable. For the deterministic point of {\bf Result~\ref{result:sec6_single_deterministic}}, the endpoint likelihood is a simple binary fringe, and a maximum-likelihood estimator saturates the Cram\'er--Rao bound asymptotically. The Monte-Carlo RMSE scaling shows the RMSE approaching $1/\sqrt{N F(\theta_{ab})}$ while staying strictly below the classical causal-path bound $1/\sqrt{N F_{\rm cl}^{\rm(path)}}$ by the predicted factor $\sqrt{2}$ in standard deviation.

\subsubsection*{\emph{(iii)} Adversarial classical benchmark: optimizing the split time}

\begin{figure}[t]
\centering
\includegraphics[width=0.50\textwidth]{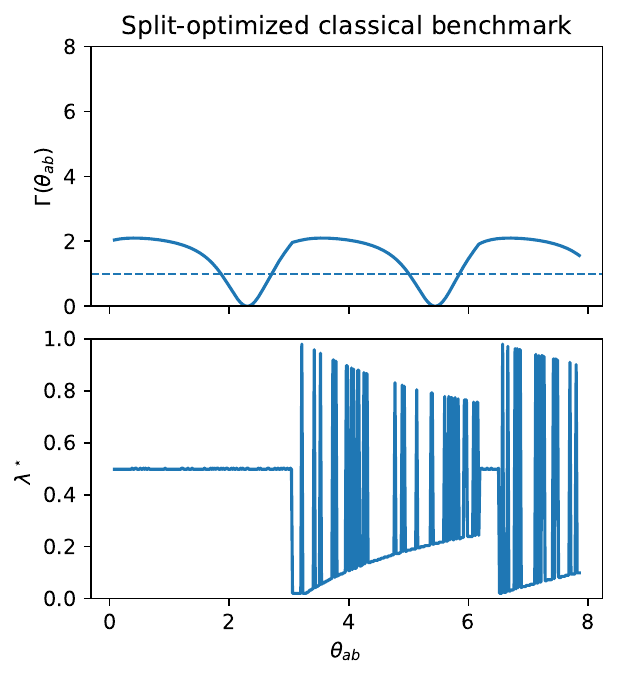}
\caption{Adversarial (split-optimized) classical benchmark for the representative single-qubit setting used in Fig.~\ref{fig:sec6_landscapes}. Top: the ratio $\Gamma(\theta_{ab})=F(\theta_{ab})/F_{\rm cl}^{\rm(opt)}(\theta_{ab})$ from Eq.~(\ref{eq:sec6_Gamma}). The dashed line $\Gamma=1$ is the most generous classical causal-path frontier once the split is optimized; $\Gamma>1$ certifies the split-robust causal-model collapse ({\bf Result~\ref{result:sec6_opt_robust}}). Bottom: the optimal split fraction $\lambda^\star(\theta_{ab})=\theta_{ac}^\star/\theta_{ab}$ from Eq.~(\ref{eq:sec6_lambda_star}).}
\label{fig:sec6_opt_benchmark}
\end{figure}

A skeptical classicalist might attempt the following escape: ``Perhaps the CFII violation occurs only because we chose an unfortunate split time $t_c$. If we pick a different intermediate time, maybe the inequality recovers.'' Within the causal-path hypothesis, this intuition is already inconsistent: a trajectory-like classical mediator should exist at any intermediate time, so the CFII must hold for every split. Nevertheless, to make the comparison maximally conservative, we now give the classical causal-path model the power to choose the split that benefits it the most.

For a fixed total interval $\theta_{ab}$, define the split-optimized classical benchmark
\begin{eqnarray}
F_{\rm cl}^{\rm(opt)}(\theta_{ab}) := \max_{0<\theta_{ac}<\theta_{ab}} \left( F(\theta_{ac})^{-1} + F(\theta_{ab}-\theta_{ac})^{-1} \right)^{-1}.
\label{eq:sec6_Fcl_opt}
\end{eqnarray}
We also define the optimized gain ratio
\begin{eqnarray}
\Gamma(\theta_{ab}) := \frac{F(\theta_{ab})}{F_{\rm cl}^{\rm(opt)}(\theta_{ab})},
\label{eq:sec6_Gamma}
\end{eqnarray}
and the corresponding optimal split fraction
\begin{eqnarray}
\lambda^\star(\theta_{ab}) := \frac{\theta_{ac}^\star(\theta_{ab})}{\theta_{ab}},
\label{eq:sec6_lambda_star}
\end{eqnarray}
where 
\begin{eqnarray}
\theta_{ac}^\star(\theta_{ab}) \in \arg\max_{0<\theta_{ac}<\theta_{ab}} \left( F(\theta_{ac})^{-1}+F(\theta_{ab}-\theta_{ac})^{-1} \right)^{-1}.
\end{eqnarray}
The line $\Gamma(\theta_{ab})=1$ is the most generous classical causal-path frontier once the split is optimized. Here, we can prove the following result:
\begin{result}[Split-robust collapse under the optimized classical benchmark]
\label{result:sec6_opt_robust}
Fix $\theta_{ab}$ and define $F_{\rm cl}^{\rm(opt)}(\theta_{ab})$ by Eq.~(\ref{eq:sec6_Fcl_opt}). If the following holds
\begin{eqnarray}
\Gamma(\theta_{ab})>1,
\label{eq:sec6_Gamma_gt1}
\end{eqnarray}
then the causal-path CFII is violated for every admissible split $\theta_{ab}=\theta_{ac}+\theta_{cb}$:
\begin{eqnarray}
V(\theta_{ac},\theta_{cb}) < 0
\quad
\forall \,\theta_{ac} \in (0,\theta_{ab}) ,\ \theta_{cb}=\theta_{ab}-\theta_{ac}.
\label{eq:sec6_violate_all_splits}
\end{eqnarray}
In particular, no choice of intermediate time $t_c$ can restore a classical causal-path explanation for that $\theta_{ab}$.
\end{result}

\begin{proof}---For each admissible split, let us define the split-dependent benchmark
\begin{eqnarray}
F_{\rm cl}^{\rm(path)}(\theta_{ac},\theta_{cb}) = \left( F(\theta_{ac})^{-1} + F(\theta_{cb})^{-1} \right)^{-1}.
\end{eqnarray}
By definition of the maximum in Eq.~(\ref{eq:sec6_Fcl_opt}),
\begin{eqnarray}
F_{\rm cl}^{\rm(path)}(\theta_{ac},\theta_{cb}) \le F_{\rm cl}^{\rm(opt)}(\theta_{ab}) \quad \mbox{for all admissible splits}.
\label{eq:sec6_path_le_opt}
\end{eqnarray}
If $\Gamma(\theta_{ab}) > 1$, then $F(\theta_{ab}) > F_{\rm cl}^{\rm(opt)}(\theta_{ab})$. By combining with Eq.~(\ref{eq:sec6_path_le_opt}), we obtain
\begin{eqnarray}
F(\theta_{ab}) > F_{\rm cl}^{\rm(path)}(\theta_{ac},\theta_{cb}) \quad \mbox{for all admissible splits}.
\end{eqnarray}
This is equivalent to $V(\theta_{ac},\theta_{cb}) < 0$ by the definition of Eq.~(\ref{eq:sec6_V_def}). Thus, the proof is completed.
\end{proof}

Fig.~\ref{fig:sec6_opt_benchmark} shows $\Gamma(\theta_{ab})$ and $\lambda^\star(\theta_{ab})$ for the generic single-qubit setting used in Fig.~\ref{fig:sec6_landscapes}. The top panel shows that $\Gamma(\theta_{ab})$ exceeds unity over wide ranges of $\theta_{ab}$ even after the split is optimized. This is precisely the ``even an adversarial classical benchmark cannot catch up'' regime: the causal-path explanation collapses regardless of how one tries to place the intermediate classical mediator.

\subsection{Chain-amplified gain in long causal-path decompositions}

The classical causal-path narrative is, in fact, far more committal than the existence of one intermediate node. If the evolution is truly describable by a classical trajectory-like mediation, then every refinement of the path into more intermediate bottlenecks must also admit a consistent classical causal decomposition. This observation leads to a stronger and more dramatic test (as already discussed in Sec.~\ref{subsec:scaling_long_chain}): instead of a single split $\theta_{ab}=\theta_{ac}+\theta_{cb}$, we impose an entire discrete chain of $K$ segments, and ask whether the classical ``series law'' can survive such a refined causal claim.

\subsubsection*{From a single split to a $K$-step causal chain} 

Fix a total interval $\theta_{ab}>0$ and choose intermediate times $t_0=t_a<t_1<\cdots<t_{K-1}<t_K=t_b$, so that the total parameter decomposes additively as
\begin{eqnarray}
\theta_{ab} = \sum_{j=1}^{K}\theta_{j-1,j}, \quad \theta_{j-1,j} := t_j - t_{j-1} > 0.
\end{eqnarray}
For the same measurement used in Sec.~\ref{subsec:sec6_single_qubit}, let $F(\theta)$ denote the (classical) Fisher information of the endpoint measurement statistics after evolution duration $\theta$. The $K$-step causal-path hypothesis asserts that the information about $\theta_{ab}$ must traverse $K-1$ intermediate classical bottlenecks in series, with modular parameter dependence on each link. By the $N$-step chain CFII ({\bf Corollary~\ref{cor:CFII_chain}}), this implies the multi-split series law
\begin{eqnarray}
F(\theta_{ab})^{-1} \ge \sum_{j=1}^{K} F(\theta_{j-1,j})^{-1}.
\end{eqnarray}
It is convenient to package this into a $K$-step violation statistic,
\begin{eqnarray}
V_K(\theta_{0,1},\ldots,\theta_{K-1,K}) := F(\theta_{ab})^{-1} - \sum_{j=1}^{K}F(\theta_{j-1,j})^{-1},
\end{eqnarray}
so that the classical causal-path model class requires $V_K \ge 0$ for all admissible partitions of $\theta_{ab}$. The corresponding $K$-step classical benchmark Fisher information is the harmonic-series composition
\begin{eqnarray}
F^{(K)}_{\mathrm{cl}}(\theta_{0,1},\ldots,\theta_{K-1,K}) := \left(\sum_{j=1}^{K}F(\theta_{j-1,j})^{-1}\right)^{-1}.
\end{eqnarray}
Hence, $V_K < 0$ is again a witness-resource certificate: it simultaneously falsifies the $K$-step classical causal-path class and certifies a precision gap relative to that class. 

Here we summarize our result:
\begin{result}[Chain-amplified collapse and gain for the single-qubit coherent example]
\label{result:long_chain_example}
Assume that, for the fixed readout employed in Sec.~\ref{subsec:sec6_single_qubit}, the Fisher information is constant,
\begin{eqnarray}
F(\theta) = F_0 > 0 \quad (\forall \, \theta > 0).
\end{eqnarray}
Then, for every integer $K \ge 2$ and for every admissible partition $\theta_{ab}=\sum_{j=1}^{K}\theta_{j-1,j}$ with $\theta_{j-1,j} > 0$, the $K$-step causal-path CFII is violated uniformly:
\begin{eqnarray}
V_K(\theta_{0,1}, \ldots, \theta_{K-1,K}) = -\frac{K-1}{F_0} < 0.
\end{eqnarray}
Equivalently, the $K$-step classical benchmark equals
\begin{eqnarray}
F^{(K)}_{\mathrm{cl}}(\theta_{0,1}, \ldots, \theta_{K-1,K}) = \frac{F_0}{K},
\end{eqnarray}
and the certified gain ratio is exactly linear in the chain depth,
\begin{eqnarray}
\Gamma_K(\theta_{ab}) = \frac{F(\theta_{ab})}{F^{(K)}_{\mathrm{cl}}(\theta_{0,1}, \ldots, \theta_{K-1,K})} = K.
\end{eqnarray}
In particular, for $N_{\mathrm{s}}$ i.i.d. repetitions of the endpoint experiment, the classical causal-path model predicts an error floor
\begin{eqnarray}
\Delta\theta_{\mathrm{cl}}^{(K)} \ge \frac{1}{\sqrt{N_{\mathrm{s}}F^{(K)}_{\mathrm{cl}}}} = \sqrt{\frac{K}{N_{\mathrm{s}}F_0}},
\end{eqnarray}
whereas the actual coherent-dynamics readout achieves asymptotically
\begin{eqnarray}
\Delta\theta\simeq \frac{1}{\sqrt{N_{\mathrm{s}}F(\theta_{ab})}} = \frac{1}{\sqrt{N_{\mathrm{s}}F_0}},
\end{eqnarray}
thus beating the $K$-step classical causal-path frontier by a factor $\sqrt{K}$ in standard deviation.
\end{result}

\begin{proof}---Under the assumption $F(\theta)=F_0$, one has $F(\theta_{ab})=F_0$ and $F(\theta_{j-1,j})=F_0$ for all segments. By substituting into the definition of $V_K$, we find
\begin{eqnarray}
V_K = F_0^{-1} - \sum_{j=1}^{K} F_0^{-1} = \frac{1-K}{F_0} = -\frac{K-1}{F_0} < 0,
\end{eqnarray}
proving the uniform violation. Similarly,
\begin{eqnarray}
F^{(K)}_{\mathrm{cl}} = \left(\sum_{j=1}^{K}F_0^{-1}\right)^{-1} = \left(\frac{K}{F_0}\right)^{-1} = \frac{F_0}{K},
\end{eqnarray}
and therefore $\Gamma_K=F_0/(F_0/K)=K$. The error statements follow from the Cram\'er--Rao scaling $\Delta\theta \simeq 1/\sqrt{N_{\mathrm{s}}F}$ applied to the endpoint Fisher information and to the classical benchmark Fisher information.
\end{proof}

{\bf Result~\ref{result:long_chain_example}} applies immediately to the deterministic coherent point established earlier. For the single-qubit setting of Sec.~\ref{subsec:sec6_single_qubit} at $\varphi=\pi/2$, {\bf Result~\ref{result:sec6_single_deterministic}} yields $F(\theta)=1$ for all $\theta$, hence $F_0=1$ and the chain-amplified violation reads
\begin{eqnarray}
V_K=-(K-1), \quad F^{(K)}_{\mathrm{cl}}=\frac{1}{K}, \quad \Gamma_K=K.
\end{eqnarray}
The metrological meaning is clear: the endpoint statistics behave as if the total parameter were estimated without suffering the classical series dilution, whereas any classical $K$-step causal-path explanation must pay a $1/K$ penalty in Fisher information.

In the causal language of this work, increasing $K$ does not correspond to increasing experimental resources. Rather, it corresponds to tightening the classical causal claim. A $K$-step causal-path model asserts the existence of $K-1$ intermediate classical mediators such that the parameter dependence modularizes across the chain. Each additional mediator is a new bottleneck and therefore contributes an additional information-resistance term $F(\theta_{j-1,j})^{-1}$ that must add in series. The coherent dynamics violates the modularity that enforces this orthogonality of score contributions, and the consequence is as dramatic as it is quantitative: the more finely one insists on discretizing a classical trajectory, the more severely the classical series law is contradicted.

\subsubsection*{Adversarially optimized multi-split benchmarks} 

One may extend the ``adversarial benchmark'' logic of Sec.~\ref{subsec:sec6_numerics}(iii) from a single split to an entire partition. Let us define the $K$-split-optimized classical benchmark by
\begin{eqnarray}
F^{(K,\mathrm{opt})}_{\mathrm{cl}}(\theta_{ab}) := \max_{\theta_{0,1},\ldots,\theta_{K-1,K}>0\atop \sum_{j=1}^{K}\theta_{j-1,j}=\theta_{ab}} \left(\sum_{j=1}^{K}F(\theta_{j-1,j})^{-1}\right)^{-1},
\end{eqnarray}
and the corresponding optimized gain ratio
\begin{eqnarray}
\Gamma_K^{(\mathrm{opt})}(\theta_{ab}) := \frac{F(\theta_{ab})}{F^{(K,\mathrm{opt})}_{\mathrm{cl}}(\theta_{ab})}.
\end{eqnarray}
In the constant Fisher information coherent point of {\bf Result~\ref{result:sec6_single_deterministic}}, the optimization is powerless: since $F(\theta)=F_0$ for all segment lengths, every partition yields the same harmonic value $F_0/K$, and therefore
\begin{eqnarray}
F^{(K,\mathrm{opt})}_{\mathrm{cl}}(\theta_{ab})=\frac{F_0}{K},
\quad
\Gamma_K^{(\mathrm{opt})}(\theta_{ab})=K.
\end{eqnarray}
Thus, the chain-amplified collapse is not only split-robust; it is partition-robust.

More generally, whenever the chosen readout yields a bounded Fisher information $F(\theta)\le F_{\max}$ over the relevant domain of segment lengths, the classicalist cannot evade the series penalty even by optimizing the entire partition:
\begin{eqnarray}
F^{(K,\mathrm{opt})}_{\mathrm{cl}}(\theta_{ab}) 
= \max \left(\sum_{j=1}^{K}F(\theta_{j-1,j})^{-1}\right)^{-1}
\le \left(\sum_{j=1}^{K}F_{\max}^{-1}\right)^{-1}
= \frac{F_{\max}}{K}.
\end{eqnarray}
Hence, a $1/K$ dilution of the classical benchmark is an unavoidable consequence of insisting on $K$ passive classical bottlenecks. Whenever the coherent dynamics maintains an $O(1)$ endpoint Fisher information for the total parameter (as it does at the deterministic point, and approximately does in broad regions of the generic landscape), the gain relative to the $K$-step causal-path frontier is naturally amplified with chain depth.

\suppnote{AI-assisted adversarial finite-data stress test}\label{sec:ai_adversarial_stress}

The examples in Sec.~\ref{sec:examples_coherence_entanglement} establish the CFII violation analytically and show that the resulting gain is estimator-achievable. In this section, we add a numerical stress test designed to address a more practical question: can the negative witness be certified from noisy finite data even when a flexible classical causal adversary is allowed to optimize its hidden mediator and stochastic kernels? The answer is affirmative. The purpose of the simulation is not to replace the theorems, but to make their operational content harder to dismiss: we deliberately give the classical causal-path model a strong, differentiable, AI-optimized adversary and verify that it can saturate, but not cross, the CFII frontier.

\subsection{Noisy coherent-fringe data generator}

We use the same binary readout structure as in the coherent single-qubit example, but include two elementary imperfections: visibility loss and symmetric readout error. The simulated endpoint probabilities are
\begin{eqnarray}
p_0^{(\gamma)}(\theta)=\frac{1+z_\gamma(\theta)}{2},\quad
p_1^{(\gamma)}(\theta)=\frac{1-z_\gamma(\theta)}{2},
\label{eq:ai_noisy_probs}
\end{eqnarray}
where
\begin{eqnarray}
z_\gamma(\theta)=\eta_{\rm r}e^{-\gamma\theta}\cos(\theta-\vartheta_0),
\quad
\eta_{\rm r}=1-2\epsilon_{\rm r}.
\label{eq:ai_noisy_z}
\end{eqnarray}
Here $\gamma$ is a dimensionless dephasing-rate parameter and $\epsilon_{\rm r}$ is a symmetric readout-error probability. The derivative is
\begin{eqnarray}
\dot z_\gamma(\theta)
=
\frac{\partial z_\gamma(\theta)}{\partial\theta}
=
-\eta_{\rm r}e^{-\gamma\theta}
\left[
\gamma\cos(\theta-\vartheta_0)+\sin(\theta-\vartheta_0)
\right],
\label{eq:ai_noisy_dz}
\end{eqnarray}
and the binary Fisher information is
\begin{eqnarray}
F_\gamma(\theta)=
\frac{\dot z_\gamma(\theta)^2}{1-z_\gamma(\theta)^2}.
\label{eq:ai_noisy_F}
\end{eqnarray}
The numerical results below use
\begin{eqnarray}
T=\frac{\pi}{2},\quad K=4,\quad \vartheta_0=0,\quad \epsilon_{\rm r}=0.02,
\label{eq:ai_parameters}
\end{eqnarray}
and compare the endpoint experiment at $T$ with an equal $K$-segment classical causal-chain benchmark. The equal-segment witness and gain ratio are
\begin{eqnarray}
V_K(\gamma)
&=&
F_\gamma(T)^{-1}
-
K F_\gamma(T/K)^{-1},
\label{eq:ai_VK_gamma}
\\
\Gamma_K(\gamma)
&=&
\frac{F_\gamma(T)}
{F_\gamma(T/K)/K}.
\label{eq:ai_GammaK_gamma}
\end{eqnarray}
For example, at $\gamma=0.25$ one obtains
\begin{eqnarray}
F_\gamma(T)=0.4202,\quad
F_\gamma(T/K)=0.8065,\quad
V_K=-2.5799,\quad
\Gamma_K=2.0841.
\label{eq:ai_numeric_gamma025}
\end{eqnarray}
The optimized coherent advantage disappears only when $\Gamma_K$ reaches unity; for the parameters above this crossing occurs near $\gamma\simeq0.444$.

\subsection{Finite-shot FI estimation and delta-method certification}

For the binary model in Eq.~(\ref{eq:ai_noisy_probs}), the score values are explicit:
\begin{eqnarray}
s_0(\theta)=\frac{\dot z_\gamma(\theta)}{1+z_\gamma(\theta)},
\quad
s_1(\theta)=-\frac{\dot z_\gamma(\theta)}{1-z_\gamma(\theta)}.
\label{eq:ai_scores}
\end{eqnarray}
Given $N_s$ independent samples in a context $\theta_s$, the plug-in FI estimator is
\begin{eqnarray}
\widehat F_s
=
\frac{1}{N_s}\sum_{i=1}^{N_s}s_{x_i}(\theta_s)^2.
\label{eq:ai_Fhat}
\end{eqnarray}
This is the concrete version of the score estimator discussed in Sec.~\ref{subsec:finite_data}. Let
\begin{eqnarray}
\mu_{4,s}
=
\sum_{x=0,1}
p_x^{(\gamma)}(\theta_s)s_x(\theta_s)^4.
\label{eq:ai_mu4}
\end{eqnarray}
Then
\begin{eqnarray}
{\rm Var}(\widehat F_s)=\frac{\mu_{4,s}-F_s^2}{N_s}.
\label{eq:ai_var_Fhat}
\end{eqnarray}
Applying the delta method to $R_s=F_s^{-1}$ gives
\begin{eqnarray}
{\rm Var}(\widehat R_s)
\simeq
\frac{\mu_{4,s}-F_s^2}{N_s F_s^4}.
\label{eq:ai_var_Rhat}
\end{eqnarray}
For the equal-segment $K$-chain test with independent data in each context, the standard error of
\begin{eqnarray}
\widehat V_K
=
\widehat F(T)^{-1}
-
\sum_{j=1}^{K}\widehat F(T/K)^{-1}
\label{eq:ai_Vhat}
\end{eqnarray}
is therefore
\begin{eqnarray}
{\rm SE}(\widehat V_K)^2
\simeq
\frac{\mu_{4,T}-F_T^2}{N_s F_T^4}
+
K
\frac{\mu_{4,T/K}-F_{T/K}^2}{N_s F_{T/K}^4}.
\label{eq:ai_SE}
\end{eqnarray}
We report the certification significance as
\begin{eqnarray}
Z=-\frac{V_K}{{\rm SE}(\widehat V_K)}.
\label{eq:ai_Zscore}
\end{eqnarray}
For $\gamma=0.25$ and $N_s=10^3$ shots per context, Eq.~(\ref{eq:ai_SE}) gives ${\rm SE}=0.2121$ and hence $Z=12.17$. A direct finite-shot Monte-Carlo simulation gives a $95\%$ interval approximately $[-3.06,-2.22]$ for $\widehat V_K$, well separated from the classical boundary at zero.

\subsection{Classifier score estimation}

The previous subsection used the analytic score only to make the certification transparent. In an actual complex experiment, the likelihood and score may not be available in closed form. We therefore also implement a classifier likelihood-ratio estimator, a standard simulation-based inference idea~\cite{Cranmer2020SBI}. Suppose that one can generate or collect samples at two nearby parameters $\theta+\delta$ and $\theta-\delta$. A binary classifier trained with equal class priors to distinguish these two sample sets has an optimal logit
\begin{eqnarray}
g^\star(x)
=
\log\frac{p(x|\theta+\delta)}{p(x|\theta-\delta)}.
\label{eq:ai_classifier_logit}
\end{eqnarray}
For small $\delta$,
\begin{eqnarray}
\frac{g^\star(x)}{2\delta}
=
\frac{\partial}{\partial\theta}\log p(x|\theta)
+
O(\delta^2),
\label{eq:ai_classifier_score}
\end{eqnarray}
so the FI can be estimated by averaging the squared classifier score. For the binary readout considered here, a saturated neural classifier is equivalent to a two-logit model over $x=0,1$. The closed-form finite-sample estimate used in the code is
\begin{eqnarray}
\widehat s_{\rm cls}(x)
=
\frac{1}{2\delta}
\log
\frac{(n_{+,x}+\alpha)/(N_++2\alpha)}
{(n_{-,x}+\alpha)/(N_-+2\alpha)},
\label{eq:ai_classifier_estimator}
\end{eqnarray}
where $n_{\pm,x}$ are the counts of outcome $x$ in the $\theta\pm\delta$ training samples and $\alpha$ is a smoothing parameter. In the reported figure we used $\delta=0.10$ and $\alpha=5$. The same formula is replaced by a neural network logit when $x$ is high-dimensional.

\subsection{Differentiable modular classical adversary}

We now describe the AI-optimized classical causal adversary. The adversary is not allowed to abandon the causal assumptions that define the CFII. It must remain a modular causal path,
\begin{eqnarray}
p_\phi(c,b|\theta_1,\theta_2)
=
\alpha_\phi(c|\theta_1)\,
\beta_\phi(b|c,\theta_2),
\label{eq:ai_modular_path}
\end{eqnarray}
with $\theta_1=\theta_{ac}$ and $\theta_2=\theta_{cb}$. However, within this constraint it is made deliberately flexible: $c$ is a latent mediator with $L$ possible values, $b$ has $M$ possible endpoint outcomes, and the local stochastic kernels are optimized by gradient descent.

Because a CFII is a local Fisher-information statement, it is enough to parameterize the local values and local derivatives of the kernels. We use softmax tangent models,
\begin{eqnarray}
\alpha_\phi(c|\theta_1)
&=&
{\rm softmax}_c\left[a_c+\dot a_c(\theta_1-\theta_1^0)\right],
\label{eq:ai_alpha_softmax}
\\
\beta_\phi(b|c,\theta_2)
&=&
{\rm softmax}_b\left[d_{cb}+\dot d_{cb}(\theta_2-\theta_2^0)\right].
\label{eq:ai_beta_softmax}
\end{eqnarray}
All logits and logit-derivatives are trainable variables. At the expansion point, let $\alpha_c=\alpha(c|\theta_1^0)$, $\dot\alpha_c=\partial_{\theta_1}\alpha(c|\theta_1)|_{\theta_1^0}$, $\beta_{bc}=\beta(b|c,\theta_2^0)$, and $\dot\beta_{bc}=\partial_{\theta_2}\beta(b|c,\theta_2)|_{\theta_2^0}$. The local module FIs are
\begin{eqnarray}
F_{ac}
=
\sum_c\frac{\dot\alpha_c^2}{\alpha_c},
\quad
F_{cb}
=
\sum_c\alpha_c
\sum_b\frac{\dot\beta_{bc}^2}{\beta_{bc}}.
\label{eq:ai_module_FIs}
\end{eqnarray}
The endpoint marginal and its derivatives are
\begin{eqnarray}
p_b=\sum_c\alpha_c\beta_{bc},
\quad
\partial_1 p_b=\sum_c\dot\alpha_c\beta_{bc},
\quad
\partial_2 p_b=\sum_c\alpha_c\dot\beta_{bc}.
\label{eq:ai_endpoint_derivatives}
\end{eqnarray}
Thus the endpoint Fisher matrix is
\begin{eqnarray}
[F_B]_{ij}
=
\sum_b\frac{(\partial_i p_b)(\partial_j p_b)}{p_b},
\quad i,j\in\{1,2\}.
\label{eq:ai_endpoint_FIM}
\end{eqnarray}
The effective endpoint FI for the sum direction $u=(1,1)^T$ is
\begin{eqnarray}
F_B^{(u)}
=
\left(u^T F_B^{-1}u\right)^{-1},
\label{eq:ai_Feff_B}
\end{eqnarray}
with the Moore--Penrose inverse used in singular cases. The adversary attempts to maximize
\begin{eqnarray}
\Gamma_{\rm adv}
=
\frac{F_B^{(u)}}
{\left(F_{ac}^{-1}+F_{cb}^{-1}\right)^{-1}}.
\label{eq:ai_Gamma_adv}
\end{eqnarray}
The causal-path CFII proves that every modular adversary must satisfy $\Gamma_{\rm adv}\leq1$. Numerically, we set $L=M=5$ and ran 36 independent Adam optimizations~\cite{KingmaBa2015}. The largest value obtained was
\begin{eqnarray}
\max_{\rm restarts}\Gamma_{\rm adv}=0.9999999998,
\label{eq:ai_adv_max}
\end{eqnarray}
with mean $0.99988$ and minimum $0.99869$. Thus the adversary learns to saturate the causal frontier, but it does not enter the forbidden region where the noisy coherent data lie.

\subsection{Numerical results and interpretation}

\begin{figure}[t]
\centering
\includegraphics[width=0.80\textwidth]{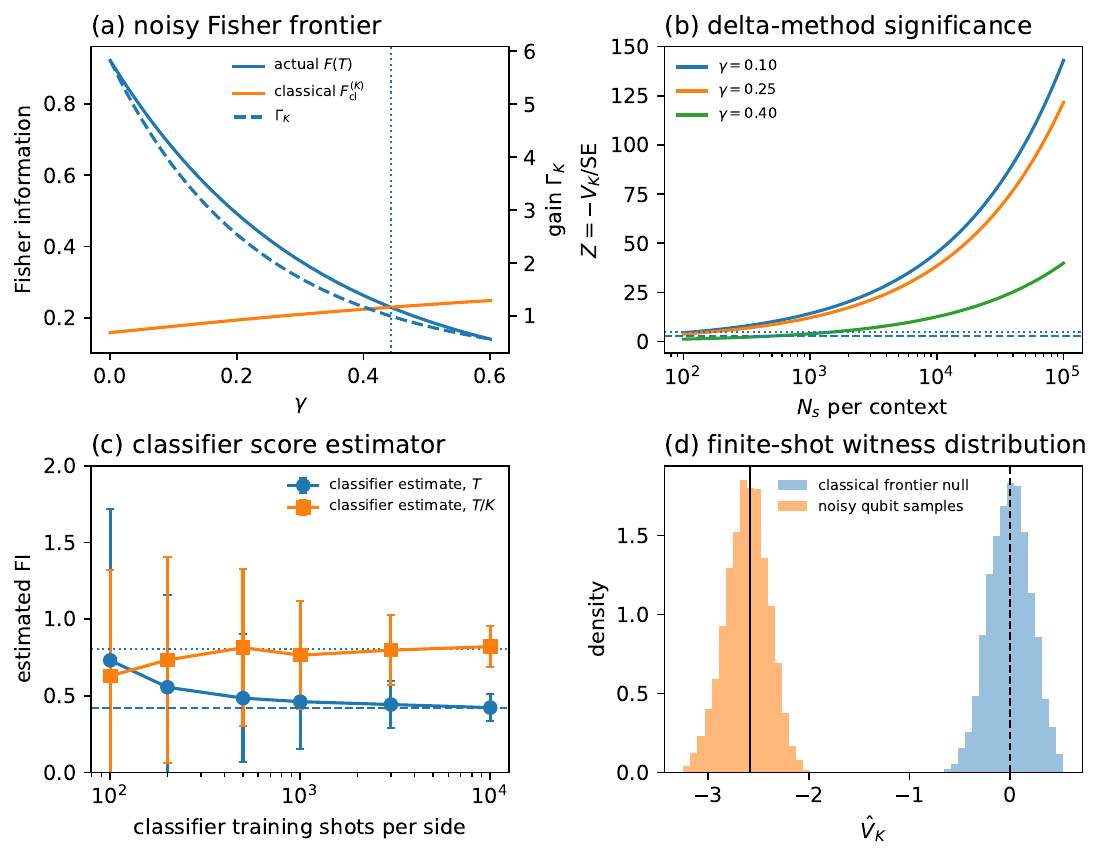}
\caption{Details of the AI-assisted finite-data stress test. (a) Actual noisy endpoint FI $F(T)$, equal-segment classical benchmark $F_{\rm cl}^{(K)}=F(T/K)/K$, and gain ratio $\Gamma_K$ as functions of dephasing rate $\gamma$. The vertical dotted line marks the $\Gamma_K=1$ crossing. (b) Delta-method significance $Z=-V_K/{\rm SE}$ versus shots per context for representative dephasing rates; horizontal lines mark $Z=3$ and $Z=5$. (c) Classifier score-estimator calibration at $\gamma=0.25$. Dashed and dotted horizontal lines indicate the analytic FIs for $T$ and $T/K$, respectively. (d) Finite-shot distribution of $\widehat V_K$ for the noisy coherent data at $\gamma=0.25$ and $N_s=10^3$, compared with a classical frontier null centered at zero.}
\label{fig:ai_stress_details}
\end{figure}

Fig.~\ref{fig:ai_stress_details} gives the detailed numerical checks supporting the main-text stress-test figure. Panel (a) shows that increasing $\gamma$ reduces the endpoint FI faster than the segment FI, eventually closing the advantage. Nevertheless, a finite noise window remains where $F(T)>F(T/K)/K$. Panel (b) shows that this region can be certified with modest sample sizes. At $\gamma=0.40$, the violation is weaker but becomes a multi-sigma effect once $N_s$ is increased to the $10^3$--$10^4$ range. Panel (c) verifies that the classifier score estimator converges to the analytic FI values as its training sample size grows. Panel (d) displays the finite-shot separation between the noisy coherent witness and the classical frontier null.

The stress test should be read in the following way. The noisy coherent data are not idealized: visibility loss, readout error, finite samples, and likelihood-free score estimation are included. The classical comparator is not weak: it is a high-capacity differentiable modular path with latent mediators and optimized local kernels. Yet the two objects occupy different sides of the same CFII frontier. This is precisely the desired robustness statement for a PRL-scale claim: the negative witness is not an artifact of an analytic toy model, nor is it rescued by a flexible hidden-variable path as long as the defining modular causal assumptions remain in force.

\suppnote{Discussion, outlook, and conclusion}\label{sec:discussion_outlook_conclusion}

This work began from a deceptively simple observation: Fisher information is not merely a figure of merit for parameter estimation. It is also a geometric object encoding how probability models deform under infinitesimal parameter changes. Once a probability model is constrained by a classical causal structure, its admissible deformations are constrained as well. The resulting constraints are Fisher-information inequalities of a fundamentally different kind: they are causal Fisher-information inequalities (CFIIs). Their violation is therefore not a numerical surprise, but a logical collapse in the underlying classical causal narrative.

Figure~\ref{fig:cfii_schematic} is intended as a compressed map of the entire witness-to-resource program developed in this SI: a classical causal assumption fixes a precision frontier; a CFII violation certifies that the assumed modular causal mediation is impossible; and the same negative witness directly quantifies a metrological gain. The schematic also makes explicit why the inverse Fisher information is the natural quantity to read as an information resistance, because it is the object that adds along classical causal paths and whose failure to add signals the breakdown of classical mediation.

\begin{figure}[t]
\centering
\includegraphics[width=0.95\textwidth]{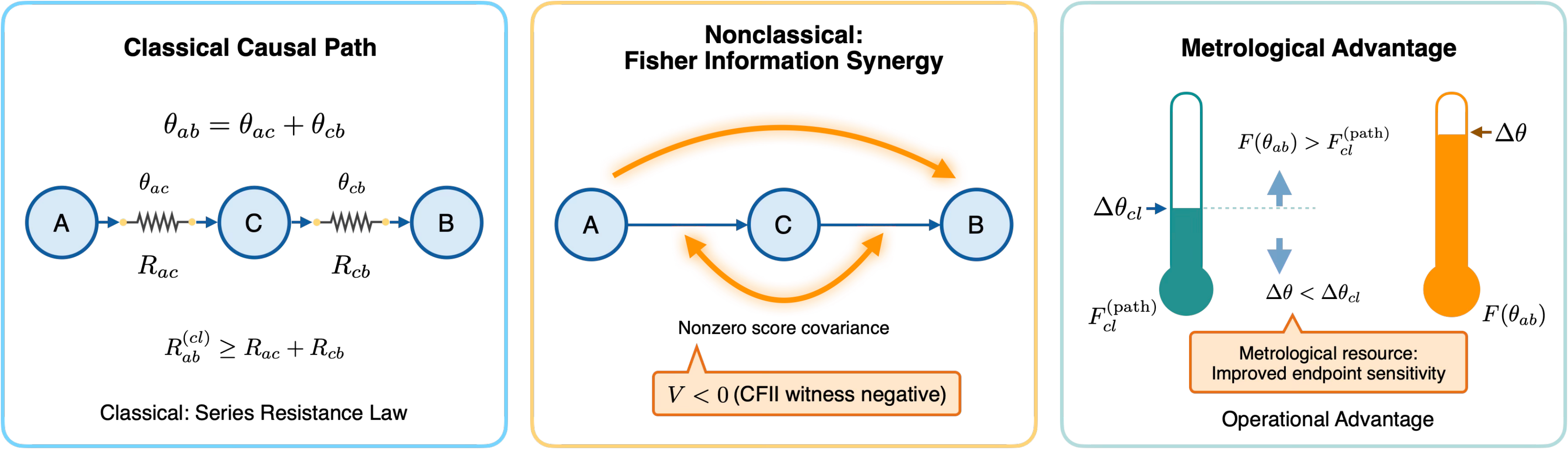}
\caption{Schematic summary of the logical flow of the paper. Left: under a classical causal-path description, the total parameter $\theta_{ab}=\theta_{ac}+\theta_{cb}$ must be mediated through an intermediate classical node $C$, so inverse Fisher information plays the role of an information resistance $R:=F^{-1}$ and obeys the series law $R_{ab}^{\rm(cl)}\ge R_{ac}+R_{cb}$. Center: coherent or otherwise nonclassical dynamics generate Fisher-information synergy (nonzero score covariance between the would-be modules), so the modular classical decomposition fails and the CFII witness becomes negative, $V<0$. Right: the same violation is operationally a metrological resource, because it implies $F(\theta_{ab})>F_{\rm cl}^{\rm(path)}$ and therefore $\Delta\theta<\Delta\theta_{\rm cl}$. For longer decompositions, the classical frontier suffers the $1/K$ series dilution, whereas engineered synergy can maintain an $O(1)$ endpoint sensitivity.}
\label{fig:cfii_schematic}
\end{figure}

\subsection{Discussions and remarks}

\subsubsection*{From Fisher-information inequalities to causal-model criteria} 

The central conceptual step of this work is to elevate Fisher-information inequalities from a metrological statement to a causal statement. In Sec.~\ref{sec:framework}--\ref{sec:cfii}, we showed how a classical causal model class (specified by a Bayesian network, conditional independences, and modular parameterizations) implies the nontrivial inequalities among Fisher informations associated with the different parameter segments or experimental conditions. In Sec.~\ref{sec:falsification}--\ref{sec:resource}, we established a sharpened interpretation: A CFII is a necessary condition for the existence of the assumed classical causal model. Therefore, if the CFII is violated, the appropriate conclusion is ``the assumed classical causal model class is impossible for the observed statistics.'' This is what we call causal-model collapse.

At the same time, the Fisher-information form of these constraints makes the collapse immediately operational. A Fisher information inequality is an inequality about precision frontiers. Thus, when a classical causal model collapses, it collapses in a way that can be quantified as a precision gap: there exists a benchmark error floor enforced by the classical model class, and the experiment achieves a strictly smaller error. This witness-resource equivalence is formalized in {\bf Theorem~\ref{thm:sec6_witness_resource_equiv}} and operationalized in the adversarial-benchmark example in Sec.~\ref{subsec:sec6_numerics}.

\subsubsection*{The physical meaning of CFII violation: breaking the series law of classical mediation} 

The causal-path CFII, expressed in the ``series'' form, i.e., as highlighted in Fig.~\ref{fig:cfii_schematic},
\begin{eqnarray}
F(\theta_{ab})^{-1}\ge F(\theta_{ac})^{-1}+F(\theta_{cb})^{-1}, \quad (\theta_{ab}=\theta_{ac}+\theta_{cb}),
\label{eq:disc_series_law}
\end{eqnarray}
is the prototypical example of the general framework. Its physical meaning is vivid. A classical causal-path explanation asserts that information about the total parameter $\theta_{ab}$ must traverse an intermediate classical mediator $C$. That assertion induces a modular decomposition into two independent estimation segments. In such a modularization, $F^{-1}$ behaves as an information resistance, and the resistances in series must add. The inequality Eq.~(\ref{eq:disc_series_law}) is therefore the unavoidable accounting rule of classical mediation.

The example in Sec.~\ref{sec:examples_coherence_entanglement} shows why coherent quantum dynamics violates this rule. In the single-qubit coherent setting, the parameter is imprinted in a way that is not compatible with a classical intermediate mediator that would render the two segments independent in the sense required by the model. When the experiment violates Eq.~(\ref{eq:disc_series_law}), it is not that the data misbehave; it is that the classical causal-path story cannot be sustained without contradicting the observed local information geometry. In this sense, the CFII violation can be a diagnostic of causal nonclassicality: the incompatibility is not merely with a classical state space, but with a classical causal modularization.

\subsubsection*{Relation to earlier trajectory-testing Fisher-information inequalities} 

Historically, Fisher-information inequalities of the form in Eq.~(\ref{eq:disc_series_law}) were motivated as the tests of discrete evolution trajectories and intermediate-state hypotheses. That viewpoint already contained the core intuition: inserting the intermediate states amounts to imposing a conditional-independence-like structure on the statistical description of dynamics. Our framework clarifies and generalizes this intuition in two decisive ways.

First, it disentangles the logical content of the inequality from the interpretational language of ``trajectory.'' The inequality is not intrinsically about whether the quantum state ``really passes through'' intermediate states. It is about whether the observed family of probability distributions is compatible with a classical causal model class whose defining feature is a modular, conditionally independent causal mediation. The trajectories constitute one special narrative within that class, but the causal logic is more general.

Second, the framework is systematic. Rather than proposing a single inequality for a single hypothesized structure, we provide a general route from a chosen causal model class (DAG plus independence/modularity assumptions) to a family of testable Fisher-information inequalities. This shifts the task from ad hoc witness construction to principled causal-model engineering.

\subsubsection*{Practical implications: designing metrology by causal-model engineering} 

A distinctive practical implication of the present framework is that it suggests a design principle for quantum sensing that is not limited to ``use entanglement'' or ``optimize measurements.'' One may instead proceed as follows:
\begin{itemize}
\item First, identify a classical causal model class that captures the desired operational notion of classicality in the target platform (for instance, a trajectory-like mediation, a context-free constraint, or a Markovian modular decomposition).
\item Second, derive the corresponding CFII constraints.
\item Third, design probe states, dynamics, and readout schemes that maximize the violation of those constraints while respecting experimental limitations.
\item Finally, turn the violation into a resource by explicitly constructing estimators and benchmarks, as done in Sec.~\ref{sec:examples_coherence_entanglement} via the $V$--$G$ equivalence and the adversarial split-optimised benchmark.
\end{itemize}
The crucial point is that this design principle is model-relative and therefore operationally meaningful: it does not claim an absolute advantage over all conceivable strategies, but a certified advantage over an entire classical causal narrative.

\subsubsection*{Limitations and methodological remarks} 

Several caveats are also worth emphasizing.

\begin{itemize}
\item  First, Fisher information is a local quantity. Our inequalities are therefore statements about local estimation geometry around a parameter point, and their operational meaning is clearest in asymptotic, regular regimes where maximum-likelihood estimators saturate Cram\'er--Rao bounds.
\item Second, a failure to observe a violation does not prove the classical causal model. It only establishes consistency. This is a fundamental asymmetry of model falsification.
\item  Third, finite-sample implementations require careful statistical treatment. Estimating Fisher information from data introduces bias and variance, and one must propagate uncertainty to obtain confidence intervals for CFII violation. The Monte-Carlo achievability demonstrations in Sec.~\ref{sec:examples_coherence_entanglement} are a simple operational check, and developing a full hypothesis-testing toolkit for CFII violation is an important direction.
\item Finally, the causal assumptions are explicit and therefore falsifiable, but also context-dependent. The physical meaning of a causal-model collapse depends on the chosen model class. This is a feature rather than a bug: it forces clarity about what is meant by ``classical explanation'' in each platform.
\end{itemize}

\subsection{Outlook}

The present work opens several concrete extensions.

A first direction is algorithmic. Given a DAG with latent variables and a specified parameter modularization, can one systematically enumerate all independent CFIIs, analogously to how entropic inequalities can be algorithmically generated for graphical models. Such a program would turn CFII derivation into a push-button tool for causal metrology.

A second direction is multiparameter estimation. The Fisher information matrix carries richer geometric content than the scalar Fisher information. Generalizing CFIIs to matrix inequalities may reveal new causal constraints and new resource regimes, especially in scenarios where the tradeoffs among parameters are central.

A third direction concerns noise and open dynamics. Realistic sensors operate under decoherence, and the relevant ``classical'' causal narratives may involve additional latent environment nodes or non-Markovian memory. Embedding such effects into causal graphs and deriving corresponding CFIIs could provide a new route to diagnosing non-Markovianity and to harvesting memory effects as metrological resources.

A fourth direction is to interface with broader causal notions in quantum theory. The classical causal models used here are intentionally conservative; one may ask how CFII-like statements extend to quantum causal models, process matrices, or generalized intervention frameworks~\cite{Allen2017QuantumCommonCauses,Barrett2021CyclicQuantumCausal,Rozema2024IndefiniteCausalOrder}. This would sharpen the boundary between classical and quantum causality in information-geometric terms.

\subsection{Conclusion}

We have presented a framework in which Fisher-information inequalities become causal-model criteria. For a chosen classical causal model class, the framework yields testable causal Fisher-information inequalities (CFIIs). Their violation has a dual meaning: it falsifies the classical causal model class and simultaneously certifies a metrological resource, in the concrete sense of a precision gap relative to the best performance compatible with that class.

The examples show that this causal reading is not abstract. Even minimal coherent dynamics can cross the classical causal-path frontier, and the crossing can be made robust against adversarial classical benchmarks. The coQM scheme, discussed in the appendix, fits naturally into this narrative: NSIT is a classical causal constraint, its violation is a causal-model collapse, and the collapsed regime can be harvested as a Fisher-information resource.

As summarized schematically in Fig.~\ref{fig:cfii_schematic}, the framework yields a single causal-to-operational pipeline:
\begin{eqnarray}
\mbox{classical causal hypothesis} \;\Longrightarrow\; \mbox{CFII frontier} \;\Longrightarrow\; \mbox{violation / falsification} \;\Longrightarrow\; \mbox{certified metrological gain}.
\end{eqnarray}
In short, CFIIs unify three themes that are often treated separately: causal inference, nonclassicality, and precision measurement. They provide a principled route to reinterpreting the nonclassical advantage as the operational signature of classical causal-model impossibility, and to designing metrological protocols that deliberately exploit that impossibility.


\appendix


\suppappendix{Contextual quantum metrology as an NSIT-based causal-model collapse and its scope}\label{app:coqm}

The main text develops a general principle: Fisher-information (FI) constraints become causal constraints once one commits to a classical causal model class, and the violation of such constraints simultaneously (i) falsifies that model class and (ii) certifies a metrological resource. The contextual quantum metrology (coQM) scheme, recently studied in Ref.~\cite{Jae2024CoQM}, provides a complementary success story in which the ``classicality'' assumption is not the existence of an intermediate causal bottleneck, but a context-free causal description encoded by the condition of no-signaling in time (NSIT). This appendix makes the connection, in a form compatible with the causal language of our work.

The discussion proceeds in the following steps. At first, we introduce the NSIT as a conditional-independence (CI) constraint induced by a classical causal model. We then show that coQM reduces to conventional metrology when NSIT holds, hence no enhancement is possible in the context-free regime. Finally, we prove that an NSIT-based framework, by itself, cannot subsume the general CFII program developed in our work.

\subsection{Preliminaries: two contexts and the NSIT constraint}\label{app:coqm_setting}

We consider a parameter estimation problem with a single real parameter $\theta$. A probe state $\hat{\rho}(\theta)$ is prepared, and the experimenter can choose between two ``contexts'' as below:
\begin{eqnarray}
S=0: && \mbox{perform measurement $\hat{B}$ only}, \nonumber \\[2pt]
S=1: && \mbox{perform measurement $\hat{A}$ and then $\hat{B}$}.
\label{eq:app_contexts}
\end{eqnarray}
Let $b \in {\cal B}$ denote the outcome of $\hat{B}$, and let $a \in {\cal A}$ denote the outcome of $\hat{A}$. The observable data are:
\begin{eqnarray}
p(b | S=0,\theta) \equiv p(b | B,\theta),
\end{eqnarray}
and
\begin{eqnarray}
p(a,b | S=1,\theta) \equiv p(a,b | A,B,\theta), \quad p(b | S=1,\theta) = \sum_a p(a,b | S=1,\theta).
\label{eq:app_pAB_def}
\end{eqnarray}

The NSIT condition~\cite{Halliwell2017MacrorealismNSIT} is the requirement that the marginal statistics of $B$ are invariant under the context of whether $A$ is performed beforehand:
\begin{eqnarray}
p(b | S=0,\theta) = p(b | S=1,\theta) \quad (\forall \, b,\theta.)
\label{eq:app_nsit_def}
\end{eqnarray}
Note that Eq.~(\ref{eq:app_nsit_def}) is a constraint directly testable from the measured frequencies. Conceptually, it expresses the absence of a causal influence of the ``selection of $A$'' on the marginal prediction of $B$.

\subsection{NSIT as a classical causal-model constraint}\label{app:coqm_nsit_causal}

We now formulate NSIT as a CI constraint, and thereby, as a causal-model class in the sense of our work.
\begin{proposition}[NSIT is a conditional-independence constraint]
\label{prop:app_nsit_CI}
Fix $\theta$ and consider the joint distribution of $(S,B)$ defined by the operational conditional distributions $p(b | S,\theta)$ and an exogenous distribution $p(s)$ of the context choice. Then, the NSIT condition in Eq.~(\ref{eq:app_nsit_def}) holds for all $b$ if and only if
\begin{eqnarray}
B \perp S \,|\, \theta.
\label{eq:app_CI_nsit}
\end{eqnarray}
\end{proposition}

\begin{proof}---Assume NSIT holds. Then $p(b | S=0,\theta) = p(b | S=1,\theta)$ for all $b$, hence $p(b | s,\theta)$ is independent of $s$. Therefore
\begin{eqnarray}
p(b,s | \theta) = p(s) p(b | s,\theta) = p(s)p(b | \theta),
\end{eqnarray}
which is exactly $B \perp S | \theta$. Conversely, if $B \perp S | \theta$, then $p(b | s,\theta) = p(b | \theta)$ for all $s$, in particular for $s=0$ and $s=1$. This is NSIT. 
\end{proof}

{\bf Proposition~\ref{prop:app_nsit_CI}} shows that NSIT is not a specifically quantum notion; it is the operational face of a CI constraint. In the causal language of Sec.~\ref{sec:framework}--\ref{sec:falsification}, it corresponds to a model class ${\cal M}_{\rm NSIT}$ whose defining feature is that the exogenous context variable $S$ has no causal arrow into $B$.

A standard latent-variable embedding makes this explicit. Let $\Lambda$ denote a latent variable (a ``classical state'' prior to measurement), distributed as $p(\lambda | \theta)$. A minimal noncontextual (context-free) causal model class ${\cal M}_{\rm NSIT}$ can be defined by the DAG constraints
\begin{eqnarray}
S \rightarrow A,
\quad
\Lambda \rightarrow A,
\quad
\Lambda \rightarrow B,
\quad
S \not\rightarrow B,
\label{eq:app_DAG_nsit}
\end{eqnarray}
together with the factorization
\begin{eqnarray}
p(a,b,\lambda|s,\theta) = p(\lambda|\theta) \, p(a | s,\lambda,\theta) \, p(b | \lambda,\theta).
\label{eq:app_factorization_nsit}
\end{eqnarray}
By marginalizing over $\lambda$, we yield
\begin{eqnarray}
p(b | s,\theta) = \sum_\lambda p(\lambda | \theta) p(b | \lambda,\theta),
\label{eq:app_marginal_nsit}
\end{eqnarray}
which is manifestly independent of $s$ and hence implies NSIT. Thus, NSIT is a necessary consequence of the classical causal story in Eq.~(\ref{eq:app_DAG_nsit}) and Eq.~(\ref{eq:app_factorization_nsit}). Within the viewpoint of our work, one may take Eq.~(\ref{eq:app_nsit_def}) itself as the defining operational signature of ${\cal M}_{\rm NSIT}$.

We can now state that:
\begin{theorem}[NSIT violation implies impossibility of the context-free causal model]
\label{thm:app_nsit_falsification}
If there exist $\theta$ and $b$ such that $p(b | S=0,\theta) \neq p(b | S=1,\theta)$, then no model in the class ${\cal M}_{\rm NSIT}$ defined by Eq.~(\ref{eq:app_factorization_nsit}) can reproduce the observed operational statistics.
\end{theorem}

\begin{proof}---Any model in ${\cal M}_{\rm NSIT}$ implies Eq.~(\ref{eq:app_marginal_nsit}), hence $p(b | s,\theta)$ must be independent of $s$ for all $(b,\theta)$. This is exactly NSIT. Therefore, if NSIT is violated by the observed statistics, those statistics cannot arise from any member of ${\cal M}_{\rm NSIT}$. For more details, see Ref.~\cite{Jae2024CoQM}.
\end{proof}

{\bf Theorem~\ref{thm:app_nsit_falsification}} is the precise sense in which NSIT violation is a causal-model collapse: it falsifies the entire class of context-free classical causal explanations that forbid the arrow $S \to B$.

\subsection{Operational quasiprobability and the reduction under NSIT}\label{app:coqm_OQ}

The central operational move in coQM is to integrate the two contexts into a single statistical model, the operational quasiprobability (OQ)~\cite{Ryu2013OQudits,Jae2017OQCV,Jae2024CoQM}. The OQ is constructed as
\begin{eqnarray}
w(a,b|\theta) = p(a,b | S=1,\theta) + \frac{1}{2}\bigl[ p(b | S=0,\theta) - p(b | S=1,\theta) \bigr].
\label{eq:app_w_def2}
\end{eqnarray}
Note that OQ is constructed purely from observable probabilities. When NSIT holds, the correction term vanishes and OQ collapses to an ordinary joint probability.

We then have the following proposition:
\begin{proposition}[Reduction of OQ under NSIT]
\label{prop:app_w_reduction}
If NSIT in Eq.~(\ref{eq:app_nsit_def}) holds, then
\begin{eqnarray}
w(a,b | \theta) = p(a,b | S=1,\theta) \quad \forall \, a,b,\theta.
\label{eq:app_w_equals_pab}
\end{eqnarray}
Conversely, if Eq.~(\ref{eq:app_w_equals_pab}) holds for all $a,b,\theta$, then NSIT holds.
\end{proposition}

\begin{proof}
If NSIT holds, then $p(b | S=0,\theta) - p(b | S=1,\theta)=0$ for all $b,\theta$, and Eq.~(\ref{eq:app_w_def2}) reduces to Eq.~(\ref{eq:app_w_equals_pab}). Conversely, summing Eq.~(\ref{eq:app_w_equals_pab}) over $a$ yields
\begin{eqnarray}
\sum_a w(a,b | \theta) = p(b | S=1,\theta).
\end{eqnarray}
But summing Eq.~(\ref{eq:app_w_def2}) over $a$ gives
\begin{eqnarray}
\sum_a w(a,b | \theta) = p(b | S=1,\theta) +\frac{1}{2}\bigl[ p(b | S=0,\theta) - p(b | S=1,\theta) \bigr].
\end{eqnarray}
The two expressions imply $p(b | S=0,\theta) = p(b | S=1,\theta)$ for all $b,\theta$, i.e., NSIT.
\end{proof}

\subsection{From witness to resource: Fisher-information consequences}\label{app:coqm_resource}

{\bf Proposition~\ref{prop:app_w_reduction}} already exposes the metrological logic. If NSIT holds, OQ does not introduce a new statistical structure; it merely reproduces the probability of the consecutive measurement. Thus, in the context-free causal regime, coQM cannot create an advantage. To make this statement quantitative, we now relate the OQ to Fisher information.

In Ref.~\cite{Jae2024CoQM}, coQM defines the contextual Fisher information (coFI) by treating $w(a,b | \theta)$ as a statistical model in the parameter regime where $w(a,b | \theta)>0$:
\begin{eqnarray}
F_{\rm co}(\theta) = \sum_{a,b} w(a,b|\theta)\left(\frac{\partial}{\partial\theta}\ln w(a,b|\theta)\right)^2.
\label{eq:app_Fco_def2}
\end{eqnarray}
For comparison, let $F_q(\theta)$ denote the quantum Fisher information (QFI) of the probe family $\hat{\rho}(\theta)$, defined as the supremum of the classical FI over all POVMs on the probe.

Then we have the following theorem:
\begin{theorem}[No advantage in the NSIT (context-free) regime]
\label{thm:app_no_adv_nsit}
Assume NSIT holds and $w(a,b | \theta) > 0$. Then,
\begin{eqnarray}
F_{\rm co}(\theta) = F_{AB}(\theta) \le F_q(\theta),
\label{eq:app_Fco_le_Fq}
\end{eqnarray}
where $F_{AB}(\theta)$ is the classical FI of the consecutive measurement statistics $p(a,b | S=1,\theta)$. Consequently, when coQM is compared to a conventional strategy using the same total number of samples, no precision enhancement beyond the conventional QFI benchmark is possible in the NSIT regime.
\end{theorem}

\begin{proof}---By {\bf Proposition~\ref{prop:app_w_reduction}}, NSIT implies $w(a,b|\theta)=p(a,b|S=1,\theta)$. By substituting this into Eq.~(\ref{eq:app_Fco_def2}), we obtain $F_{\rm co}(\theta)=F_{AB}(\theta)$. The inequality $F_{AB}(\theta) \le F_q(\theta)$ follows from the definition of QFI as the maximum FI achievable by any measurement on $\hat{\rho}(\theta)$. The consecutive measurement defines a valid measurement scheme on the probe and is therefore included in the maximization defining $F_q(\theta)$. Thus, Eq.~(\ref{eq:app_Fco_le_Fq}) holds. Finally, if a total of $2N_s$ probe uses are available, any conventional strategy satisfies the Cram\'er--Rao benchmark $\Delta\theta \ge 1/\sqrt{2N_s F_q(\theta)}$. In the NSIT regime, $F_{\rm co} \le F_q$ and the effective model underlying coQM reduces to the consecutive-measurement model; hence coQM cannot produce an error scaling that beats the conventional benchmark.
\end{proof}

{\bf Theorem~\ref{thm:app_no_adv_nsit}} is the rigorous content behind the informal statement ``if the causal model does not collapse, there is no resource to harvest.'' In the NSIT regime, the OQ carries no additional information beyond the ordinary probability model, and the Fisher information cannot exceed the QFI benchmark.

The decisive point is therefore the violation of NSIT. When $p(b | S=0,\theta) \neq p(b|S=1,\theta)$, the correction term in Eq.~(\ref{eq:app_w_def2}) becomes nonzero. Then, $w$ becomes a new effective statistical model that integrates two contexts. This is precisely where coQM can exhibit a precision enhancement.

\subsection{NSIT does not subsume the CFII program}\label{app:coqm_scope_limit}

The NSIT condition is an extremely sharp and operationally meaningful constraint, but it constrains only one specific kind of conditional independence, namely the absence of a causal arrow from a context choice into a marginal prediction. In the coQM/NSIT setting, one introduces a context variable $S$ (e.g., whether a preliminary measurement $A$ is selected before measuring $B$) and imposes the constraint
\begin{eqnarray}
p(b | S=0,\theta) = p(b | S=1,\theta)\quad(\forall \, b,\theta),
\label{eq:app_nsit_repeat}
\end{eqnarray}
which is nothing but the CI relation $B \perp S | \theta$ in a particular DAG. In other words, NSIT decides whether the marginal statistics of $B$ are invariant under the context choice $S$.

Our CFII program addresses a different question. Given an arbitrary Bayesian network (DAG, possibly with latent variables), one extracts its CI relations and then derives Fisher-information inequalities implied by the entire causal structure. The causal-path CFII in our work is not a constraint about context dependence at fixed $\theta$, but a constraint about modularity and additivity along a causal chain. In particular, the causal-path CFII
\begin{eqnarray}
F(\theta_{ab})^{-1} \ge F(\theta_{ac})^{-1} + F(\theta_{cb})^{-1}, \quad (\theta_{ab}=\theta_{ac}+\theta_{cb}),
\label{eq:app_path_CFII_repeat}
\end{eqnarray}
encodes the series law of information that follows when a causal-path model enforces (i) a modular decomposition into two independent estimation modules and (ii) an additive parameter split. NSIT, by itself, does not impose either (i) or (ii), and therefore cannot be expected to reproduce Eq.~(\ref{eq:app_path_CFII_repeat}).

The separation can be made completely explicit. We now exhibit an operational model that is perfectly NSIT (hence fully compatible with the NSIT-based causal class ${\cal M}_{\rm NSIT}$), yet violates the causal-path CFII everywhere once one attempts to apply the causal-path modularity logic.
\begin{proposition}[NSIT-compatible models can be maximally CFII-violating]
\label{prop:app_nsit_not_generalize_strong}
There exist operational models for which the NSIT condition in Eq.~(\ref{eq:app_nsit_repeat}) holds identically for all $\theta$ (hence the model is compatible with ${\cal M}_{\rm NSIT}$), but for which the causal-path CFII in Eq.~(\ref{eq:app_path_CFII_repeat}) is violated for \emph{every} nontrivial split $\theta_{ab}=\theta_{ac}+\theta_{cb}$ with $\theta_{ac}>0$ and $\theta_{cb}>0$. Therefore, an NSIT-only framework cannot subsume the causal-path CFII framework.
\end{proposition}

\begin{proof}---Consider the binary one-parameter family
\begin{eqnarray}
p(0 | \theta)=\cos^2\left(\tfrac{\theta}{2}\right), \quad p(1 | \theta)=\sin^2\left(\tfrac{\theta}{2}\right).
\label{eq:app_binary_family}
\end{eqnarray}
This family is regular and has constant (classical) Fisher information. Since $\partial_\theta p(0 | \theta) = -\frac{1}{2}\sin\theta$ and $p(0|\theta)\bigl(1-p(0 | \theta)\bigr)=\tfrac{1}{4}\sin^2\theta$, the Fisher information
\begin{eqnarray}
F(\theta) = \sum_{b=0,1} p(b|\theta)\left(\frac{\partial}{\partial\theta}\log p(b|\theta)\right)^2 = \frac{\bigl(\partial_\theta p(0|\theta)\bigr)^2}{p(0|\theta)\bigl(1-p(0|\theta)\bigr)} = 1
\quad(\forall\,\theta),
\label{eq:app_FI_const_1}
\end{eqnarray}
where the apparent $0/0$ points at $\theta=0,\pi,2\pi,\ldots$ are removable and the continuous extension yields the same constant.

\medskip\noindent
{\bf Step 1: enforce NSIT.} Introduce a context label $S \in \{0,1\}$ and define the marginal model of $B$ to be context-independent:
\begin{eqnarray}
p(b | S=0,\theta) = p(b | S=1,\theta) = p(b|\theta) \quad (\forall \, b,\theta),
\label{eq:app_nsit_extension}
\end{eqnarray}
with $p(b | \theta)$ given by Eq.~(\ref{eq:app_binary_family}). Then, Eq.~(\ref{eq:app_nsit_repeat}) holds identically, i.e., the model is perfectly NSIT, and therefore, compatible with ${\cal M}_{\rm NSIT}$. Operationally, one may realize such an NSIT situation, for example, by letting the ``optional'' measurement ${A}$ commute with (or even coincide with) ${B}$ and ensuring that no parameter encoding occurs between them; then selecting ${A}$ does not disturb the marginal prediction of ${B}$.

\medskip\noindent
{\bf Step 2: apply the causal-path CFII to an additive split.} Now consider an additive decomposition $\theta_{ab}=\theta_{ac}+\theta_{cb}$ with $\theta_{ac}, \theta_{cb}>0$ and evaluate the CFII gap
\begin{eqnarray}
V(\theta_{ac},\theta_{cb}) := F(\theta_{ab})^{-1} - F(\theta_{ac})^{-1} - F(\theta_{cb})^{-1}.
\label{eq:app_V_def}
\end{eqnarray}
Using Eq.~(\ref{eq:app_FI_const_1}), we obtain
\begin{eqnarray}
V(\theta_{ac},\theta_{cb}) = 1-1-1 = -1 < 0,
\label{eq:app_V_minus1}
\end{eqnarray}
for every nontrivial split. Hence, the causal-path CFII in Eq.~(\ref{eq:app_path_CFII_repeat}) is violated everywhere. This proves the claim.
\end{proof}

{\bf Proposition~\ref{prop:app_nsit_not_generalize_strong}} sharpens the scope statement in a causal-model language. NSIT constrains a single arrow--prohibition ($S \not\to B$) and is silent about whether a process admits a modular causal mediation $A \to C \to B$ with an additive parameter split. The causal-path CFII, by contrast, is precisely the quantitative footprint of that modular mediation: it is derived from the path factorization and from the causal data-processing law (Sec.~\ref{sec:cfii}). Therefore, an NSIT-based viewpoint cannot be a universal ``generalization principle'' for the CFII program.

This also clarifies the narrative hierarchy. coQM (and NSIT) fits naturally as one special instance of our broader causal perspective: NSIT is one CI constraint, its violation is one type of causal-model collapse, and coQM turns that collapse into a Fisher-information gain, exactly in the sense of the witness-to-resource transition of Sec.~\ref{sec:resource}. Our CFII framework, however, is strictly broader: it provides a systematic route from arbitrary classical causal model classes (DAGs + modularity) to testable Fisher-information inequalities and to resource interpretations of their violations.


%

\end{document}